\documentclass{article}
\usepackage{fullpage}
\usepackage[utf8]{inputenc}
\usepackage[usenames,dvipsnames]{xcolor}
\usepackage{amsmath,amsthm,amsfonts,amssymb,enumerate,makecell}
\usepackage{tikz}
\usetikzlibrary{matrix,arrows,decorations.pathmorphing,decorations.markings}
\usepackage{graphicx}
\usepackage{subcaption}
\numberwithin{equation}{section}
\usepackage{amssymb,mathrsfs}
\usepackage{bbm}
\usepackage{comment}
\usepackage{amsthm}
\usepackage[all,cmtip]{xy}
\usepackage[toc,page]{appendix}
\usepackage[british]{babel}
\usepackage{courier}
\usepackage{enumerate}
\usepackage[T1]{fontenc}
\usepackage{listings}
\usepackage{algorithm}
\usepackage[noend]{algpseudocode}
\theoremstyle{plain}
\usepackage[version=4]{mhchem}
\usepackage{chemfig}
\usepackage{etex}
\usepackage{authblk}

\newcommand{\lee}[1]{\textit{\textbf{{\color{red}#1}}}}

\newcommand{\ingrid}[1]{\textit{\textbf{{\color{violet}#1}}}}

\newcommand{\mobius}{M\"obius }
\newtheorem{theorem}{Theorem}[section]
\newtheorem*{theorem*}{Theorem}

\newtheorem{example}[theorem]{Example}

\newtheorem{proposition}[theorem]{Proposition}

\theoremstyle{remark}
\newtheorem{remark}[theorem]{Remark}

\newcommand{\Sym}{\operatorname{Sym}}

\newtheorem{De}[theorem]{Definition}

\title{Topology and geometry of molecular conformational spaces\\ and energy landscapes}

\author[1]{Ingrid Membrillo-Solis}
\author[1]{Mariam Pirashvili}
\author[2]{Lee Steinberg}
\author[1]{Jacek Brodzki}
\author[2]{Jeremy  G. Frey}

\affil[1]{School of Mathematical Sciences, University of Southampton}
\affil[2]{School of Chemistry, University of Southampton}
\date{July 2019}

\begin{document}

\maketitle

%NB: The commands \verb-\lee{}-, \verb-\mariam{}-,  \verb-\ingrid{}- and \verb-\jacek{}- make the text look like \lee{this}, \mariam{this}, \ingrid{this} and \jacek{this}.

\begin{abstract}
Understanding the geometry and topology of configuration or conformational spaces of molecules has relevant applications in chemistry and biology such as the  proteins folding problem, drug design and the structure activity relationship problem. Despite their relevance, configuration spaces of molecules are only partially understood. In this paper we discuss both theoretical and computational approaches to the configuration spaces of molecules and their associated energy landscapes. Our mathematical approach shows that when symmetries of the molecules are taken into account, configuration spaces of molecules give rise to certain principal bundles and orbifolds. We also make use of a variety of geometric and topological tools for data analysis to study the topology and geometry of these spaces. 
\end{abstract}
\begin{comment}
\lee{Can we write here (in bullet points) the agreed upon structure? Then I'll go through deleting/moving/restructuring my stuff.}

\ingrid{Here is what we agreed:
\begin{enumerate}
    \item Introduction
    \begin{itemize}
        \item Literature review on conformational spaces, why it is important study this spaces 
        \item Methods in topological and geometric data analysis, why these methods might improve the understanding of conformational space 
    \end{itemize}
    \item Materials and methods
    \begin{itemize}
        \item Conformers generation
        \item Potential energy and free energy surfaces
        \item Metrics: Euclidean, RMSD, Orbifold. 
        \item Local PCA
        \item Persistent homology 
        \item Morse-Samle complex of energy landscape
    \end{itemize}
    \item Mathematical framework
    \item Results and discussion
    \begin{itemize}
        \item Local dimension and orientability 
        \item Persistent homology of all molecules 
        \item Energy landscapes of all molecules
    \end{itemize}
    \item Conclusions
    \item Appendix: principal bundles, orbifolds, morse theory, etcetera,
\end{enumerate}}
\end{comment}

\section{Introduction}

%\ingrid{This is part of what I have written before for the introduction. You are free to modify, delete or add anything you might find convenient}

A molecule is an example of a $n$-body system formed by the nuclei and the electrons of its constituent atoms. The first step of the Born-Oppenheimer approximation allows to represent the nuclei as points in Euclidean space. %The kinetic energy of a molecule appears as the sum of three main terms, rotational, translational and vibrational energies. 
%The molecular system as a dynamical system can be regarded as an $N$-body system whose degrees of freedom come from the position of the $n$ nuclei in the molecule.
The space defined by the degrees of freedom of the molecular system after elimination of the symmetry group actions is called the \textit{internal configuration or conformational space}. The potential energy surface (PES) describes the energy of a molecule. It is a function on the internal conformational space of the molecule. %It is a function on the positions of the nuclei.  The reduction of the degrees of freedom of a molecular system is carried out taking when the groups of symmetry that act on it are considered. 
%The indistinguishability property of the particles in physical system in the quantum mechanical description seem to be in agreement with the experimental facts. It has been pointed out that the configuration space of an $N$-particle system is not the cartesian product of the single particle spaces but rather a quotient space . %The principle  of indistinguishability of identical particles did not start with the quantum mechanical description. In order to calculate the the zero-entropy change in a process of mixing two identical fluids or gases at the same tamperature Gibbs postulated that states differing only by permutations of identical particles should not be counted as distinct. The global properties of the quotient configuration space become essential in the quantum mechanical description. In \cite{LeiMyr} it was showed that the usual symmetry constrains on the wave function defined on CS can be derived as consequences of the topological differences with respect to $OCS$. 

The study of the conformational spaces of molecules and their associated PES is of particular interest in Chemistry. It can help to understand the relationship between structure and properties of molecules. For instance, PES can give insight into the structure and dynamics of macromolecules such as proteins. It has been shown that protein native-state structures can arise from considerations of symmetry and geometry associated with the polypeptide chain \cite{ATSFM}. Furthermore, understanding of the binding pose of a drug with its potential target requires knowledge of both their underlying PES, and conformational spaces \cite{Tao2001}.
% Conformer generation is regarded as a fundamental factor in drug design \cite{Confgen}. Several algorithms have been created to produce as many conformers as possible. One of the challenges is obtaining an algorithm that reduces the number of duplicate conformers. To do so, it is important to take into account the symmetries of the molecule. It has been claimed that prediction of melting points can be improved by taking molecular symmetry into account. The study of the vibrational motion of molecules is carried out through either classical or quantum mechanics. Group theory has been used by physicists and chemists to classify the state of molecules. Most of their studies focus on rigid molecules which are always close to a unique symmetric configuration as opposed to non-rigid molecules. 
Conformer generation is regarded as a fundamental factor in drug design \cite{Confgen}. Therefore, several methods exist that produce conformer sets, sampling the conformational space. One of the challenges is obtaining an algorithm that reduces the number of duplicate conformers. This could be achieved if the symmetries of a given molecule are taken into account. Also, it has been claimed that prediction of melting points can be improved by taking molecular symmetry into account \cite{Wei,Pinal}. 

In this paper we explore the geometry and topology of the conformational space of molecules and their quotients by symmetry groups. Despite the importance of the group of symmetries and its action on the space of conformers, to our knowledge there are no works on spaces of conformers that include it. Furthermore, conformational spaces are often discussed only in terms of their torsional degrees of freedom, operating under the rigid geometry hypothesis \cite{Gibson1997}. The impact of other degrees of freedom on the conformational spaces themselves is often ignored.

\subsection{Geometric and topological methods in data analysis}
Topology-based data analysis methods have seen continued interest in recent years. Persistent homology and discrete Morse theory are two topological data analysis tools, which are closely interlinked, and which we have applied to the exploration of conformation spaces of molecules.

Persistent homology, which may be more familiar in the chemistry community, is a method of assigning numerical descriptors to data, based on topological notions of shape, which emerges through a process of creating a combinatorial structure, called a simplicial complex, from the data, together with a filter function. These descriptors satisfy robust stability results with respect to the so-called bottleneck distance. The use of this method allowed us to explore the topology of the conformation spaces.

Discrete Morse theory is mathematically very closely linked to persistent homology, but it has different applications. It is used for topological simplification, for distilling the information down to the most relevant. We used it to explore the energy landscapes via their extrema.

In the paper \emph{PHoS: Persistent Homology for Virtual Screening} \cite{PHoS}, Keller et al. use multi-dimensional persistence to investigate molecules in the context of drug discovery. Their idea is using two filter functions, one of which is a scalar field defined around the molecule. This is similar to our ideas, however we focus on the conformational space, rather than individual molecules, and instead of merely calculating the persistence of the scalar function (the energy landscape, in our case), we explore it using discrete Morse theory. 

Discrete Morse theory has recently been used to reconstruct hidden graph-like structures from potentially noisy data. This has found application in vastly diverse areas. For example, Sousbie et al. used the simulated density field of dark matter to reconstruct the network of filaments in the large scale distribution of matter in the universe, the so-called cosmic web \cite{Sousbie1}. Given a collection of GPS trajectories, Dey et al. recovered the hidden road network by modelling it as reconstructing a geometric graph embedded in the plane \cite{Wang1}. Another paper by Delgado-Friedrichs et al. defines skeletons and partitions of greyscale digital images by modelling a greyscale image as a cubical complex with a real-valued function defined on its vertices and using discrete Morse theory \cite{cubical}.

A more fundamental application of discrete Morse theory in topological data analysis is topological simplification. Here, the link with persistent homology allows a topology-based denoising of data, as explored in \cite{Edelsbrunner2000} and \cite{Bauer2012}.

This methodology has already been introduced to the chemical setting as well. Gyulassi et al. \cite{porous} used Morse-theoretic approaches to investigate the properties of a simulated porous solid as it is hit by a projectile by generating distance fields containing a minimal number of topological features, and using them to identify features of the material. In \cite{Beketayev2011}, the authors construct the Morse-Smale complex of 

Our approach relies on these results, however our focus is somewhat different. We use the connection between Morse theory and persistent homology to construct a combinatorial summary of the conformation space of a given molecule, which takes into account both the topological properties of the conformation space, as well as the energy landscape defined on it.
\subsection{The workflow and outline of the paper}
The general outline of the paper is shown in Figure \ref{fig:workflow}. Starting with the computational details in the second section, we explain the conformer generation procedure used, followed by the calculations of the potential and free energy surfaces. Next, we move on to outline the geometric and topological methods used for the analysis of the conformational space.

Section 3 discusses our mathematical framework with regards to the conformational spaces arising from the molecular graphs and group actions on these spaces, as well as the metrics defined on them. This section contains our original theoretical results.

The fourth section discusses the results of our analyses performed on several benchmark molecules. These are based on the mathematical framework of the previous section, and use the aforementioned mathematical data analysis methods.

Finally, we end with our conclusions in Section 5.

In the Appendix, we explain in more detail the mathematical methods applied in the paper.
\begin{figure}
\centering
\includegraphics[width=\linewidth]{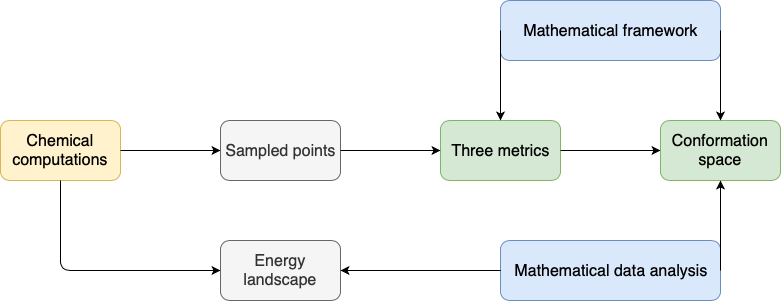}
\caption{The workflow of the paper.}
\label{fig:workflow}
\end{figure}

\section{Computational details}
%\section{Materials and methods}
\subsection{Conformer Generation Procedure}
The task of creating sets of molecular conformations is inherently complex due to the large number of degrees of freedom in a molecule. Furthermore, it is often the case that in reality what is actually desired is a set of low-energy structures, and often the ability of an algorithm or program to create these conformers is used as its quality metric \cite{Ebejer2012}. In general conformer generation procedures can be separated into knowledge-based, grid search, or distance geometry based approaches. Knowledge-based approaches use known low-energy conformers, such as crystal structures, to define rules which can then be used to generate conformers for a new molecule. Grid search approaches simply enumerate combinations of different degrees of freedom. Finally, distance geometry uses upper and lower bounds to create sets of conformations that satisfy these bounds. The reader is directed to \cite{Ebejer2012} for more information regarding these methods.

For this work, we are in general unrestricted by energy, instead using a more general \emph{physicality} criteria. This is because we would like to sample the conformational space as best as we can, rather than simply obtain representative low energy conformers. However, if energy were totally unrestricted, this would lead to conformers that we would consider unphysical, in particular caused by clashes between atoms. This can be rectified by recognising that any commonly used forcefield would give such a configuration an energy orders of magnitude higher than any other, due to the near exponential scaling of any reasonable Pauli repulsion approximation. 

RDKit \cite{Landrum2018} has been used in this work to create conformation sets, using the ETKDG method \cite{Riniker2015}. This creates a set of conformers with reasonable distributions in bond lengths and angles, but fairly fixed torsions. To ensure that the conformational space was covered, each conformer had its torsions determined from independent uniform distributions on $(-\pi,\pi]$. Lastly, conformers with energies in excess of $200$kcal/mol were removed, ensuring there were no atomic clashes.

As well as studying conformer sets with variation in all molecular degrees of freedom, we have also generated sets with only torsional variation. Firstly, conformer sets have been created through a grid search in the torsion degrees of freedom. This set also requires energy pruning to remove unphysical molecules. Secondly, a metadynamics \cite{Laio2002} approach was used. This also allowed the creation of free energy landscapes and do not require energy pruning. The drawback from this approach is that we no longer have information as to all degrees of freedom, instead reducing our space to the reaction coordinates.

Lastly, we have studied the conformational space of cyclooctane, using data obtained from \cite{Martin2010}. This is a set of 6040 conformations of cyclooctane, varying in torsion angles. Hydrogen coordinates were found through a constrained geometry optimisation. The reader is referred to \cite{Martin2010,Brown2008} for more information. Table \ref{tab:molecule_betti} contains information as to the size of the generated conformer sets.
% \subsection{Conformer Generation Procedure}

% \subsubsection{Cyclooctane}
% The set of cyclooctane conformers was obtained from the authors of `Topology of Cyclooctane Energy Landscape' \cite{Martin2010}. This resulted in a set of 6040 conformations, differing in torsion angles only. Originally, a constrained geometry optimisation was used to calculate the positions of hydrogen atoms. We refer the reader to \cite{Martin2010,Brown2008} for more information.

% \subsubsection{All other molecules}
% All other molecules had conformations generated using the distance geometry methods implemented in RDKit \cite{Landrum2018,Riniker2015}. This would often lead to unfavourable distributions in the torsion angles, so these were redrawn from uniform distributions on $(-\pi, \pi]$. Geometrically unphysical conformers were removed (for example, a minority of molecules had \chemfig{C-[:0,0.5]N-[:0,0.5]C} bond angles of $\pi$), and an energy pruning was used to remove the high-energy conformers caused by steric overlap and atom intersections.  

\subsection{Potential energy and free energy surfaces}
Operating within the Born-Oppenheimer approximation, we can consider the molecular energy to be a function of the atomic coordinates. There are many different methods for calculating this molecular energy, broadly split into those that are classical and quantum mechanical. Here, we use a classical forcefield to calculate the potential energy of a single conformer. These forcefields contain parameters describing the relative strength of bond bending, bond stretching, and torsional rotations within a molecule, broadly written as:
\begin{equation}
    E_{molecule} = \sum_{bonds} k(x-x^\prime)^2 + \sum_{angles} t(\theta - \theta^\prime)^2 + \sum_{torsions} \left( 1+\sum_n V_n \cos(n\omega)\right)
\end{equation}
Where $k$, $t$ and $V$ modulate the relative strength of interactions, and non-bonded terms have been dropped as we are not studying systems with more than one molecule in this work. We use the MMFF94 forcefield as implemented in RDKit \cite{Halgren1996} to calculate the potential energy of a single conformer.

Often, a more useful quantity to study is the free energy. The free energy can be thought of as a 'smoothed out' potential energy, where various degrees of freedom integrated out through an appropriately weighted Boltzmann average. In our work, we use metadynamics \cite{Laio2002} to calculate free energy surfaces over the torsional degrees of freedom for a molecule.

Both the potential and free energy functions can be considered as maps from a conformational space to a subset of the real line. The potential energy would be a map from the full conformational space, whereas the free energy is a map from the torsion-only subspace of the conformational space.

\subsection{Local PCA and orientability}
 %We compare the conformational spaces  $C_{\mathcal M}$ defined by the following distances: $d_{OC_{\mathcal M}}(A,B)$, $d{C_{\mathcal M}}(A,B)$ and $d_{RMSD}(A,B)$. As a transformation group $G$ we use $O(3)\times GNP_{\mathcal M}$, where $GNP_\mathcal M$ is the $GNP$ group  associated to $\mathcal M$.
 Computations of the distance matrix associated to the internal configuration spaces $({C_{\mathcal M}^{int}}, d_P)$ and $({C_{\mathcal M}^{int}}, d_\mathcal O)$ were carried out in Matlab using the in-built function `Procrustes'. %The computation of  this function allows to compute the matrix $R$ as defined in  \eqref{e:distance2}.
We estimated the local dimension of ${C_{\mathcal M}^{int}}$ using local PCA. The algorithm was implemented in Matlab, and it is shown in Figure \ref{pca}. A more detailed exposition on local PCA is presented in Appendix \ref{a:lpca}. Let $\mathcal S=\{C_i\}_{i=1}^N$ be a data set of $N$ conformers of a molecule $\mathcal M$. Given a conformer $C_i\in\mathcal S$ we computed its $k$-nearest neighbours, where $k \ll N$ using the the Matlab function `knn' together with the results of Procrustes. We gave a matrix representation to each element of the permutation invertion group (GPI) which is defined in Section 3. We used this representation to compute the distance matrix associated to the conformational spaces $(\mathcal C^{int}_\mathcal M,d_P)$ and $(\mathcal  C_\mathcal M^{int},d_\mathcal O)$ along with its local dimension.

\begin{algorithm}
%\caption{Distance Matrix OCS and local dimension}
\begin{algorithmic}[1]
\caption{Local PCA with group actions }
\label{pca}
%\Procedure{Roy}{$a,b$}       \Comment{This is a test}
    \Require Data set $\mathcal S$ of $N$ conformers, $C_1,\ldots,C_N$, of a molecule with $n$ atoms in $\mathbb R^{3n}$, $N> 3n$
    \Require The automorphism group $P$ of the molecular graph, as a subgroup of the symmetric group $S_n$
    \Require A constant $\gamma \in (0,1)$; higher values of $\gamma$ result in higher predicted dimension
%    \State Read automorphism group matrix $P$
%    \State Compute the matrix $D$, the $N \times N$ distance matrix associated to the metric space $(\mathcal C_\mathcal M,d_P)$.  
    \For {$i \in \{1,\ldots,N\}$}
    \For {$j \in \{1,\ldots,N\}$}
    \State Let $\tilde p_{j} \in P$ and $\tilde A_j \in SO(3)$ be elements minimising $d_F(C_i,A(C_j \cdot p))$, $p\in P$, $A \in SO(3)$
	\State Set $\tilde C_{ij} = d_{\mathcal O}(C_i,C_j)=d_F(C_i,\tilde A_j(C_j \cdot \tilde p_j))=d_P(C_i,C_j \cdot \tilde p_j)$ 
%    \State Set distance matrix to be DM=zeros(size(C,1), 3N)
%    \For {j=1:size(C,1)}
%    	\For i=1:k
%    \State [A(k),B]=procrustes(C(i), P(k,:))C(j,:))
%    \EndFor
%     \State DM(i,j)=DM(j,i)=min\{A(k)\}
    \EndFor
    \State %Search for $k$ nearest neighbours:
    Let $NN \subseteq \{1,\ldots,n\}$ be such that $\{ \tilde C_{ij} \mid j \in NN \}$ are the $k$ lowest values of $\tilde C_{ij}$ for $1 \leq j\leq N$
   \State Compute PCA for $\{ \tilde A_j(C_j \cdot \tilde p_j) \mid j \in NN \}$ and let $\lambda_1,\ldots,\lambda_{3n}$ be the resulting eigenvalues
   \State Let $d_i \in \{ 1,\ldots,3n \}$ be the smallest $d$ such that $(\sum_{k=1}^d \lambda_k) / (\sum_{k=1}^{3n} \lambda_k) > \gamma$;
   \State \quad $d_i$ is the predicted dimension of $\mathcal{OC}_\mathcal{M}^{int}$ at $C_i$
    \EndFor
%    \If{$variability(j) > \gamma$}
    \State The predicted dimension $ld$  of $\mathcal{OC}_\mathcal M^{int}$ is the median of the $d_i$, $1 \leq i \leq N$
\end{algorithmic}
\end{algorithm}

We tested orientability of the conformational spaces. In \cite{SH} Singer et al.\ developed an algorithm to detect orientability on large data set which are sample from manifolds. In Algorithm \ref{orienta} we present a version of this algorithm that includes the group action of a discrete group.

\begin{algorithm}
%\caption{Distance Matrix OCS and local dimension}
\begin{algorithmic}[1]
\caption{Orientability of conformational spaces} \label{orienta}
%\Procedure{Roy}{$a,b$}       \Comment{This is a test}
    \Require Data set $\mathcal S$ with $N$ conformers $C_i$ of a molecule with $n$ atoms
    \State Perform Algorithm \ref{pca} to obtain the predicted dimension $ld$, and a  $3N\times ld$ matrix $O_i$, $i\in\{1,\dots, N\}$, with column vectors form an orthonormal basis that approximates the tangent space $T_{ C_i}\mathcal C_\mathcal M^{int}$.
    \State For neighbour conformers $C_i$ and $C_j$ obtain $O_{ij}=\underset{O\in O(ld)}{argmin}\|O-O^T_iO_j\|_F$.
    \State Let $Z$ be the $N\times N$ matrix with entries given by $z_{ij}=det O_{ij}$ for nearby points and 0 otherwise.
    \State Define the matrix $A=D^{-1}Z$, where $D$ is diagonal and $D_{ii}=\sum_{i=1}^N|z_{ij}|$.
    \State Compute the top eigenvector $v_{top}$ of A.
    \State Decide the orientability analysing the histogram of the coordinates of $v_{top}$.
       
%\EndProcedure
\end{algorithmic}
\end{algorithm}

\subsection{Persistent homology}
Persistent homology \cite{Edelsbrunner2000,Zomorodian1} is a method of topological data analysis, which has been used to analyse different types of data sets from different areas in recent years, including chemistry.

It is a method of calculating topological, or more accurately, homological, features at different spatial resolutions. Features that persist for a wider range of the spatial parameter are deemed to be more likely to represent true features of the underlying space the data was sampled from, rather than noise, sampling errors or particular choices of some parameters.

In order to calculate persistent homology of a data set, we need to represent the data as a space with a triangulation, called a simplicial complex. For a set of points $S$, with $\vert S \vert = k$, the most common way to do this is to define the $k$-simplex (or $k$-dimensional polytope) $\Delta^S$ with the points as its vertices.

This work involves the use of persistent homology in two different contexts. First of all, we use it to investigate the homology of the points sampled from the conformation space. With some sensible assumptions on the sampling quality of these points, we can assume that we can deduce from this the homology of the conformation space, and therefore say something about its topological features.

Secondly, we use the connection between persistent homology and discrete Morse theory in our analysis of the energy landscapes defined on the conformation spaces.

\subsection{Multidimensional or parametrised persistence vs. discrete Morse theory}
In order to investigate the energy landscapes on the conformation spaces via persistent homology, we need to filter the conformation space (or its combinatorial approximation) by the real-valued energy function. In essence, we are filtering our simplex $\Delta^S$ by both the pairwise distance function $f$ defined above, and the energy function $E:\Delta^S\to \mathbb{R}$.
However, we are not actually interested in the two-parameter persistence module we get this way.

Instead, we wish to construct a combinatorial structure, called the \emph{Morse-Smale complex} of the energy function, which represents the associated gradient flow and summarises it according to its critical points. This construction is described in more detail in the appendix.

In order to construct the Morse-Smale complex with the correct number of critical points, we need to filter a simplicial complex which has the Euler characteristic of the conformational space.
Given a value of $f$ for which the preimage $f^{-1}$ is a triangulation with the 'correct' topology of our conformational space, we can use the lower-star filtration on this complex using the filter function $E$, and compute one-parameter persistence.

In practice, the construction we use to get the triangulation, are $3D$ alpha complexes, which can substantially reduce the computational complexity from having to construct Rips complexes with a threshold value. In fact, alpha complexes are thresholded Rips complexes, but there are easier implementations available, e.g. in Matlab.

To describe this construction, first let us say a word about Voronoi diagrams. Given a set $S$ of points in Euclidean space $\mathbb{R}^n$, one defines convex polytopes $V_s$, $s\in S$ called Voronoi cells, which consist of all points $x\in\mathbb{R}^n$ such that the distance between $x$ and $s$ is less than the distance between $x$ and $s'$ for any other $s'\in S$. The subsets $V_s$ give a tessellation of $\mathbb{R}^n$.

Given a finite set of points $S\subset \mathbb{R}^n$ and a real number $r\in \mathbb{R}$, one defines the region $R_s(r) = \bar{B}_s(r)\cap V_s$, where $\bar{B}_s(r)$ is the closure of the ball of radius $r$ centred at $s$. Now we can form the $\alpha$-complex (or alpha complex) $K_r$ as follows: a subset $\sigma\subset S$ is called an $\alpha$-simplex if
$$\bigcap_{s\in\sigma} R_s(\sigma)\neq\emptyset.$$

After we have discovered the topology of the conformational space from which we sampled our point set $S$, we can safely choose a triangulation $T$ of the point cloud that reflects this topology (i.e. choose a radius for the Rips complex) and linearly extend the energy function $E:S\to\mathbb{R}$ to the entire simplicial complex, which gives us a piecewise-linear function $\tilde{E}:T\to \mathbb{R}$. This allows us to explore the energy landscape of the given molecule.

Let $M$ be a compact manifold and let $f$ be a Morse function defined on $M$. Then the alternating sum of the number of critical points of index $k$ of $f$ equals the Euler characteristic of the manifold. The same goes for a simplicial complex. Therefore, we have a linear relation between the minima of our energy landscape and the other critical points, which is unique to the conformation space (unique in the sense that it depends on the topology of the conformation space).

The Morse-Smale complex is commonly applied for surface segmentation \cite{Sousbie1,cubical,Wang1}. Each segment of the surface has uniform integral lines. This leads to a reduction of information. Instead of the entire energy landscape, we are left with a cell complex which accounts for the unique features of the energy function. Each cell has uniform flow, meaning that we can read off the Morse-Smale complex directly which energy minimum a given conformer will flow to as it loses energy. It also shows the unstable equilibria, the maxima and saddle points where there is no flow, but the slightest perturbation can lead to a more drastic change of the conformation.

Moreover, this cell complex still has the same homology as the conformation space, combining the characteristics of both the conformation space and the energy landscape in a compact, combinatorial structure, which can be regarded as a unique descriptor of the molecule.

\section{A mathematical framework for conformational spaces of molecules}\label{MF}
\subsection{Molecular graphs and conformational spaces}
Our model of conformation spaces of molecules will be given by embeddings of graphs in $\mathbb{R}^3$. In such graphs, we combinatorially encode atoms and bonds in a molecule as vertices and edges of a graph, respectively. To introduce this, we need to define molecular graphs.

\begin{De}
\label{d:mg}
A \emph{molecular graph} is a tuple $\mathcal{M} = (V,E,c_V,L,\Theta)$ consisting of the following data.
\begin{enumerate}
\item $\Gamma = (V,E)$ is a finite undirected \emph{graph}: that is, $V$ is a finite set, and $E \subset V \times V$ is a subset such that, for any $v,w \in V$, we have $(v,v) \notin E$, and $(v,w) \in E$ if and only if $(w,v) \in E$. We refer to $V$ as the set of \emph{vertices} of $\mathcal{M}$, and to $E$ as the set of \emph{edges} of $\mathcal{M}$.
\item $c_V: V \to \mathbb{N}$ is a \emph{vertex colouring}; for $v \in V$, we will think of $c_V(v)$ as the chemical element the atom corresponding to $v$ represents, and will set $c_V(v)$ to be the atomic number of this element.
\item $L: E \to (0,\infty)$ is a set of \emph{length constraints}; $L(e)$ will be called the \emph{length} of an edge $e \in E$.
\item $\Theta: E_2 \to (0,\pi]$ is a set of \emph{angle constraints}, where
\[
E_2 = \{ (v,w_1,w_2) \in V \times V \times V \mid (v,w_1),(v,w_2) \in E, w_1 \neq w_2 \}
\]
is viewed as the set of \emph{adjacent edges} in $\Gamma$; we will refer to $\Theta(v,w_1,w_2)$ as the \emph{angle} between edges $(v,w_1) \in E$ and $(v,w_2) \in E$.
\end{enumerate}
For consistency, we also require that $L(v,w) = L(w,v)$ for every $(v,w) \in E$ and $\Theta(v,w_1,w_2) = \Theta(v,w_2,w_1)$ for every $(v,w_1,w_2) \in E_2$.
\end{De}

\begin{remark}
In Definition \ref{d:mg} the length and angle constraints have the following chemical meanings. The bond constraint is associated to the average of a bond length between two atoms A and B, whose centres of mass are modelled as points in $\mathbb R^n$.  The angle constraint is associated to the angle given by the hybridation type of the atom around which the bond angle is defined. Therefore in this definition we assume that the bond lengths and bond angles are rigid. That is, they are constant for every conformer.
\end{remark}

To define an embedding of a molecular graph $\mathcal{M}$, note that the length constraints can be used to make $\mathcal{M}$ into a geodesic metric space. Indeed, we may define the \emph{geometric realisation} of $\mathcal{M}$ as a metric space defined by gluing an interval $[0,L(e)]$ for each edge $e \in E$ along the vertices $V$ in the obvious way. From now on, slightly abusing notation, we will identify a molecular graph $\mathcal{M}$ with its geometric realisation.

\begin{De} \label{d:emb}
An \emph{configuration} of a molecular graph $\mathcal{M} = (V,E,c_V,L,\Theta)$ is an embedding (injective continuous map) $\varphi: \mathcal{M} \to \mathbb{R}^3$ from the geometric realisation of $\mathcal{M}$ to $\mathbb{R}^3$ such that
\begin{enumerate}
\item \label{i:demb-isom} for any edge $e \subseteq \mathcal{M}$, the restriction $\varphi|_e: e \to \mathbb R^3$ is an isometric embedding; and
\item for any adjacent edges $(v,w_1),(v,w_2) \in E$ (that is, for any $(v,w_1,w_2) \in E_2$), we have $\angle(\varphi(w_1),\varphi(v),\varphi(w_2)) = \Theta(v,w_1,w_2)$, where $\angle(x,y,z)$ is the angle between the line $yx$ and the line $yz$ in $\mathbb{R}^3$ for $x,y,z \in \mathbb{R}^3$.
\end{enumerate}

The \emph{conformational space} of $\mathcal{M}$, denoted $\mathcal C_\mathcal M$, is the set of all embeddings $\varphi$ of $\mathcal{M}$.
\end{De}

From the condition \eqref{i:demb-isom} in Definition \ref{d:emb} and from the fact that geodesics in $\mathbb{R}^3$ are unique, it follows that an embedding $\varphi: \mathcal{M} \to \mathbb{R}^3$ can be recovered uniquely from its values on the finite set $V$. In particular, if $V = \{ v_1,\ldots,v_n \}$ contains $n$ points, then any embedding $\varphi$ can be recovered from the tuple $C_\varphi = (\varphi(v_1),\ldots,\varphi(v_n)) \in \mathbb{R}^{3n}$, called the \emph{vector representation} of $\varphi$. The map $\mathcal C_\mathcal M \to\mathbb R ^{3n}, \varphi \mapsto C_\varphi$ is therefore injective. We use this map to realise $C_{\mathcal M}$ as a subspace of $\mathbb{R}^{3n}$; in particular, as $\mathbb R^{3n}$ is a topological space, this induces a subspace topology on $\mathcal C_\mathcal M$. Thus, from now on, we will regard $\mathcal C_\mathcal M$ as a topological space.

In order to analyse the connectivity properties of $\mathcal C_\mathcal M$, we will use the following.

\begin{De}\label{d:emg}
The \emph{orientation} of an embedding $\varphi$ of the molecular graph $\mathcal{M}$ is the map $O_\varphi: E_3 \to \{-1,0,1\}$, where
\[
E_3 = \{ (v,w_1,w_2,w_3) \in V \times V \times V \times V \mid (v,w_i) \in E, w_i \neq w_j \text{ whenever } i \neq j \},
\]
defined by
\[
O_\varphi(v,w_1,w_2,w_3) = \operatorname{sign}\ (\varphi(w_1)-\varphi(v)) \cdot [(\varphi(w_2)-\varphi(v)) \times (\varphi(w_3)-\varphi(v))],
\]
where $\operatorname{sign} c = \begin{cases} -1 & \text{if } c < 0, \\ 0 & \text{if } c = 0, \\ 1 & \text{if } c > 0, \end{cases}$ for any $c \in $.
\end{De}

Given a molecular graph $\mathcal{M} = (V,E,c_V,L,\Theta)$ and a quadruple $(v,w_1,w_2,w_3) \in E_3$, one can see that the number $|O_\varphi(v,w_1,w_2,w_3)| \in \{0,1\}$ is independent of the embedding $\varphi$ of $\mathcal{M}$. Indeed, it follows from the fact that any two points $x,y \in S^2$ on the sphere $S^2 \subseteq \mathbb R^3$ are joined by a unique geodesic on $S^2$ (unless $y = -x$) that we have
\begin{equation} \label{e:Ozero}
\begin{aligned}
O_\varphi(v,w_1,w_2,w_3) = 0 \qquad \Leftrightarrow \qquad \text{either } &\theta_{12}+\theta_{13}+\theta_{23} = 2\pi \\ \text{or } &\theta_{ij}+\theta_{ik}=\theta_{jk} \text{ for some } \{i,j,k\}=\{1,2,3\},
\end{aligned}
\end{equation}
where $\theta_{ij} = \Theta(v,w_i,w_j)$. We will call a tuple $(v,w_1,w_2,w_3) \in E_3$ \emph{planar} if $O_\varphi(v,w_1,w_2,w_3) = 0$ for some, or any, $\varphi \in \mathcal C_\mathcal M$ (that is, if both sides of \eqref{e:Ozero} are true). We will say that the molecular graph $\mathcal{M}$ is \emph{planar} if all tuples in $E_3$ are planar. Note that, by convention, a disconnected molecular graph is planar if and only if all its connected components are planar.

\begin{remark}
A planar graph $\Gamma$ is a graph for which there exists an embbeding from the geometric realisation of $\Gamma$ into $\mathbb R^2$. Thus our definition of a planar conformation is more restrictive and does not coincide with that of a planar graph. 
\end{remark}

On the other hand, the map $O_\varphi$ itself is \emph{not} independent of the embedding $\varphi$, unless all tuples in $E_3$ are planar. Indeed, given an embedding $\varphi$ of $\mathcal{M}$ and a quadruple $(v,w_1,w_2,w_3) \in E_3$, it is easy to check that $O_{\iota \circ \varphi}(v,w_1,w_2,w_3) = -O_\varphi(v,w_1,w_2,w_3)$, where $\iota: \mathbb R^3 \to \mathbb R^3, x \mapsto -x$ is the antipodal map (it is clear from the definition that $\iota \circ \varphi: \mathcal{M} \to \mathbb R^3$ is also an embedding of $\mathcal{M}$). This allows us to show the following result.

\begin{proposition} \label{p:disconn}
If a molecular graph $\mathcal{M}$ is not planar, then the conformational space $\mathcal C_\mathcal M$ is not path connected. More precisely, if $\varphi,\psi \in \mathcal C_\mathcal M$ and $(v,w_1,w_2,w_3) \in E_3$ are such that $O_\varphi(v,w_1,w_2,w_3) = -O_\psi(v,w_1,w_2,w_3) \neq 0$, then $\varphi$ and $\psi$ are in different path components of $\mathcal C_\mathcal M$.
\end{proposition}

\begin{proof}
As $\mathcal{M}$ is not planar, there exists a quadruple $(v,w_1,w_2,w_3) \in E_3$ that is not planar. Let $\varphi \in \mathcal C_\mathcal M$ be any embedding: then clearly $\iota \circ \varphi \in \mathcal C_\mathcal M$, where $\iota: \mathbb{R}^3 \to \mathbb{R}^3$ is the antipodal map. Thus the first part of the Proposition follows from the second part, by taking $\psi = \iota \circ \varphi$.

To prove the second part of the Proposition, suppose for contradiction that $\varphi,\psi \in \mathcal C_\mathcal M$ are in the same path component of $\mathcal C_\mathcal M$. Then there exists a path $\alpha: [0,1] \to \mathcal C_\mathcal M$ such that $\alpha(0) = \varphi$ and $\alpha(1) = \psi$. Consider the map
\begin{align*}
D: [0,1] &\to \mathbb{R}, \\
t &\mapsto (\alpha_t(w_1)-\alpha_t(v)) \cdot [(\alpha_t(w_2)-\alpha_t(v)) \times (\alpha_t(w_3)-\alpha_t(v))]
\end{align*}
where $\alpha_t = \alpha(t): \mathcal{M} \to \mathbb{R}^3$. The map $D$ is clearly continuous, and it follows from the assumption that $D(0) = -D(1) \neq 0$. Therefore, by the Intermediate Value Theorem, there exists $t \in (0,1)$ such that $D(t) = 0$. But this implies that $O_{\alpha(t)}(v,w_1,w_2,w_3) = 0$, contradicting the fact that $(v,w_1,w_2,w_3) \in E_3$ is not planar. Thus $\mathcal C_\mathcal M$ is not path connected, as required.
\end{proof}
\subsection{Group actions on conformational spaces}
It is known that the quantum mechanical description of a molecular systems must be invariant to the several types of transformations \cite{Bunker}, in which the following are included:
\begin{enumerate}[1)]
\item Rotation of the positions of all particles about any axis through the centre of mass.
\item Translation in space.
\item Permutation of the positions of any set of identical nuclei.
\item Simultaneous inversion of the positions of all particles in the centre of mass. 
\end{enumerate}

Therefore it is crucial to study the symmetries present in our model of the conformational space $\mathcal C_\mathcal M$ by studying actions of groups on $\mathcal C_\mathcal M$. In what follows, for any $n \in \mathbb{N}$, let $S_n$ be the symmetric group on $n$ elements (the group of all permutations of an $n$-element set), let $D_n$ be the dihedral group of order $2n$ (the group of all symmetries of a regular $n$-gon), and let $\mathbb{Z}_n$ be the group of integers modulo $n$. Let $E(3)$ the group of isometries of $\mathbb R^3$ and let $SE^+(3)$ be the subgroup of $E(3)$ of orientation preserving isometries. It is known that there are group isomorphisms $E(3)\cong O(3)\ltimes \mathbb R^3$ and $SE(3)^+\cong SO(3)\ltimes \mathbb R^3$, where $O(3)$ and $SO(3)$ are the group of real invertible matrices with determinant $\pm1$ and $+1$, respectively. 

 In \cite{Gui}, Guichardet approches to configuration spaces of molecules defining the following spaces of $n$-points, $n\geq3$:

\begin{equation}\label{gui0}
    X_\mathcal M^0=\{(x_1,\dots,x_n)\in\mathbb R^{3n}|x_i\neq x_j,\; \mathrm{if}\; i\neq j\}
\end{equation}
\begin{equation}\label{gui1}
    X_\mathcal M^{1}=\{(x_1,\dots,x_n)\in \mathbb R^{3n}|x_i\neq x_j,\; \mathrm{if}\; i\neq j,\;\sum_{i=1}^nm_ix_i=0\}
\end{equation}
%\begin{equation}
 %   X^{ext}_\mathcal M=\{(x_1,\dots,x_n)\in \mathbb R^{3n}|x_i\neq x_j,\; \mathrm{if}\; i\neq j,\:\sum_{i=1}^nm_ix_i=0, span(x_j,\dots,x_n)=\mathbb R^3\}
  %  \end{equation}
The group $\mathbb R^3$ of $E(3)$ acts on $X^0_\mathcal M$ by translation and it is easy to see that $X_\mathcal M ^1=X_\mathcal M^0\times \mathbb R^3$. Adding the condition $span(x_j,\dots,x_n)=\mathbb R^3$ in \eqref{gui1} we obtain another, configuration space, $X_\mathcal M^{P}$. The group $SO(3)$ acts freely on $X_\mathcal M^{P}$ and there is a principal $SO(3)$-bundle:
\begin{equation}
\xymatrix{
SO(3)\ar[r]&X_\mathcal M^{P}\ar[r]^-{\pi_{1}}&X_\mathcal M^{int}}
\end{equation}
where $X_\mathcal M^{int}$ is homeomorphic to the orbit space $X_\mathcal M^{P}/SO(3)$. In our model there are restrictions on the nuclei positions which are imposed by the underlying molecular graph. Therefore it is expected that our  $\mathcal C_\mathcal M$ is a subspace of $X_\mathcal M^{P}$.

Let $\mathcal M=(V.E,c_V,L,\Theta)$ be a molecular graph such that $|V|=n$. Given a conformation $C_\varphi$ of a $\mathcal M$, $\mathcal C_\mathcal M=(\varphi(v_1),\dots,\varphi(v_n))\in\mathbb R^{3n}$, the groups $E(3)$ and $O(3)\cong E(3)/\mathbb R^3$ act on it by isometries. We define the following space
\begin{equation}
    \mathcal C_\mathcal M^{P}=\{C_\varphi=(\varphi(v_1),\dots,\varphi(v_n)) \in\mathbb R^3|\sum_{i=1}^nc_V(v_i)\varphi (v_i)=0\}
\end{equation}
We can also see that $\mathcal C_\mathcal M\cong C^P_\mathcal M\times\mathbb R^3$. We might assume additionally that $|V|\geq 3$. The group $SO(3)$ acts on $\mathcal M$ by rotation which induces an action $C^P_\mathcal M$. Let $\mathcal \mathcal C_\mathcal M^{int}$ be the orbit space of the action of $SO(3)$ on $\mathcal C_\mathcal M^P$. We give the space $C^{int}_\mathcal M$ the quotient topology. Then we have the following result.

\begin{theorem}\label{t:pb}
Let  $\mathcal M$ be a molecular graph. If $\Theta(E_2)\nsubseteq\{\pi\}$ then the quotient map $q:C^P_\mathcal M\to C^{int}_\mathcal M$ defines a  principal $SO(3)$-bundle over $\mathcal C_\mathcal M^{int}$. Moreover this principal $SO(3)$-bundle is trivial.   
\end{theorem}
\begin{proof}
Let $C_\varphi=(\varphi(v_1),\dots,\varphi(v_n))\in \mathcal C_\mathcal M ^P\subset \mathbb R^{3n-3}$. Let us write $x_i=\varphi(v_i)$ for all $0\leq i\leq n$. The elements of $ SO(3)$ act on $C_\varphi=(x_1,\dots,x_n)\in \mathcal C_\mathcal M^P$ with the canonical action. Let $\pi:\mathcal C_\mathcal M^P\to \mathcal C_\mathcal M^{int}$ be the quotient map. By assumption there exists a triple $(v,W_1,W_2)\in E_2$ such that $\Theta(v,w_1,w_2)\neq\pi$ then this action has no fixed points. It follows that $q^{-1}(C_\varphi)\cong SO(3)$ for all $C_\varphi\in \mathcal C_\mathcal M^P$. Therefore, the quotient map $q:\mathcal C_\mathcal M^P\to \mathcal C_\mathcal M^{int}$ defines a principal $SO(3)$-bundle over $\mathcal C_\mathcal M^{int}$. It is a routine calculation to check that $s$ is a continuous map. Since the map $\pi$ that defines the principal $SO(3)$-bundles has a section, then this bundle is trivial.

Now we show that the bundle $\pi:\mathcal C_\mathcal M^P\to \mathcal C_\mathcal M^{int}$ is trivial, that is $\mathcal C_\mathcal M^P\cong \mathcal C_\mathcal M^{int}\times SO(3)$. It suffices to show that there is a map $s:\mathcal C_\mathcal M^{int}\to \mathcal C_\mathcal M^P$ such that the composition $\pi\circ s$ is the identity map on $\mathcal C_\mathcal M^{int}$. This is equivalent to choosing one point from each $SO(3)$ orbit in $\mathcal C_\mathcal M^P$ in a continuous way. Now since the action of $SO(3)$ there exists a unique $A\in SO(3)$ and $D\in\mathbb R^3$ such that 
\begin{align*}
    A\varphi(v)+D&= (0,0,0)\\
    A\varphi(w_1)+D&=(L(v,w_1),0,0)\\
    A\varphi(w_2)+D&= (L(v,w_2)\cos\theta,L(v,w_2)\sin\theta,0)
\end{align*}
where $\theta=\Theta(v,w_1,w_2)$. We define a map $s:\mathcal C_\mathcal M^{int}\to \mathcal C_\mathcal M^P$ by sending  $\pi(C_\varphi)$ to $C_{A\varphi}$. By construction the map $\pi\circ s$ is the identity on $\mathcal C_\mathcal M^{int}$.
\end{proof}

We will start analysing the action of discrete groups.  In \cite{LH}, Longuet-Higgins introduces a symmetry group of non-rigid molecules, called the \emph{complete nuclear permutation inversion group}, or \emph{CNPI group} for short. Suppose we are given a molecule $M$ consisting of $n = n_1+\cdots+n_k$ atoms where $n_i \geq 1$ is the number of atoms of element $c_i$, for some distinct chemical elements $c_1,\ldots,c_k$. We let the \emph{conformational space} $\mathcal C_\mathcal M$ of $\mathcal M$ be the set of all embeddings $V \to \mathbb{R}^3$, where $V$ is a finite set of cardinality $n$ (corresponding to the $n$ atoms of $M$); we then have an obvious injective map $\mathcal C_\mathcal M \to \mathbb{R}^{3n}$, which makes $\mathcal C_\mathcal M$ into a topological space as before. The $CNPI$ group of $M$ is then $CNPI_\mathcal M = S_{n_1} \times \cdots \times S_{n_k} \times C_2$. This group acts on $\mathcal C_\mathcal M$ as follows:
\begin{enumerate}
\item for $i=1,\ldots,k$, the group $S_{n_i}$ permutes the images of the atoms of element $c_i$; and
\item the non-trivial element of $C_2$ sends $\varphi \in C_\mathcal M$ to $\iota \circ \varphi$, where $\iota: \mathbb{R}^3 \to \mathbb{R}^3$ is the antipodal map.
\end{enumerate}

This construction fits in our setting as follows. Let $\Gamma = (V,\varnothing)$ be the \emph{discrete} graph on $|V|=n=n_1+\cdots+n_k$ vertices -- that is, $\Gamma$ has no edges. Label the vertices of $\Gamma$ as $$V = \{ v_{1,1},\ldots,v_{1,n_1},v_{2,1},\ldots,v_{2,n_2},\ldots,v_{k,1},\ldots,v_{k,n_k} \},$$ and define $c_V: V \to \mathbb{N}$ by $c_V(v_{i,j}) = c_i \in \mathbb{N}$. Note that since $\Gamma$ has no edges we also have $E_2 = \varnothing$, and so the functions $L: \varnothing \to (0,\infty)$ and $\Theta: \varnothing \to (0,\pi]$ in Definition \ref{d:mg} are defined uniquely. Thus, we have a molecular graph $\mathcal{M} = (V,\varnothing,c_V,L,\Theta)$. As a topological space, $\mathcal{M}$ is a discrete space of cardinality $n$, and an embedding $\mathcal{M} \to \mathbb{R}^3$ is just a choice of $n$ distinct points in $\mathbb{R}^3$. Thus, our construction generalises the construction in \cite{Gui}. The graph permutation inversion group, introduced below, generalises the concept of the CNPI group. Loosely speaking, this is the group of automorphisms that preserve the structure of the molecular graph $\mathcal{M}$.

\begin{De} \label{d:gp}
Let $\mathcal{M} = (V,E,c_V,L,\Theta)$ be a molecular graph. A \emph{graph permutation} of $\mathcal{M}$ is a bijection $g: V \to V$ such that
\begin{enumerate}
\item for any $v,w \in V$, we have $(v,w) \in E$ if and only if $(g(v),g(w)) \in E$;
\item $c_V(v) = c_V(g(v))$ for every $v \in V$;
\item $L(v,w) = L(g(v),g(w))$ for every $(v,w) \in E$; and
\item $\Theta(v,w_1,w_2) = \Theta(g(v),g(w_1),g(w_2))$ for every $(v,w_1,w_2) \in E_2$.
\end{enumerate}
Moreover, let $\psi \in \mathcal C_\mathcal M^{int}$, and let $C_\psi \subseteq \mathcal C_\mathcal M$ be the set of embeddings $\varphi \in \mathcal C_\mathcal M$ such that $O_\varphi \equiv O_\psi$; by Proposition \ref{p:disconn}, $C_\psi$ is a union of path components of $\mathcal C_\mathcal M$. Then a graph permutation $g$ is said to be \emph{orientation preserving} with respect to $C_\psi$ if
\begin{enumerate}
\setcounter{enumi}{4}
\item $O_\psi(v,w_1,w_2,w_3) = O_\psi(g(v),g(w_1),g(w_2),g(w_3))$ for every $(v,w_1,w_2,w_3) \in E_3$,
\end{enumerate}
and \emph{orientation reversing} with respect to $C_\psi$ if
\begin{enumerate}[\quad\ (1')]
\setcounter{enumi}{4}
\item $O_\psi(v,w_1,w_2,w_3) = -O_\psi(g(v),g(w_1),g(w_2),g(w_3))$ for every $(v,w_1,w_2,w_3) \in E_3$.
\end{enumerate}
We denote by $\Sym_C^+(\mathcal{M})$ and $\Sym_C^-(\mathcal{M})$ the sets of graph permutations that are orientation preserving and orientation reversing with respect to $C = C_\psi$, respectively, and we let $\Sym_C^{\pm}(\mathcal{M}) = \Sym_C^+(\mathcal{M}) \cup \Sym_C^-(\mathcal{M})$. Clearly, the sets $\Sym_C^+(\mathcal{M})$ and $\Sym_C^\pm(\mathcal{M})$ are groups under composition.
\end{De}

It is clear that a graph permutation $g$ defines a map $g: \mathcal C_\mathcal M \to \mathcal C_\mathcal M$ by permuting the points $\{ \varphi(v) \mid v \in V \}$ for an embedding $\varphi: \mathcal{M} \to \mathbb{R}^3$. Moreover, if $\psi \in C_\psi \subseteq \mathcal C_\mathcal M$ are as in Definition \ref{d:gp}, then for a graph permutation $g$ that is orientation preserving with respect to $C_\psi$, this map restricts to a map $g: C_\psi \to C_\psi$. Similarly, if a graph permutation $g$ is orientation reversing with respect to $C_\psi$, then we have a map $\hat\iota \circ g: C_\psi \to C_\psi$, where $\hat\iota: \mathcal C_\mathcal M \to \mathcal C_\mathcal M$ is defined by $\hat\iota(\varphi) = \iota \circ \varphi$ and $\iota: \mathbb{R}^3 \to \mathbb{R}^3$ is the antipodal map. It is also easy to check given two such maps, each of which is either $g$ for $g$ orientation preserving or $\hat\iota \circ g$ for $g$ orientation reversing with respect to $C_\psi$, the composite of these two maps will also have this form. Thus, we define the \emph{graph permutation inversion group}, or \emph{GPI group}, of $\mathcal{M}$ with respect to $C = C_\psi$ to be
\[
GPI_{\mathcal{M},C} = \left\{ g \mid g \in \Sym_C^+(\mathcal{M}) \right\} \sqcup \left\{ \hat\iota \circ g \mid g \in \Sym_C^-(\mathcal{M}) \right\}.
\]
Note that
\[
GPI_{\mathcal{M},C} \cong \begin{cases} \Sym_C^\pm(\mathcal{M}) \times \mathbb{Z}_2 & \text{if $\mathcal{M}$ is planar}, \\ \Sym_C^\pm(\mathcal{M}) & \text{otherwise}. \end{cases}
\]
We will usually omit $C$ from the notation and will simply talk about the GPI group $GPI_{\mathcal{M}}$ of $\mathcal{M}$. %If we consider the group action of the finite group $GPI$ on $\mathcal C_\mathcal M^{int}$, we can extend Guichardet construction to a  $SO(3)$-bundle over an orbifold $\mathcal OC_{\mathcal M}^{int}$.
The following theorem is clear because $GPI$ is a finite group. An  introduction to orbifolds is given in Appendix \ref{pgborbi}.
\begin{theorem}\label{t:orbifoldcs}
Let $\mathcal M$ be a molecular graph and $G$ be its $GPI$ group. Suppose that $\mathcal C_\mathcal M^{int}$ is a manifold. Then there is a properly discontinuous action of the group $G$ on $\mathcal C_{\mathcal M}^{int}$. In particular, if $GPI$ is non trivial, the quotient space $\mathcal C_{\mathcal M}^{int}/G$ has the structure of an orbifold.
\end{theorem}

\begin{example}
Let $\mathcal M=(\Gamma,c_V,\Theta)$ be the molecular graph associated to pentane. If we only consider the carbon atoms with a rigid conformation, then the  orbifold conformational space is homeomorphic to the quotient of a torus $S^1\times S^1$ by the action of the group Aut$(\Gamma)=C_2$. The groups $C_2$ acts on the conformational space by permuting the ends of the molecular graph $\mathcal M$. This action  is equivalent to permute the parameters $t_1$ and $t_2$ associated to two torsion angles.  As it is shown in Figure \ref{MB}.
\end{example}
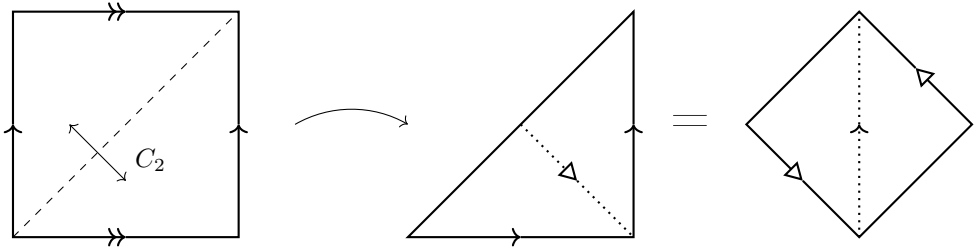
\begin{figure}[ht]
\begin{center}
\begin{tikzpicture}[scale=1.5,baseline={(current bounding box.center)}]
\begin{scope}[xshift=-3cm]
\node at (-0.1,2.1) {};
\node at (2.1,-0.1) {};
\draw [thick,<<->] (1,0) -- (0,0) -- (0,1);
\draw [thick,->] (1,0) -- (2,0) -- (2,1);
\draw [thick,->>] (0,1) -- (0,2) -- (1,2);
\draw [thick] (2,1) -- (2,2) -- (1,2);
\draw [dashed] (0,0) -- (2,2);
\draw [<->] (0.5,1) -- (1,0.5) node [above right] {$C_2$};
\end{scope}
\draw [->] (-0.5,1) arc (120:60:1);
\begin{scope}[xshift=0.5cm]
\draw [thick,->] (1,0) -- (2,0) -- (2,1);
\draw [thick,->] (2,1) -- (2,2) -- (0,0) -- (1,0);
\draw [thick,dotted,-open triangle 60] (1,1) -- (1.5,0.5);
\draw [thick,dotted] (1.5,0.5) -- (2,0);
\end{scope}
\node at (3,1) {\huge =};
\begin{scope}[xshift=2.5cm]
\draw [thick,-open triangle 60] (2.5,1.5) -- (2,2) -- (1,1) -- (1.5,0.5);
\draw [thick,-open triangle 60] (1.5,0.5) -- (2,0) -- (3,1) -- (2.5,1.5);
\draw [thick,dotted,->] (2,0) -- (2,1);
\draw [thick,dotted] (2,1) -- (2,2);
\end{scope}
\end{tikzpicture}
\end{center}
\caption{Action of $C_2$ on $S^1\times S^1$. The orbifold conformational space $\mathcal {OC}_\mathcal M^{int}=\mathcal C_\mathcal M^{int}/G$ is homeomorphic to a Moebius strip.} \label{MB}
\end{figure}

\begin{table}
\begin{center}
\begin{tabular}{r cll}
Graph && CNPI group & GPI group \\
\hline\vspace{-1em}\\
$n$-cycle & \begin{tikzpicture}[scale=0.4,baseline={(current bounding box.center)}]
\fill (0,1) circle (3pt);
\fill (0,2.5) circle (3pt);
\fill (1,3.5) circle (3pt);
\fill (2.5,3.5) circle (3pt);
\fill (3.5,2.5) circle (3pt);
\fill (3.5,1) circle (3pt);
\fill (2.5,0) circle (3pt);
\fill (1,0) circle (3pt);
\draw (0,1) -- (0,2.5) -- (1,3.5) -- (2.5,3.5) -- (3.5,2.5) -- (3.5,1) -- (2.5,0) -- (1,0) -- cycle;
\end{tikzpicture} & $S_n \times C_2$ & $D_{2n} \times C_2$ \\
$n$-path & \begin{tikzpicture}[scale=0.4,baseline={(current bounding box.center)}]
\fill (0,0) circle (3pt);
\fill (1,0.5) circle (3pt);
\fill (2,0) circle (3pt);
\fill (3,0.5) circle (3pt);
\fill (4,0) circle (3pt);
\fill (5,0.5) circle (3pt);
\fill (6,0) circle (3pt);
\fill [white] (3,-0.5) circle (3pt);
\fill [white] (3,1) circle (3pt);
\draw (0,0) -- (1,0.5) -- (2,0) -- (3,0.5) -- (4,0) -- (5,0.5) -- (6,0);
\end{tikzpicture} & $S_{n+1} \times C_2$ & $C_2 \times C_2$ \\
pentane & \begin{tikzpicture}[scale=0.4,baseline={(current bounding box.center)}]
\draw (-0.6,-0.3) -- (0,0) -- (1,0.5) -- (2,0) -- (3,0.5) -- (4,0) -- (4.6,-0.3);
\draw (-0.6,0.3) -- (0,0) -- (0,-0.6);
\draw (4.6,0.3) -- (4,0) -- (4,-0.6);
\draw (0.7,1.1) -- (1,0.5) -- (0.9,-0.1);
\draw (3.3,1.1) -- (3,0.5) -- (3.1,-0.1);
\draw (2,0.6) -- (2,-0.6);
\fill (0,0) circle (3pt);
\fill (1,0.5) circle (3pt);
\fill (2,0) circle (3pt);
\fill (3,0.5) circle (3pt);
\fill (4,0) circle (3pt);
\filldraw [fill=white] (-0.6,0.3) circle (3pt);
\filldraw [fill=white] (-0.6,-0.3) circle (3pt);
\filldraw [fill=white] (0,-0.6) circle (3pt);
\filldraw [fill=white] (4.6,0.3) circle (3pt);
\filldraw [fill=white] (4.6,-0.3) circle (3pt);
\filldraw [fill=white] (4,-0.6) circle (3pt);
\filldraw [fill=white] (0.7,1.1) circle (3pt);
\filldraw [fill=white] (0.9,-0.1) circle (3pt);
\filldraw [fill=white] (3.3,1.1) circle (3pt);
\filldraw [fill=white] (3.1,-0.1) circle (3pt);
\filldraw [fill=white] (2,0.6) circle (3pt);
\filldraw [fill=white] (2,-0.6) circle (3pt);
\end{tikzpicture} & $S_5 \times S_{12} \times C_2$ & $(C_3 \wr C_2) \rtimes C_2$ \\
\makecell{tetraethylmethane \\ \hfill (heavy atoms)} & \begin{tikzpicture}[scale=0.4,baseline={(current bounding box.center)}]
\fill (0,0) circle (3pt);
\fill (1,0.5) circle (3pt);
\fill (2,0) circle (3pt);
\fill (3,0.5) circle (3pt);
\fill (4,0) circle (3pt);
\fill (0.5,-1) circle (3pt);
\fill (1.5,-1) circle (3pt);
\fill (3,-0.5) circle (3pt);
\fill (3.5,-1.5) circle (3pt);
\draw (0,0) -- (1,0.5) -- (2,0) -- (3,0.5) -- (4,0);
\draw (0.5,-1) -- (1.5,-1) -- (2,0) -- (3,-0.5) -- (3.5,-1.5);
\end{tikzpicture} & $S_9 \times C_2$ & $S_4$
\end{tabular}
\caption{Graphs and their associated symmetry groups}
\label{t:graphs}
\end{center}
\end{table}

\subsection{Metrics on conformational spaces}

 %By Theorem \ref{t:orbifoldcs}
%We consider the action on the conformational space $\mathcal C_\mathcal M^{int}$ of the symmetry group associated to $\mathcal M$ :
%\begin{enumerate}
%\item If molecule is achiral, the group $SO(3)\times \widetilde {CGP}$ acts on $C_{\mathcal M}$.
%\item  If molecule is chiral, the group $SE(3)\times \widetilde {CGP}$ acts on $C_{\mathcal M}$
%\end{enumerate}
%The quotient space of $C_\mathcal M$ by the action of the symmetry group of $\mathcal M$ is the orbifold configuration space, $\mathcal OC_\mathcal M$. This action might have a non empty set of fixed points.

%The geometry and topology of the conformational spaces associated to $\mathcal M$ is defined by its metric. 
Let $\mathcal M=(V,E,c_V,\Theta)$ be a molecular graph such that $|V|=n$. We can endow a conformational space $\mathcal C_{\mathcal M}^{int}$ with the following metrics: 
\begin{enumerate}
\item Given two matrices $X,Y\in\mathrm{Mat}_{3\times n}(\mathbb R)$ representing two conformers in $\mathcal C_\mathcal M^{int}$ the Frobenius distance $d_F(X,Y)$ between them, is defined to be $$d_F(X,Y)=\|X-Y\|_F,$$ where $\|-\|_F$ is a matrix norm of an $3\times n$ matrix $A$ defined as the squared root of the sum of the absolute squares of its entries
\begin{equation}
\|A\|_F=\sqrt{\sum_{i=1}^3\sum _{i=1}^n\mid a_{ij}\mid^2}.
\end{equation}
Observe that given a matrix $A\in \mathrm{Mat}_{p\times n}$ the Frobenius norm $\mathcal C^{int}_\mathcal M$ coincides with the Euclidean metric $\|-\|_2$ in the vector representation of $\mathcal C^{int}_\mathcal M$. 
\item %Let $\Pi=E(3)$ or $\Pi=SE(3)$. 
Given two conformers $X,Y\in \mathcal C_\mathcal M^{int}$ and a matrix $R\in SO(3)$, the Procrustes distance function is defined as, %$d^P_{C_\mathcal M}(x,y)$ as  solution to the equation
    $$d_P(X,Y)=\underset{R\in SO(3)}{\inf}\{\| X-RY\|_F\}$$
%we define the Procrustes distance function, %$d^P_{C_\mathcal M}(x,y)$ as
%\begin{equation}\label{e:distance1}
%d_{C}^P(x,y)=\| X-RY\|_F
%\end{equation}

\item A common distance metric used in molecular sciences is the root-mean-square deviation (RMSD):
$$d_{RMSD}(X,Y) = \frac{1}{\sqrt{N}}d_P(X,Y)$$
%$$
    %d_{RMSD}^{\mathcal C_{\mathcal M}^{int}}(X,Y) = \sqrt{\frac{\sum_{i=1}^3\sum _{i=1}^n\mid a_{ij}\mid^2}{N}}
%$$
which is commonly used when aligning chemical structures such as proteins \cite{Duan2014}, and can be shown to be a metric \cite{Steipe2002,Steipe2002a,Sadeghi2013}.
\end{enumerate}

Let $\mathcal M$ be a molecular graph and let $ {GPI}$ be the graph permutation inversion group of $\mathcal M$. We can give the orbifold configuration space $O\mathcal C_\mathcal M^{int}$ the following metric
\begin{equation}\label{e:distance2}
d_{\mathcal O}(X,Y)=\underset{g\in{GPI}}{min}\{d_{P}(X,g\cdot Y)\}
\end{equation}

\section{Results and discussion}

Following our discussion in Section \ref{MF}, each conformer $C_{\varphi}$ in the configuration space $\mathcal C_\mathcal M$ is uniquely  determined by the positions of the atoms in the molecules.  A conformation of a molecule $\mathcal M$ with $n$ atoms is represented by an $n$-by-$3$ matrix with real coefficients. The entries of the $i$-th row vector in this matrix are the spatial coordinates of the $i$-th atom in $C_{\varphi}$. We eliminate the action of the subgroup of translations $T$ by fixing the centre of mass of each conformer at $(0,0,0)\in\mathbb R^3$. From the set of molecular configurations $\mathcal S=\{C_i\}_{i=1}^N$ sampled from the configuration space $\mathcal C_\mathcal M^P$ we generated a data set of the following metric spaces:
\begin{itemize}
    \item The metric space $(X,d_F)$.  Assuming that the bond angles are almost constant then by Theorem \ref{t:pb}, we have that $\mathcal C_\mathcal M\cong \mathcal C_\mathcal  M^{int}\times SO(3)$. We generate a set of $n$-by-3 matrix in $\mathbb{R}^{n\times3}$, where $n$ is the number of atoms in the molecule. Each matrix is associated to an aligned conformer. The euclidean metric has been used in other studies of conformational spaces, such as in \cite{Martin2010,Martin2011}.
    \item The metric space $(X,d_P)$ defined by the Procrustes metric. We obtained a distance matrix.
    \item The metric space $(X,d_{RMSD})$ defined by the RMSD metric. We obtained a distance matrix. 
    \item The metric space $(X,d_\mathcal O)$ defined by the orbifold metric. We obtained a distance matrix.
\end{itemize}

\subsection{Local dimension and orientability}

 %Thus in this case the conformational space $\mathcal C_\mathcal M^{int}$ corresponds to a subset of $\mathbb  R^{3n-6}$ topologized with the metric topology induced by the Euclidean distance $\|-\|_2$.

We used geometric and topological data analysis to study the conformational spaces of small molecules: butane, pentane, alanine dipeptide, and cyclooctane.  We used principal component analysis, or PCA -- one the most common tools in data analysis.

The data set of conformers associated to a molecule $\mathcal M$,  consists of a set $\mathcal S$ of 3-by-$n$ real matrices $\mathcal S=\{A_i\}_{i=1}^N$, with $N$ the number of  conformers. We compute the distance matrix to study the geometric and topological properties of the conformational space $\mathcal C_\mathcal M$.
In Table \ref{locpca} we show the result of local dimension and orientability. In our analysis, fluoromethane is the only molecule with no torsional angles. That is, it is the only molecule that has a contractible conformational space. %Although fluoromethane is the smallest molecule, the local dimension observed in this case is the highest: this is due to small variations of lengths of bonds and angles between them. In contrast, for molecules with torsional angles, the geometry and topology of their conformational spaces is determined (in most part) by the torsional angles. %This observation shows that for the case of molecules with torsional angles the variables that determine  are precisely their torsion angles.

Local dimension shows that data variability depends, to a great extent, on torsional angles. Indeed, all the non-rigid molecules  have at most 2 torsional angles, and this allows to reduce the local dimension of the conformational space from $3n-6$ to 2 or 3 dimensions. In contrast, although fluoromethane is the smallest molecule, the local dimension observed in this case is the highest. For this molecule, the reduction of the local dimension depends on other parameters such as lengths of bonds and angles between them.

One interesting outcome of our local dimensionality study is the detection of singularities. This algorithm approximates the minimal dimension required to span a neighbourhood of a point. If the analysed point is a singular point, its neighbourhood  will typically have a higher dimension than the one observed for a point whose neighbourhood can locally be modelled as a subset of Euclidean space. Thus we could detect the singularities in the conformational spaces of cyclooctane and pentane after quotienting out the action of $C_2$. In the former case, the local dimension detected at a singular point is higher that the local dimension at non-singular points. In \cite{Martin2011}, it was shown that the conformational space of cyclooctane is a non-manifold. More specifically, the authors showed that the conformational space can be embedded in 5-dimensional Euclidean space and that it corresponds to the space formed by the Klein bottle and a $2$-sphere intersecting in $2$ circles. The local dimension and the orientability of the clusters was determined. Our results show that the local dimension of the set of singular points is one, whereas the local dimension of the other clusters is 2. Moreover, one of these clusters is orientable and the other is not. Persistent homology of these spaces is shown in Table \ref{locpca}.

\begin{table}[ht]
\centering
\begin{tabular}{c c c c}
\hline
Molecule &Loc. Dim.&dim. singularities &Orientable\\ [0.5ex] % inserts table %heading
\hline
%fluoromethane (including H)&5&1,2,4\\
%ethane (including H)&5&1,2,4& Yes\\
butane&2&- &Yes\\
pentane&2 &-&Yes\\
alanine dipeptide&2&-&Yes\\
cyclooctane & 2,3 &1&No\\
pentane($C_2$) & 2 & 1 &No\\ 
%Pentane (singular set) & 2 & -&Yes \\ 
Cyclooctane (singular set) & 1 & -&Yes \\ 
Cyclooctane (sphere)&2&-&Yes \\
Cyclooctane (Klein bottle) & 2&-&No \\[1ex]
\hline
\end{tabular}
\caption{Local dimension and orientability} \label{locpca}
\label{table:nonlin}
\end{table}

\begin{figure}[ht]
\centering
\begin{subfigure}{0.49\textwidth}
\includegraphics[width=\textwidth]{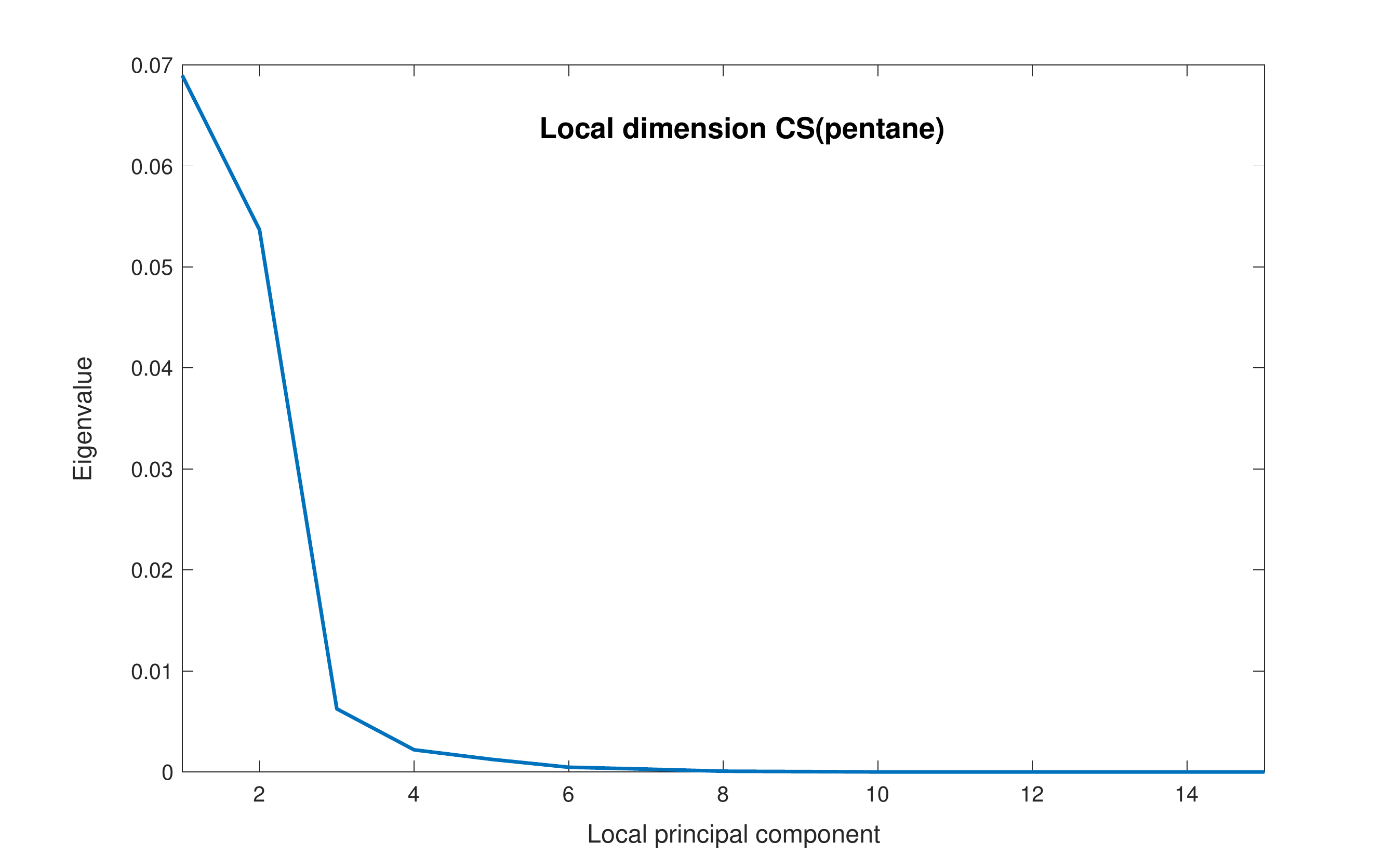}
\caption{}
\end{subfigure}
\begin{subfigure}{0.49\textwidth}
\includegraphics[width=\textwidth]{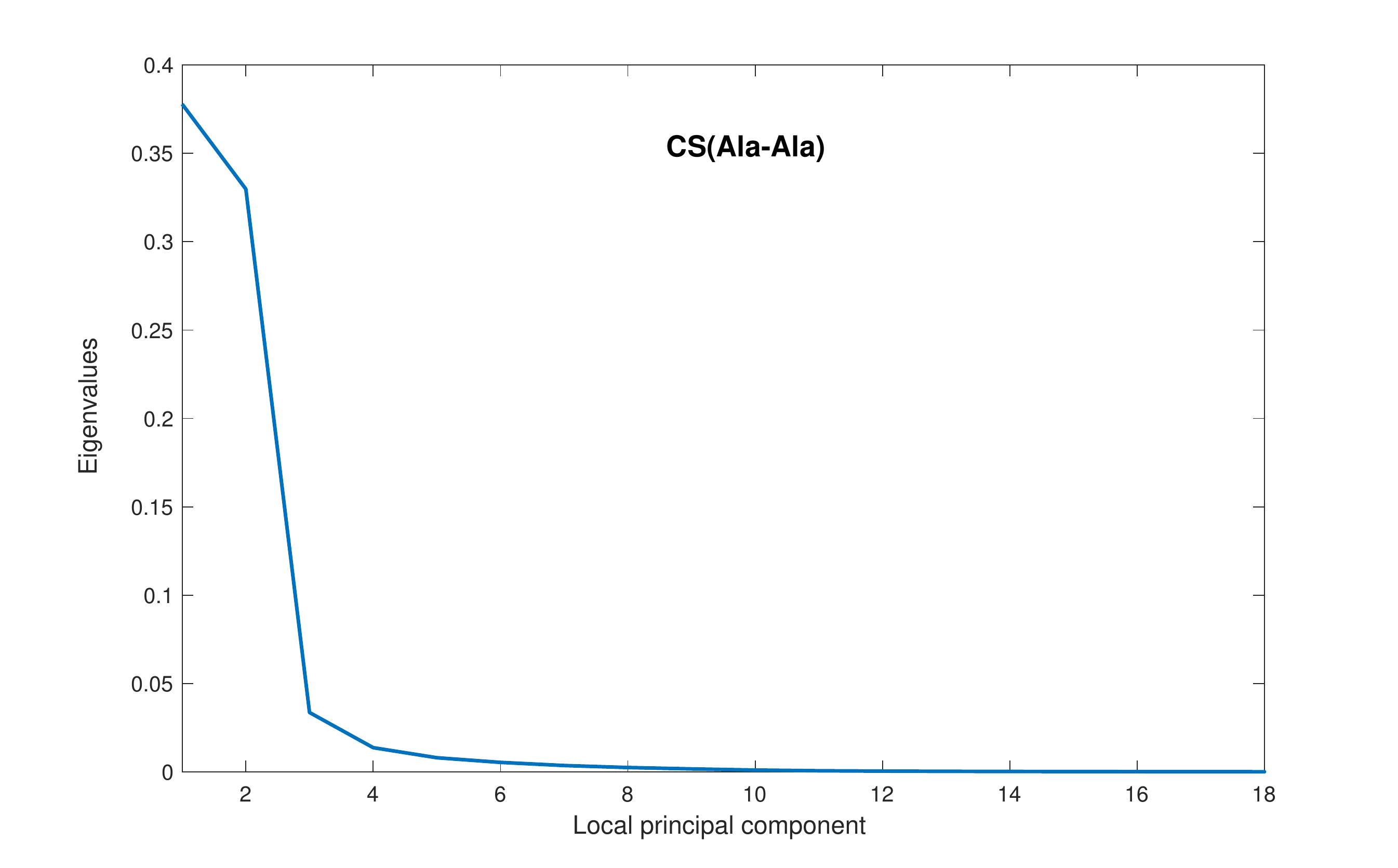}
\caption{}
\end{subfigure}
\caption{Local dimension of  $ \mathcal C_\mathcal M^{int}$ was determined using PCA at each point. Plots (a) and (b) show the local PCA at one point of $\mathcal C_\mathcal M^{int}$ of pentane and alanine dipeptide, respectively, with the euclidean metric. In both cases local PCA suggests that there are two principal components. This implies that the local dimension at a  chosen point $C_\varphi\in \mathcal C_\mathcal M^{int}$ of both pentane and alanine dipeptide is 2. }
\end{figure}

\begin{figure}[ht]
\centering
\begin{subfigure}{0.49\textwidth}
\includegraphics[width=\textwidth]{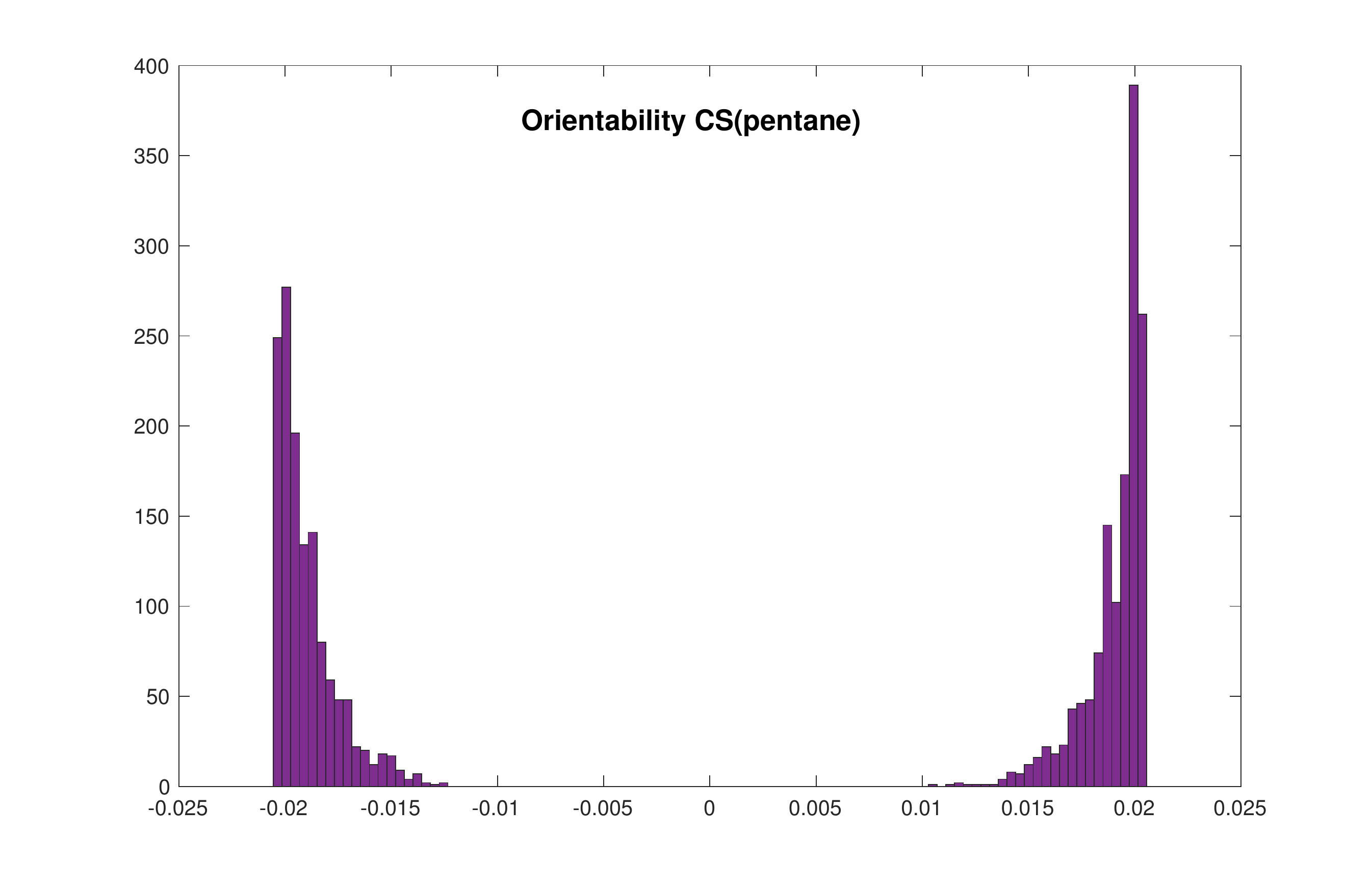}
\caption{}
\end{subfigure}
\begin{subfigure}{0.49\textwidth}
\includegraphics[width=\textwidth]{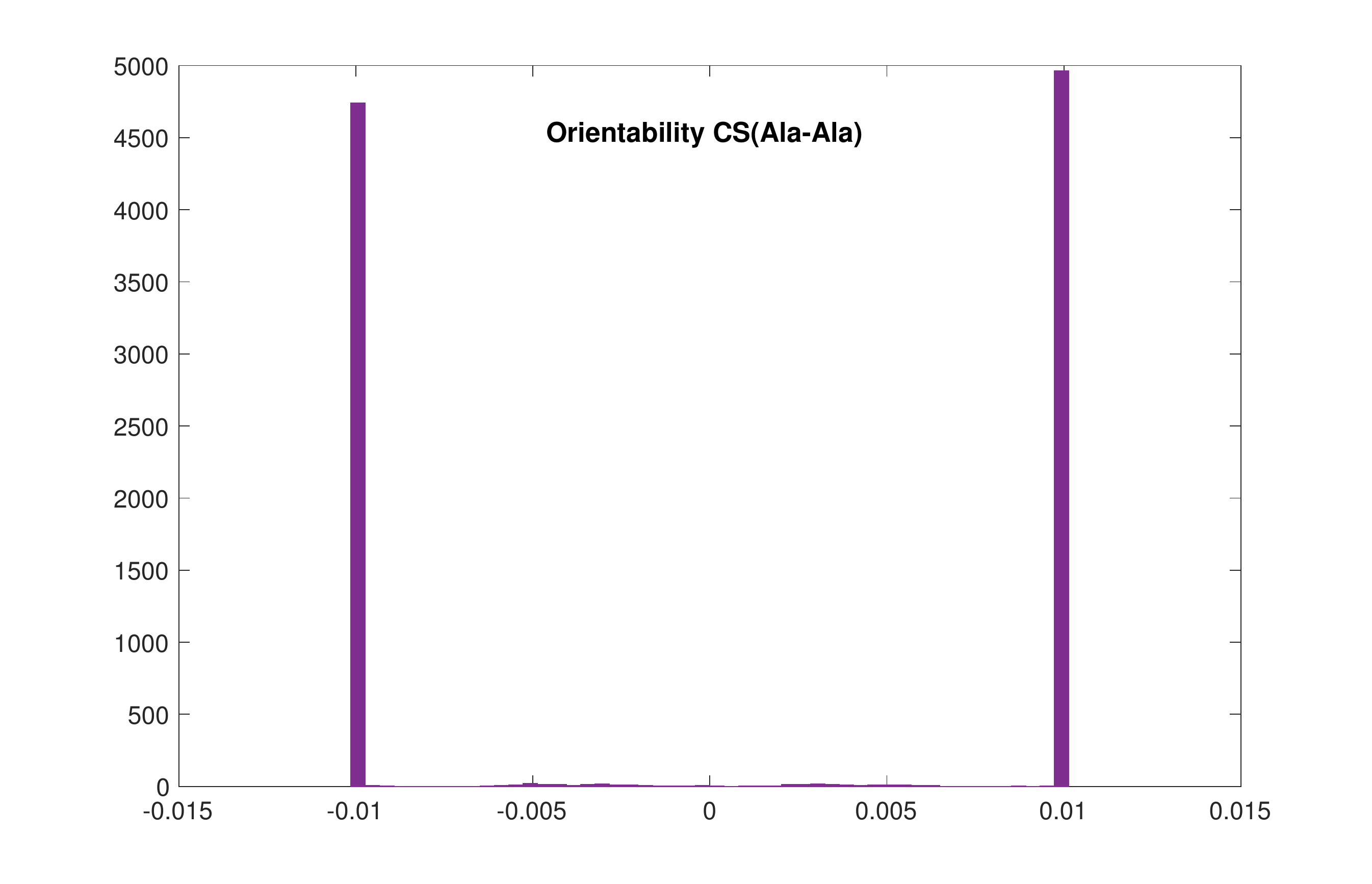}
\caption{}
\end{subfigure}
\caption{The results of the detection of orientability of pentane and dipeptide alanine are shown in figures (a) and (b), respectively. There is noise in the data set of pentane, however it is possible to distinguish two well separated region which corresponds to a choice of orientation at each point in both pentane and alanine dipeptide.}
\end{figure}

\begin{figure}[ht]
\centering
\begin{subfigure}{0.44\textwidth}
\includegraphics[width=\textwidth]{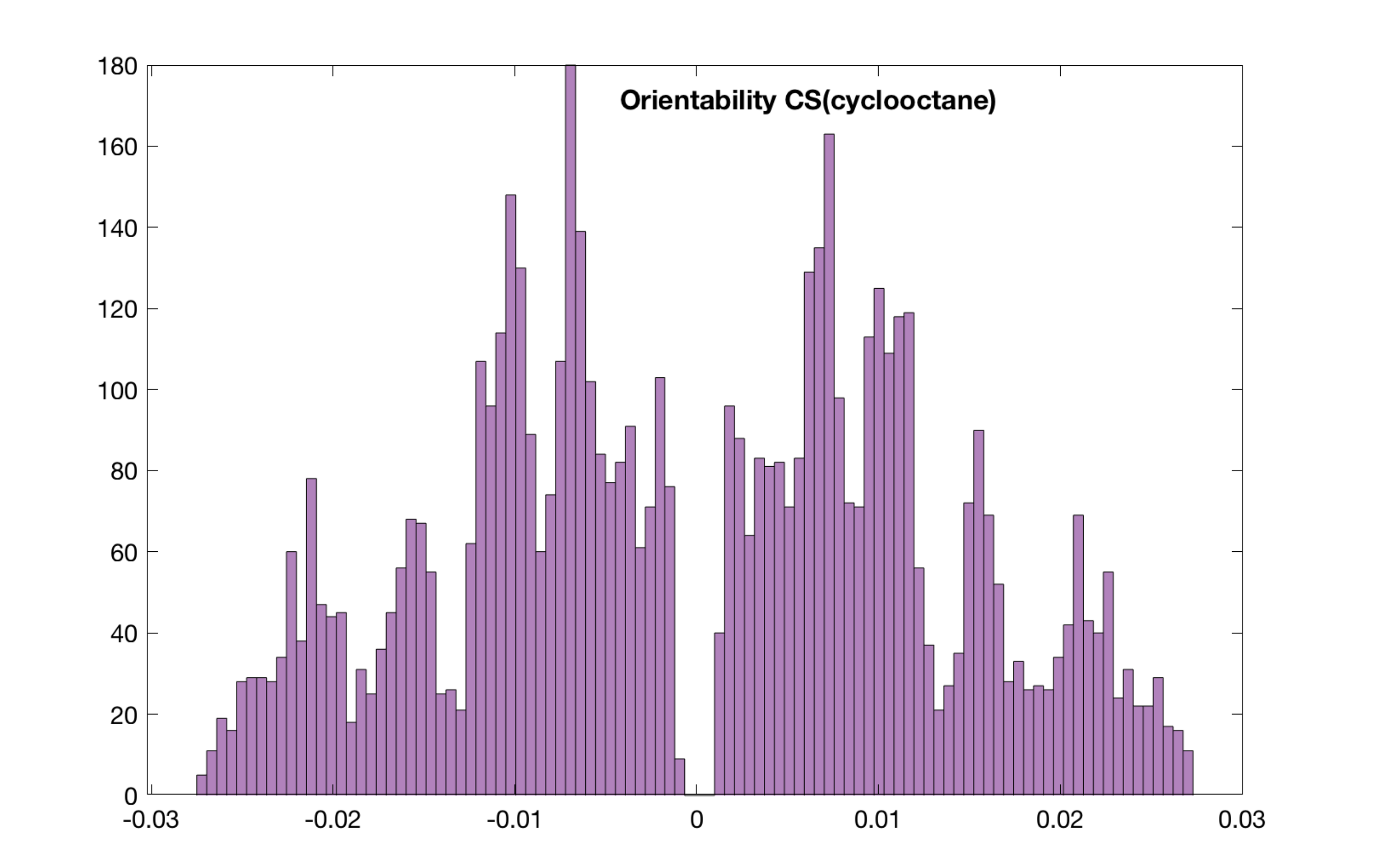}
\caption{}
\end{subfigure}
\begin{subfigure}{0.50\textwidth}
\includegraphics[width=\textwidth]{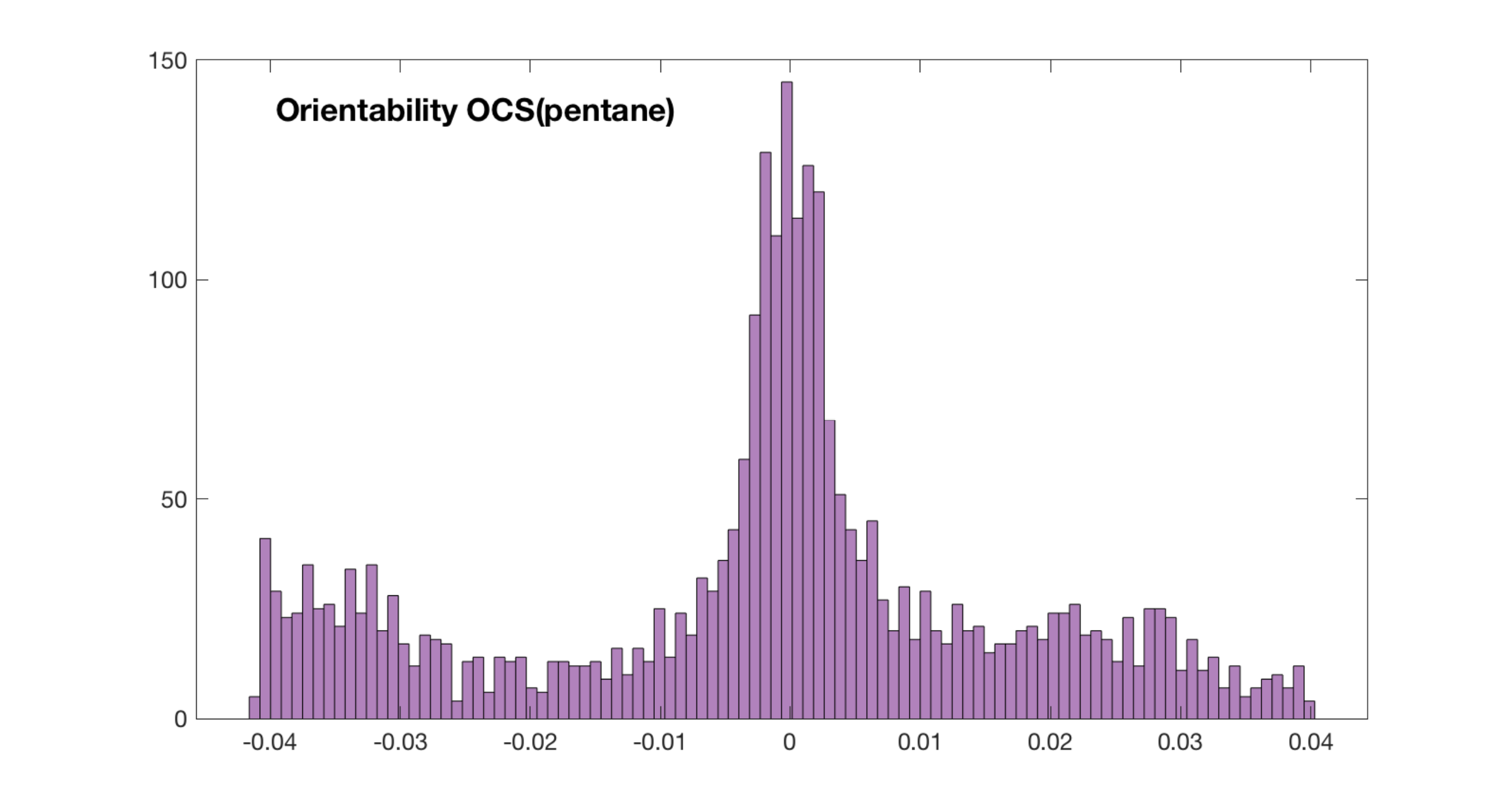}
\caption{}
\end{subfigure}
\caption{Figure (a) shows the orientability test result for $\mathcal C_\mathcal M^{int}$ of cyclooctane with the euclidean metric whereas figure (b) shows the corresponding result for $\mathcal C_\mathcal M^{int}$ of pentane with the orbifold metric.}
\end{figure}

\begin{figure}[ht]
\centering
\begin{subfigure}{0.45\textwidth}
\includegraphics[width=\textwidth]{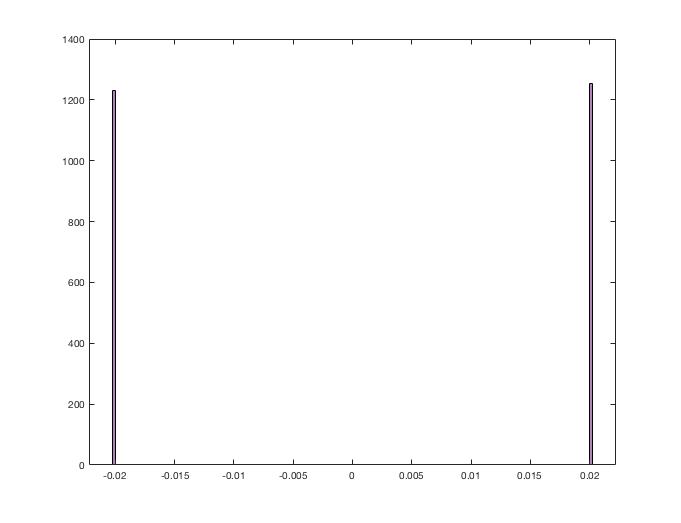}
\caption{}
\end{subfigure}
\begin{subfigure}{0.5\textwidth}
\includegraphics[width=\textwidth]{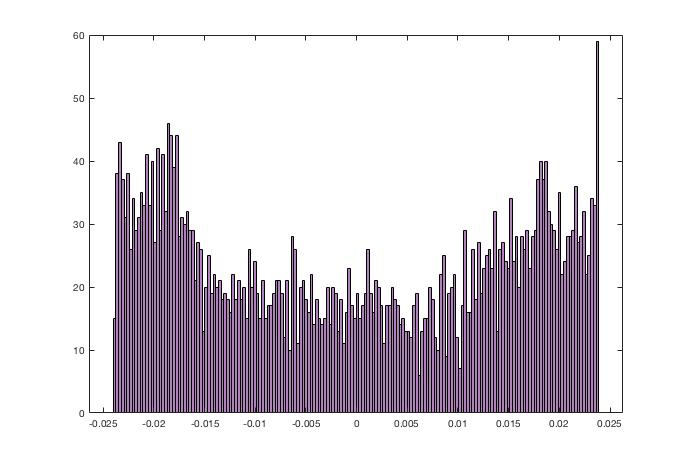}
\caption{}
\end{subfigure}
\caption{Orientability of clusters of cyclooctane: (a) orientable cluster (b) non-orientable cluster. The clustering method identified two clusters of local dimension 2 which present different orientability properties. }
\end{figure}

%\begin{comment}
\begin{figure}[ht]
\begin{subfigure}{0.50\textwidth}
\includegraphics[width=\textwidth]{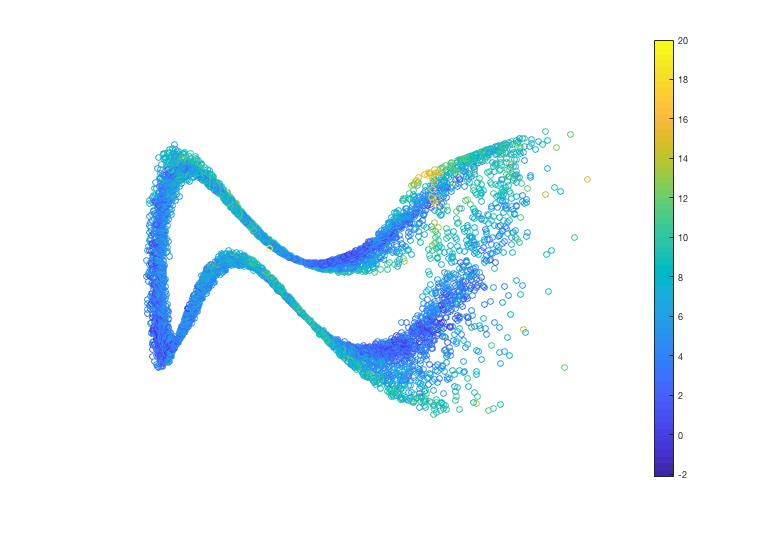}
\caption{}
\end{subfigure}
\begin{subfigure}{0.50\textwidth}
\includegraphics[width=\textwidth]{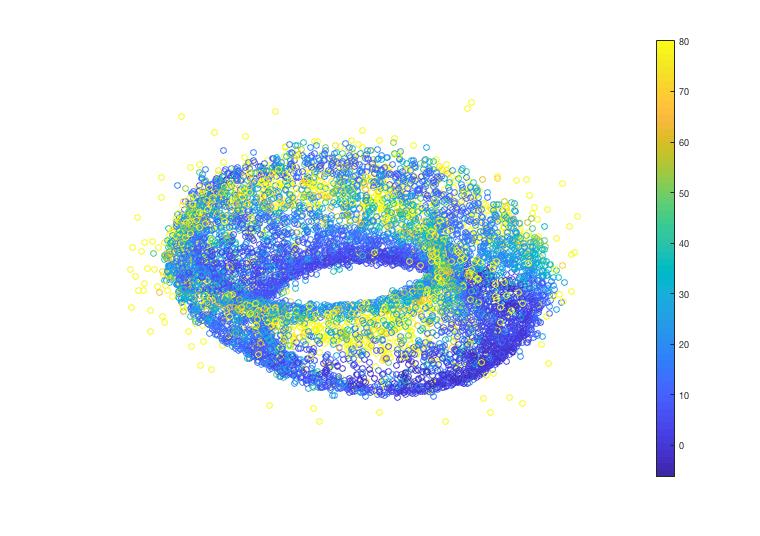}
\caption{}
\end{subfigure}
\caption{3d-embbeding of $\mathcal C_\mathcal M^{int}$ spanned by heavy atoms of (a) butane and (b) pentane. The scatter plots are coloured by the energy function.}
\end{figure}
%\end{comment}

%\subsection{The homology of the conformation space}

%Let $M$ be a compact manifold with boundary. Let $f,$ $g$ be a pair of a Morse function and metric such that $-\nabla f$ is contained in the boundary of $M$. Assume now that the negative gradient flow points inwards at the boundary. Then the alternating sum of the number of critical points of index $k$ equals the Euler characteristic of the manifold. If the gradient points outwards, the same holds, but one needs to compute the Euler characteristic of the manifold relative to the boundary.
\subsection{Persistent homology of conformational spaces}
% \lee{Here is what I'm thinking here}
% \begin{itemize}
%      \item Alanine dipeptide: showing that RMSD and Euclid. representations can be equivalent, and that its a torus. maybe also that hydrogens might not matter (not sure about the last one)
%      \item Pentane: showing that RMSD and Euclid represenations are not always the same - you need to be clever. \ingrid{what do you mean when you say they are not always the same?} Also that RMSD + symmetry gets the topology correct. \ingrid{the topology is the same, both metrics induce the same topology,  what is different is the distance between points, since as metric spaces they are not the same.} \lee{with the euclidean metric, you need to do some alignment to a reference. Depending on what alignment is you get different results. With alanine dipeptide we align to a core and get a really nice matching, but with pentane it doesnt work so well (see Figure \ref{fig:pentane_representations})}
%      \item Cyclooctane: showing that RMSD clustering works, and that the topology is correct (matching the euclid. rep found earlier)
% \end{itemize}
% \lee{What do you think?}

%\subsubsection{Effect of Representation}
\begin{table}[ht]
\centering
\begin{tabular}{c c c c c}
\hline 
Molecule &$N$& $\beta_0$ & $\beta_1$ & $\beta_2$ \\% [0.5ex] \\% inserts table %heading
\hline
Alanine dipeptide & 9112 & 1 & 2 & 1\\
Pentane & 9108 &1 & 2 & 1\\
Pentane/$C_2$ & 9108& 1 & 1 & 0\\
Cyclooctane (full) & 6040 & 1 & 1 & 2\\
Cyclooctane (sphere) & 2483 & 1 & 0 & 1\\
Cyclooctane (Klein bottle) & 4196 & 1 & 2 & 1\\
Cyclooctane (singularities) &639 &1 & 1&0\\
Cyclooctane (Klein bottle mod 3) & 4196 & 1 & 1 & 0\\
Fluoromethane & 10000 & 1 & 0 & 0\\
\hline
\end{tabular}
\caption{Betti numbers $\beta_k$ for the conformational spaces of the molecules studied in this work, calculated using the RMSD for all molecules and orbifold metric for pentane. The Betti numbers of four subspaces of the conformational space of cyclooctane are shown in the table. } %\lee{I might add the conformational space we think it is as another column - but which acronym is the correct one? \ingrid{homotopy type}} \lee{I mean, what symbol do you use for the conformation space when viewed as this quotient etc.? Lik3 $C_\mathcal{M}$ or something?}}
\label{tab:molecule_betti}
\end{table}
%\ingrid{For the discussion of the classical persistent homology I suggest to create a table with  the betti numbers of each molecule}
In this section we analyse the topology of the internal configuration spaces $C_{\mathcal M}^{int}$. We investigate whether the choice of metric for the conformational space leads to a significantly different persistent homology.

\subsubsection{Alanine Dipeptide}
The alanine dipeptide molecule can be seen in Figure \ref{fig:aladip_structure}. We note that there are two free torsions in the alanine dipeptide molecule, as the peptide bonds themselves are considered to be inflexible due to its resonance stabilisation. Therefore, ignoring bond stretching and bending, alanine dipeptide would be predicted to have a conformational space of $S^1\times S^1=T^2$.
\begin{figure}[h]
\centering
\includegraphics[width=0.35\linewidth]{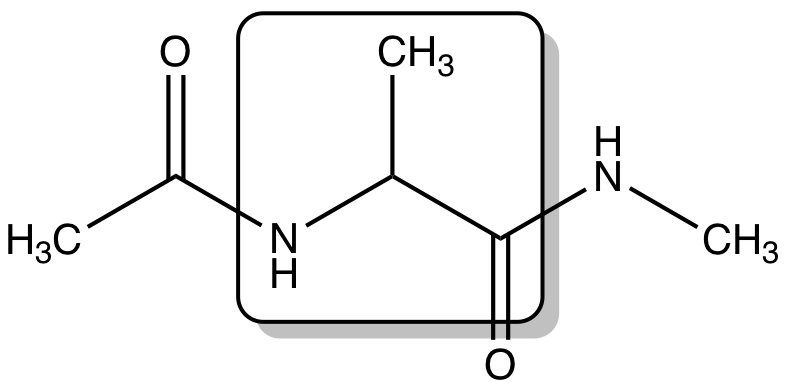}
\caption{The structure of alanine dipeptide. The alignment core refers to the heavy atoms inside the square box.}
\label{fig:aladip_structure}
\end{figure}

The vector representation of alanine dipeptide was defined by aligning the set of conformers to a minimum energy conformer calculated by density functional theory. Furthermore, we aligned each conformer to a core set of atoms. This flexibility is inherent within the vector representation of the conformational space. 

Persistent homology was calculated by using the Rips complex persistence on the Euclidean distance on the vector representation. Persistence was calculated on the vector space using all atoms, and also for heavy (non-hydrogen) atoms. These can be seen in Figure \ref{fig:aladip_vector_representations}.

\begin{figure}
\centering
\begin{subfigure}{.49\textwidth}
\includegraphics[width=\linewidth]{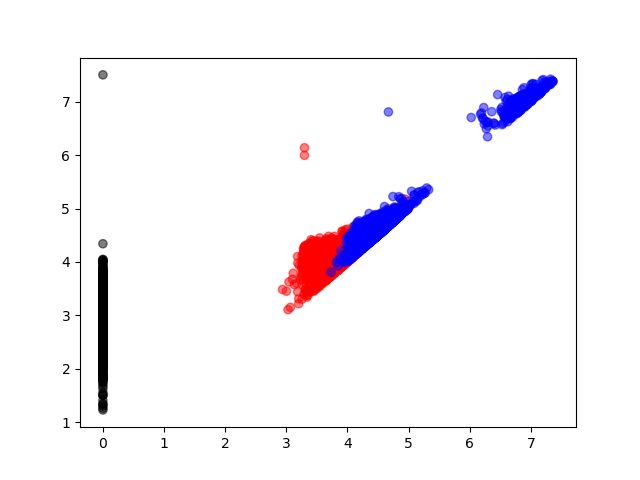}
\caption{All atoms}
\label{subfig:aladip_vector_all}
\end{subfigure}
\begin{subfigure}{.49\textwidth}
\includegraphics[width=\linewidth]{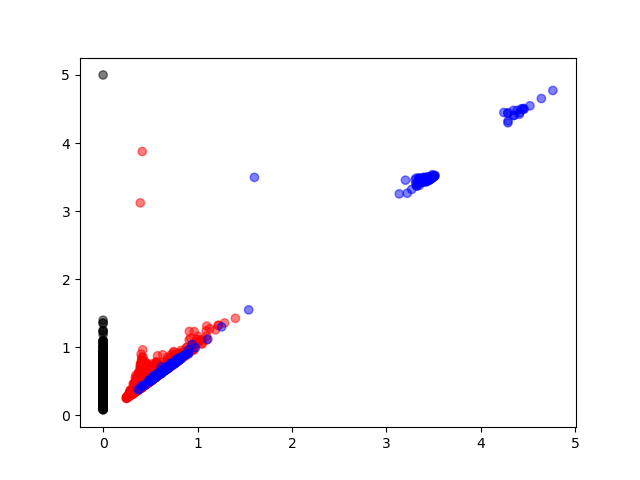}
\caption{Heavy atoms only}
\label{subfig:aladip_vector_heavy}
\end{subfigure}
\caption{Persistence of the vector space representation of alanine dipeptide}
\label{fig:aladip_vector_representations}
\end{figure}
Firstly, it is clear that both sub-representations have similar persistent Betti numbers of $(1,2,1)$. This matches those of a torus, as earlier predicted. However, these features appear at different times depending on representation. This can be explained when considering the behavior of the Euclidean metric as the number of dimensions increases. The all atom system is 66-dimensional, whereas the heavy atom system has only 30 dimensions. This leads to a shorter average distance between conformers in the conformational space in the heavy atom system. 

To create the RMSD representation for alanine dipeptide, the optimal alignment between every pair of conformers was found by optimising the RMSD between their atoms. This was calculated for both all atom and heavy atom sets. Persistent homology was then calculated using the Rips filtration on the optimum RMSD metric, with the resulting persistence diagram in Figure \ref{fig:aladip_RMSD_representations}.
\begin{figure}
\centering
\begin{subfigure}{.49\textwidth}
\includegraphics[width=\linewidth]{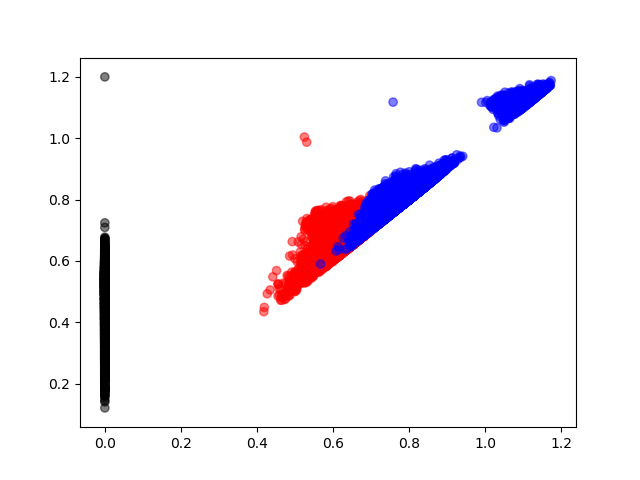}
\caption{All atoms}
\label{subfig:aladip_RMSD_all}
\end{subfigure}
\begin{subfigure}{.49\textwidth}
\includegraphics[width=\linewidth]{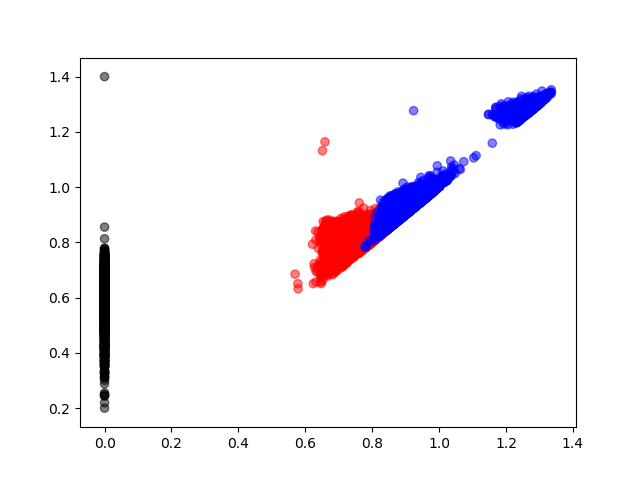}
\caption{Heavy atoms only}
\label{subfig:aladip_RMSD_heavy}
\end{subfigure}
\caption{Persistence of the RMSD representation of alanine dipeptide}
\label{fig:aladip_RMSD_representations}
\end{figure}
Again, we find similar persistent Betti numbers of (1,2,1). This suggests that the topology of our conformational space is independent of representation - however it will be shown that this is not always the case. Further, there is a much smaller difference in choice of atom subsets in the case of the RMSD representation. This is due to the difference in the behaviour of the RMSD metric itself. For example, features appear slightly earlier in the all atom system. This is because the average displacement of hydrogen atoms tends to be quite small, but the increase in the denominator of the RMSD metric causes a slightly lower metric. In the case of the vector representation, each hydrogen adds an extra 3 dimensions to the Euclidean distance - the RMSD does not suffer from this curse of dimensionality in the same way.

\subsubsection{Pentane}
The structure for pentane can be seen in Figure \ref{fig:pentane_structure}. Similarly to alanine dipeptide, there are two free torsions in pentane. For this section we are ignoring the symmetry of the pentane molecule, and therefore we expect the conformational space to have the topology of the torus $T^2$. 
\begin{figure}[h]
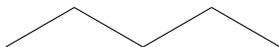

\centering
\chemfig{-[:30]-[:-30]-[:30]-[:-30]}
\caption{The structure of pentane}
\label{fig:pentane_structure}
\end{figure}

Similarly to alanine dipeptide, we define our vector representation by aligning each conformer to some reference, in this case a DFT optimised conformation. However, in contrast to alanine dipeptide, we do not align to a core, but instead align to minimise the RMSD between all carbon atoms. Persistence is then calculated analogously to alanine dipeptide.
The RMSD representation for pentane was defined by calculating the optimum RMSD distance between all carbon atoms for each pair of pentane conformers. Persistence was then calculated using this metric. The vector and RMSD representation persistent homology can be found in Figure \ref{fig:pentane_representations}. 
\begin{figure}
\centering
\begin{subfigure}{.49\textwidth}
\includegraphics[width=\linewidth]{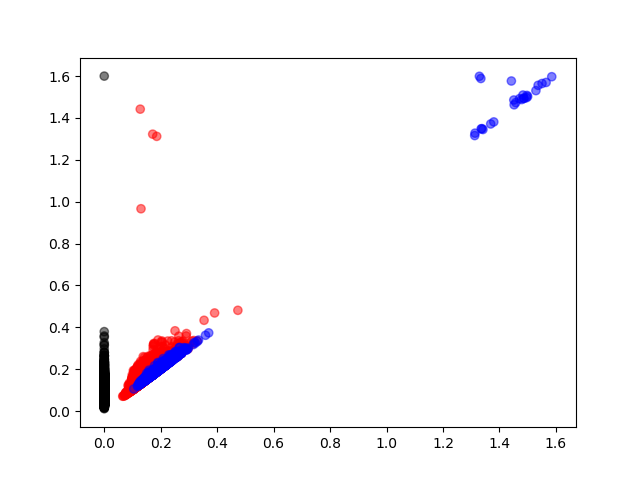}
\caption{Vector representation}
\label{subfig:pentane_vector}
\end{subfigure}
\begin{subfigure}{.49\textwidth}
\includegraphics[width=\linewidth]{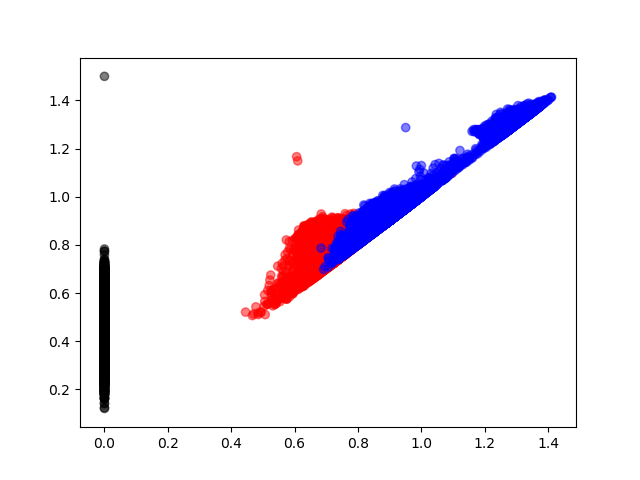}
\caption{RMSD representation}
\label{subfig:pentane_RMSD}
\end{subfigure}
\caption{Persistence of the different representations of the pentane conformational space. Symmetry is ignored for these representations}
\label{fig:pentane_representations}
\end{figure}

There is a clear difference in the persistent homology for these two representations. Whereas the RMSD representation has persistent Betti numbers of (1,2,1), those of the vector representation are far more unclear. This makes clear one of the main drawbacks of the vector representation of conformational spaces, in that they are dependent on the alignment to a reference. If we had aligned to some core of the three central carbons, we would have seen persistent Betti numbers of (1,2,1), as we did for alanine dipeptide. Furthermore, although aligning each conformer to some reference may make it easier to visualise the entire set at once (as is often done in the study of proteins etc.), it makes pairwise comparisons of different conformers hard to perform. In contrast, the RMSD representation does not suffer from this, and therefore leads to the correct persistent Betti numbers. Furthermore, as we will now see, the RMSD representation allows us to directly take into account molecular symmetry - which would be impossible for the vector representation.

The molecular graph of pentane has an inherent symmetry, as discussed earlier. In particular, the two torsions are equivalent, rather than being distinguishable. It is a standard result that leads to a \mobius band topology.

We can take this symmetry into account when calculating the RMSD metric between two conformers. This is done by performing two separate RMSD alignments. In the first, we align the two conformers such that each carbon in the first conformer is matched to the same carbon in the second conformer. In the second, we match each carbon in the first conformer to the opposite carbon in the second. We then choose the optimum RMSD to be our metric. The resulting persistent homology can be seen in Figure \ref{fig:rmsd_permute}.
\begin{figure}[h]
\centering
\includegraphics[width=0.49\textwidth]{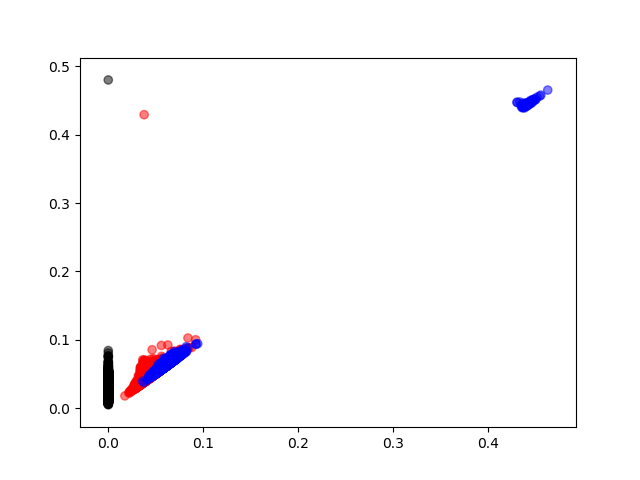}
\caption{Persistence of the RMSD representation of pentane conformational space, with symmmetry taken into account}
\label{fig:rmsd_permute}
\end{figure}
There is a significant difference in the persistent homology when compared to the original RMSD representation. In particular, the persistent Betti numbers are now (1,1,0) matching those of the \mobius band. 

%\lee{There feels like there ought to be more to talk about here - any help?}

\subsection{ Cyclooctane}
\subsubsection{Finding singular points and clustering}
%\lee{Can one of you write the procedure for finding singular points here, or at least a comment as to where someone should look?}\ingrid{I will, this will be part of local dimension analysis} \lee{Can the local dimension analysis go first then?}

Once singular points were identified, they were removed from the data set. In principle, this separated the data into its manifold components. These can then be found using a clustering algorithm, in our work, HDBScan \cite{Campello2013}. Subsequently, the manifolds were matched, in an attempt to recreate the spherical and Klein bottle components found in the original work. This could then be verified using persistence.

We used the software \emph{Ripser} \cite{ripser1} to compute the persistent homology of our point data sampled from the conformational spaces. %\lee{I might move this to a brief section at the start of the results, talking about software used}

To verify that we had correctly found the spherical component, we calculated the persistent homology of the Rips complex constructed on the RMSD metric between conformers. The resulting persistence diagram can be seen in Figure \ref{fig:CO_S2}. We can see that the persistent Betti numbers are (1,0,1), as expected for a sphere.

\begin{figure}[h]
\centering
\includegraphics[width=0.49\textwidth]{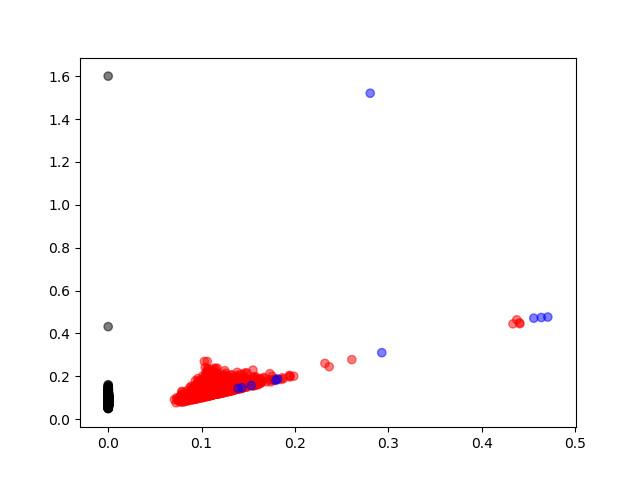}
\caption{Persistence of the RMSD representation of the spherical component of the cyclooctane conformational space}
\label{fig:CO_S2}
\end{figure}

Verifying the presence of the Klein bottle component is slightly more involved. Here, we perform persistent homology calculations in the same manner as before. However, we calculate homology over two different fields of coefficients, namely $\mathbb{Z}_2$ and $\mathbb{Z}_3$. The Klein bottle has different (persistent) Betti numbers over these fields, (1,2,1) and (1,1,0) respectively - a torus would not. This allows us to verify with more confidence the presence of a Klein bottle component. The persistence diagrams can be seen in Figure \ref{fig:CO_KB}. The correct persistent Betti numbers are found.

\begin{figure}
\centering
\begin{subfigure}{.49\textwidth}
\includegraphics[width=\linewidth]{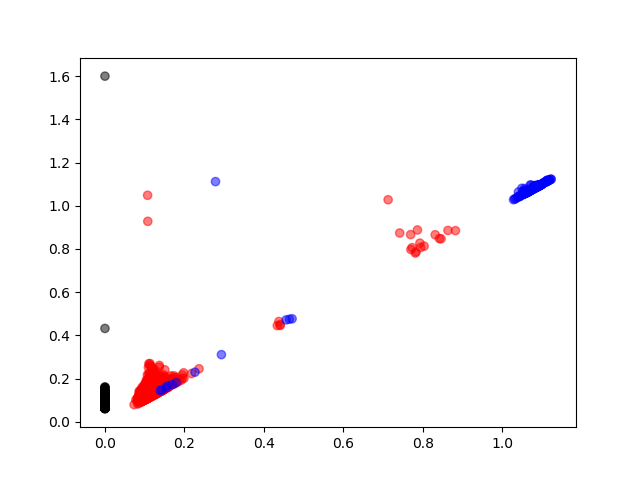}
\caption{$\mathbb{Z}_2$}
\label{subfig:KB_mod2}
\end{subfigure}
\begin{subfigure}{.49\textwidth}
\includegraphics[width=\linewidth]{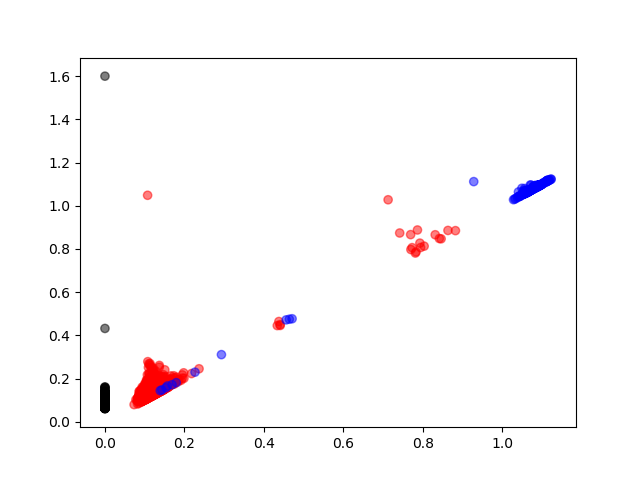}
\caption{$\mathbb{Z}_3$}
\label{subfig:KB_mod3}
\end{subfigure}
\caption{Persistence diagrams to verify the presence of a Klein bottle component to the cyclooctane conformational space.}
\label{fig:CO_KB}
\end{figure}

\subsection{Energy landscapes}
%note = "\url{https://topology-tool-kit.github.io/}"

%\ingrid{I feel like we need to explain in terms of physics or chemistry what  the PES is and how it is related to the conformational sapces. Maybe the force field equations or so?? }

In order to compute the Morse-Smale complexes of the energy landscapes, we made use of the Topology ToolKit (ttk) \cite{ttk1}, a software platform designed for the topological analysis of scalar data. We also used Matlab's \emph{alphaShape} function for initial processing, which produces an alpha-shape triangulation from a point cloud, of a specified radius. Extracting the boundary of this triangulation gave us the surface triangulation we required. This, together with the energy values at each point was input into the Topology ToolKit software.

We analysed the potential energy landscapes of cyclooctane, alanine dipeptide, pentane and fluoropentane. The conformational spaces of alanine dipeptide, pentane and fluoropentane are all tori, which makes for a side-by-side comparison of different energy landscapes on what is topologically the same conformational space. There is also an analysis of the free energy landscape of alanine dipeptide.
%\subsubsection{Cyclooctane}

\

The results for cyclooctane are compared with the results from \cite{Martin2011}.
After finding the singular points, and separating out the sphere and Klein bottle components, we did separate analyses of the energy landscapes on these two components. The points sampled from a sphere, an orientable, low-dimensional manifold, can be triangulated so that the resulting simplicial complex has the topology of a sphere. This simplicial complex can then be input into the \emph{ttk} software and filtered by the energy function.

We use the connection between persistent homology and discrete Morse theory, to smooth this energy surface by removing topological features below a certain persistence threshold. In order to do this, we compute the persistence diagram, as well as a statistical summary of it, called the \emph{ttkPersistenceCurve} in the \emph{ttk} software, which plots the distance from the diagonal against the number of persistent points in the persistence diagram. Depicting this curve in the log scale allows us to estimate a sensible level of noise, by observing a change in the gradient of this curve. In Figure \ref{cycloDenoising}, this computation is shown for the spherical part of cyclooctane, while Figure \ref{cyclo} shows the computed Morse-Smale complex.

\begin{figure}
\centering
\begin{subfigure}{.32\textwidth}
\includegraphics[width=\linewidth]{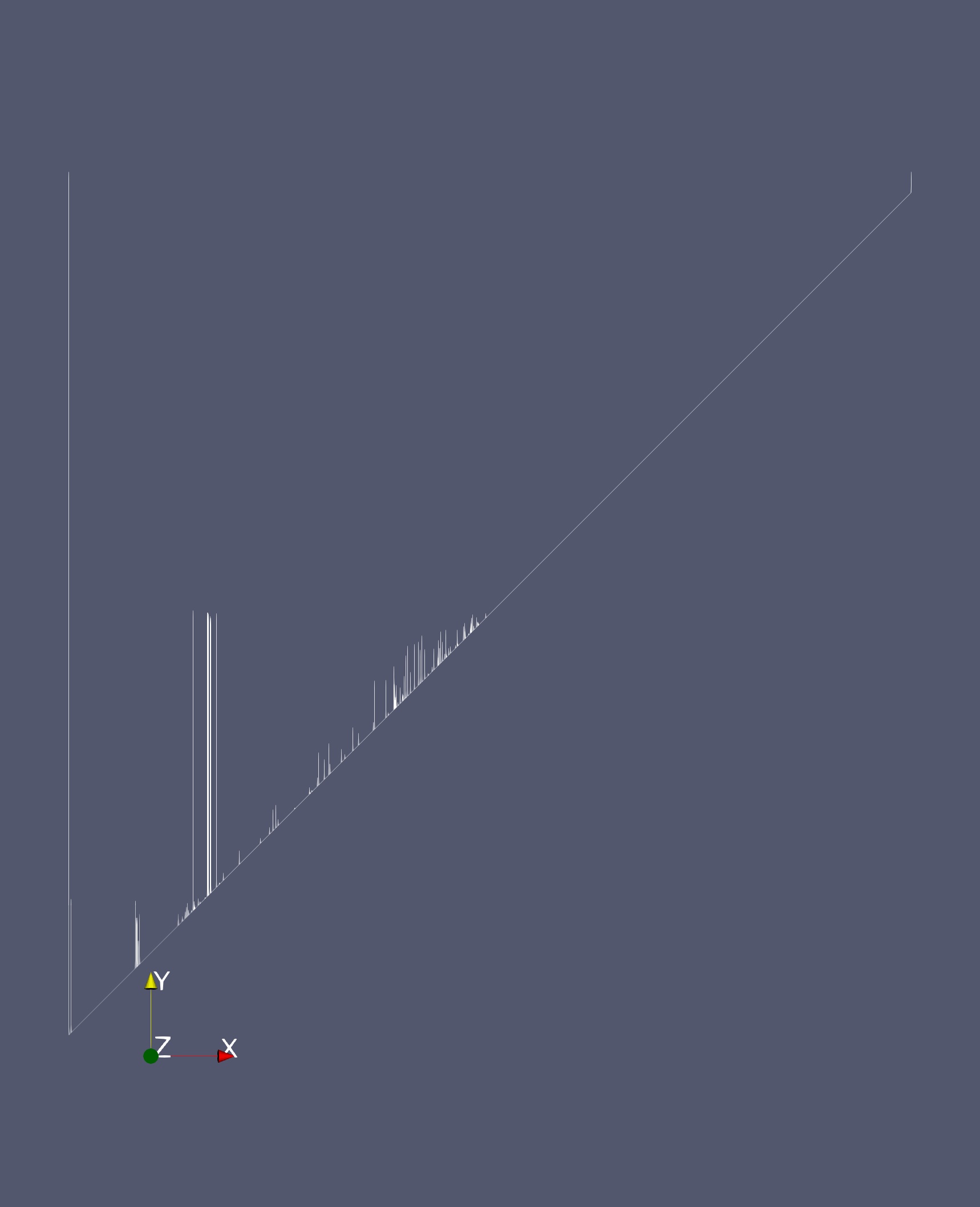}
\caption{}
\end{subfigure}
\begin{subfigure}{.32\textwidth}
\includegraphics[width=\linewidth]{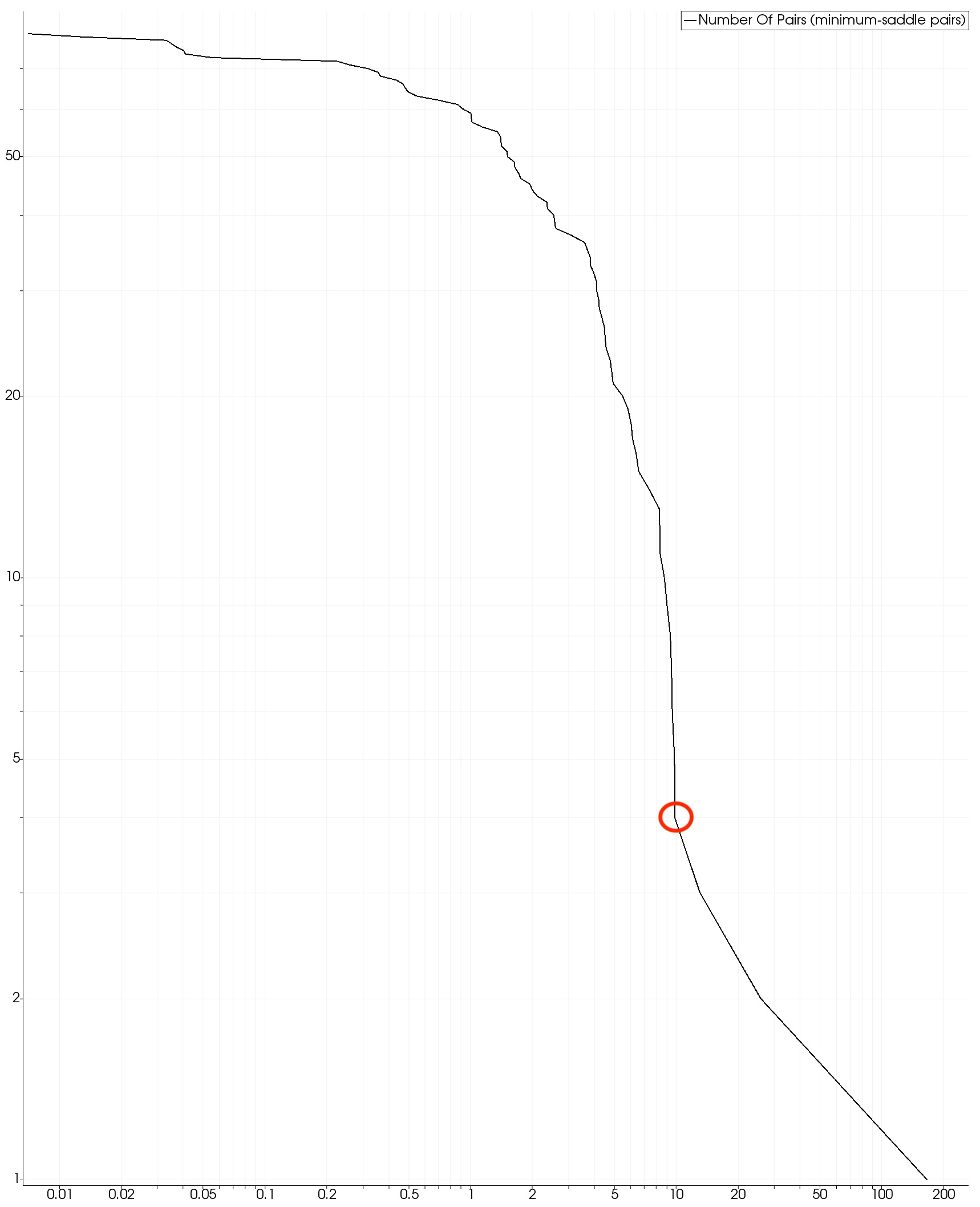}
\caption{}
\end{subfigure}
\begin{subfigure}{.32\textwidth}
\includegraphics[width=\textwidth]{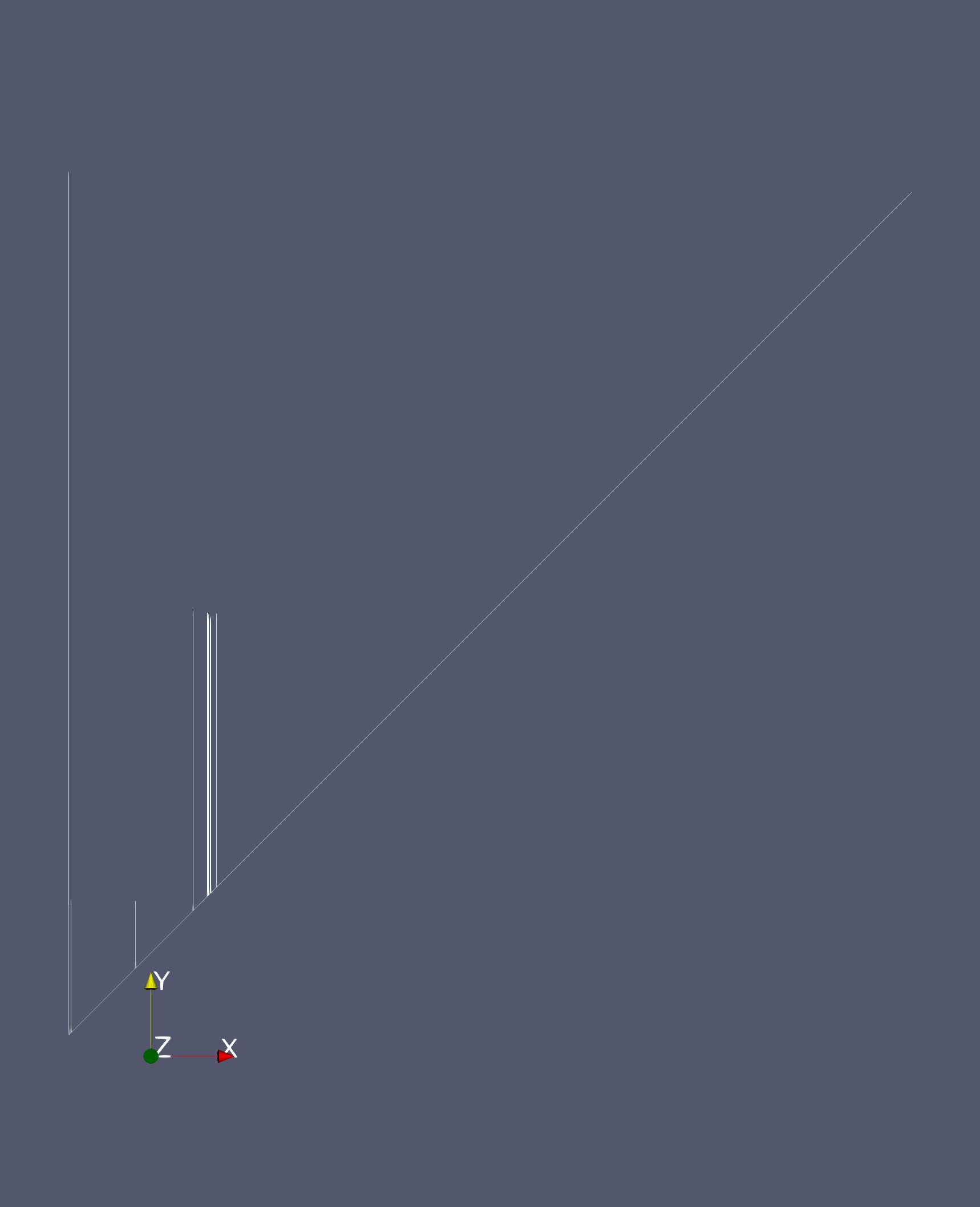}
\caption{}
\end{subfigure}
\caption{The noise detection for the spherical part of cyclooctane, computed using the \emph{ttk} software. The first figure (a) shows the zero-dimensional persistence diagram, the second (b) shows the persistence curve with the likely persistence threshold for noise circled in red, and the final (c) image shows the denoised persistence diagram, using the threshold discovered through the persistence curve.}
\label{cycloDenoising}
\end{figure}

\begin{figure}[ht]
\centering
\begin{subfigure}{.32\textwidth}
\includegraphics[width=\linewidth]{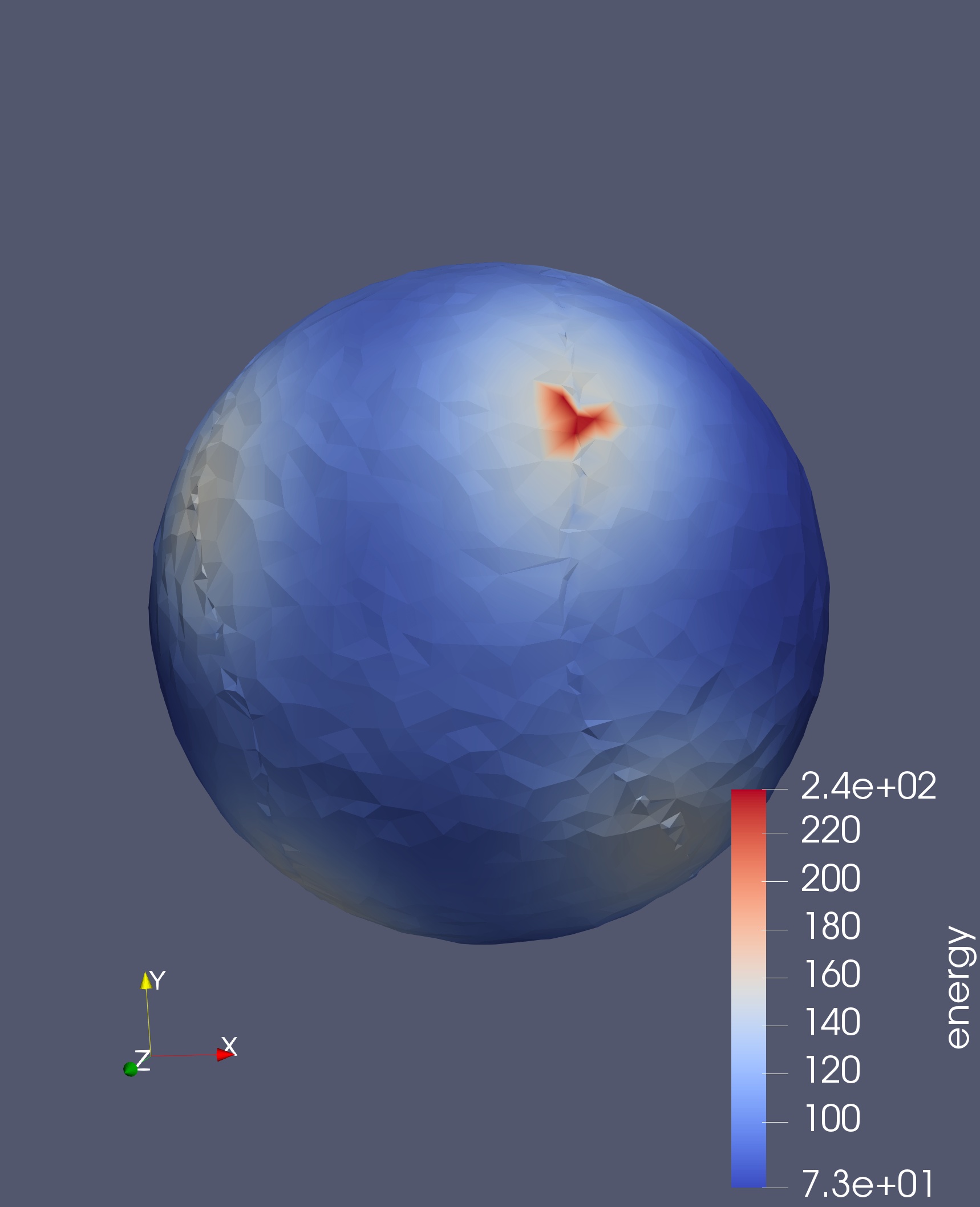}
\caption{}
\end{subfigure}
\begin{subfigure}{.32\textwidth}
\includegraphics[width=\linewidth]{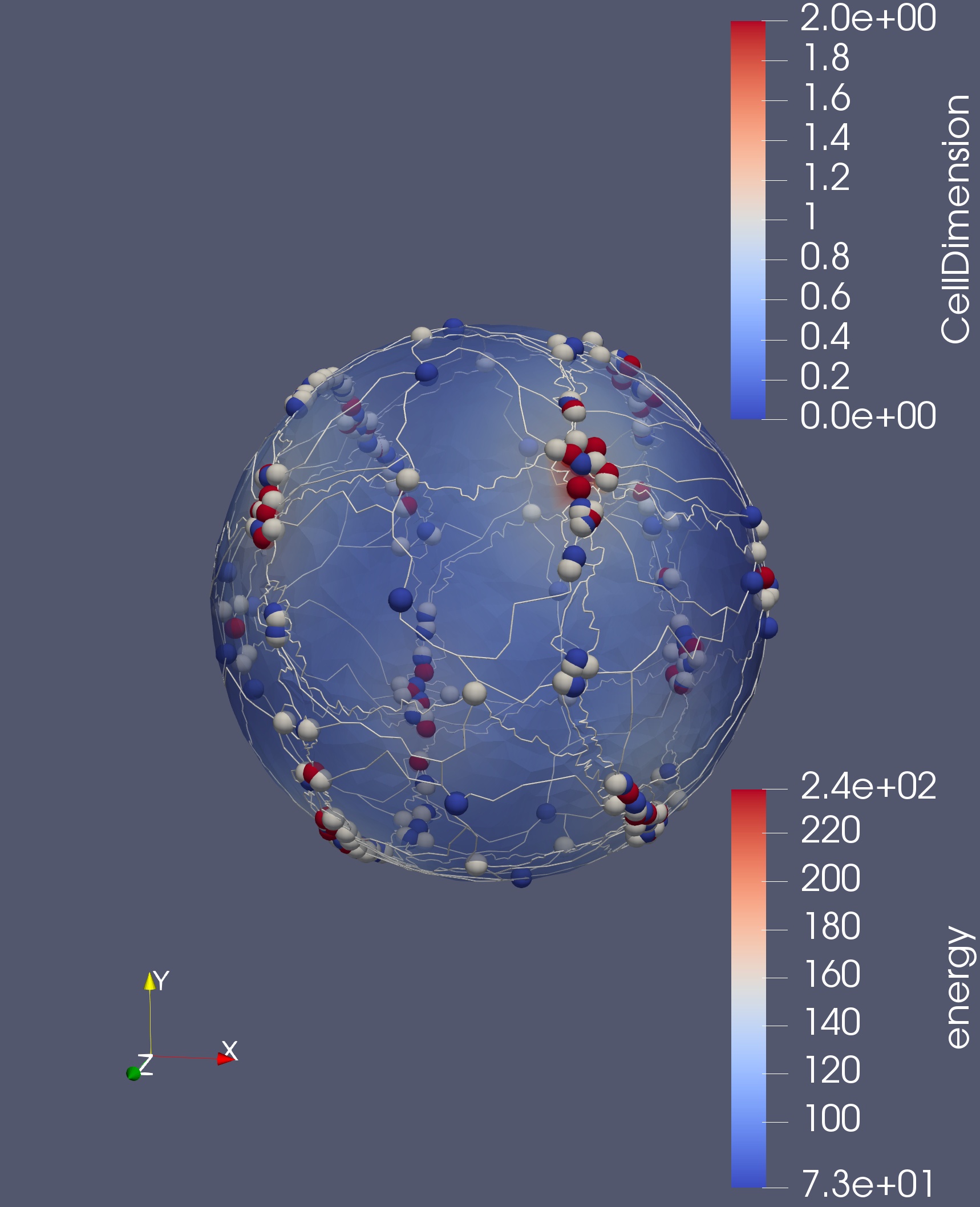}
\caption{}
\end{subfigure}
\begin{subfigure}{.32\textwidth}
\includegraphics[width=\textwidth]{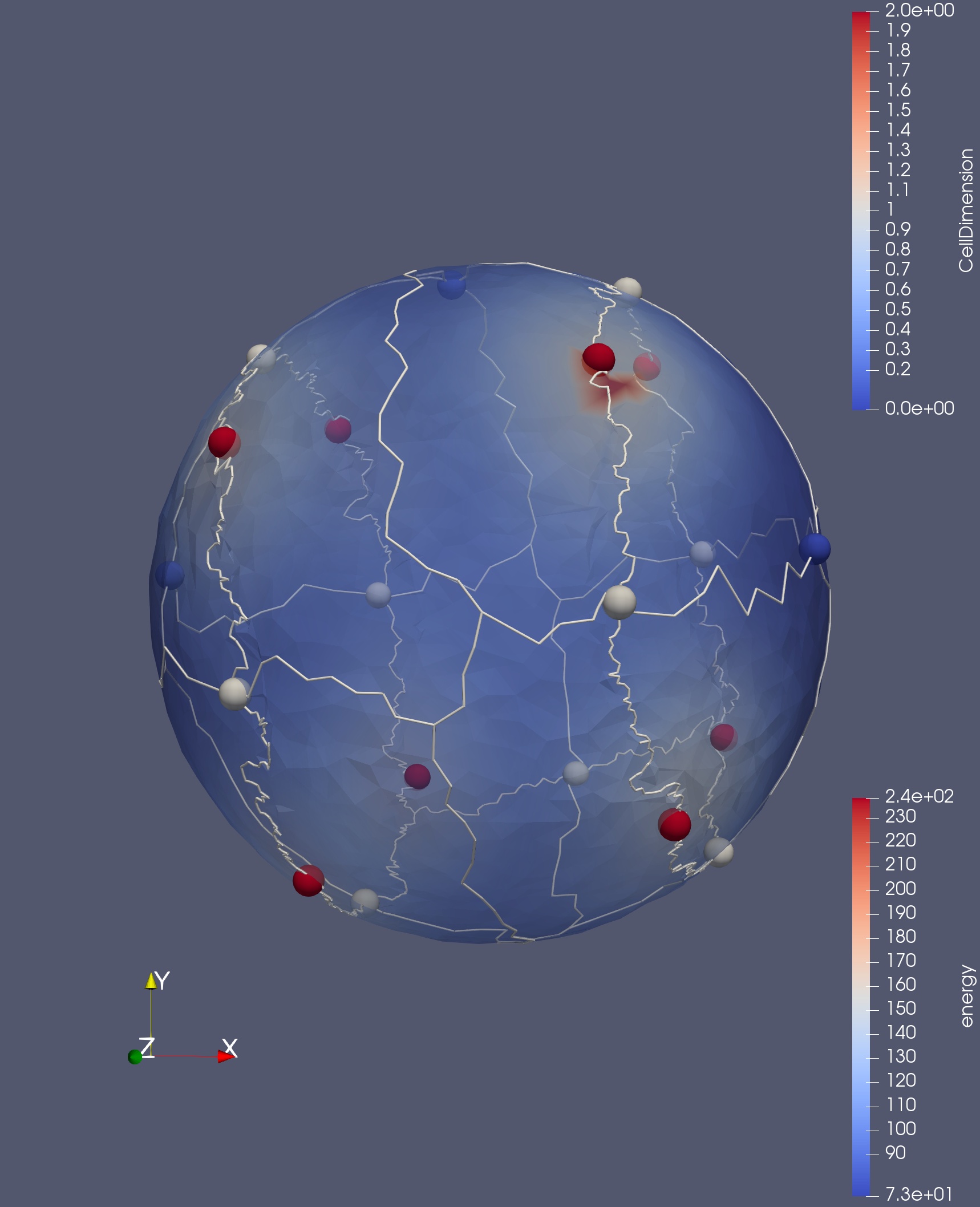}
\caption{}
\end{subfigure}
\caption{The Morse-Smale complex for the spherical component of cyclooctane, computed using the \emph{ttk} software. The first figure (a) shows the potential energy function. The second figure (b) shows the Morse-Smale complex without topological simplification and after simplification (c). The red points are the maxima, the white points are saddle points and the blue points are the energy minima.}
\label{cyclo}
\end{figure}

\

%\subsubsection{Alanine Dipeptide}
For alanine dipeptide, we consider the model of the conformational space where only the two torsional angles are allowed to rotate, giving a torus. The scalar values of the potential energy are then given on this two-dimensional surface, displayed in Figure \ref{alanine}.

\begin{figure}[ht]
\centering
\begin{subfigure}{.24\textwidth}
\includegraphics[width=\linewidth]{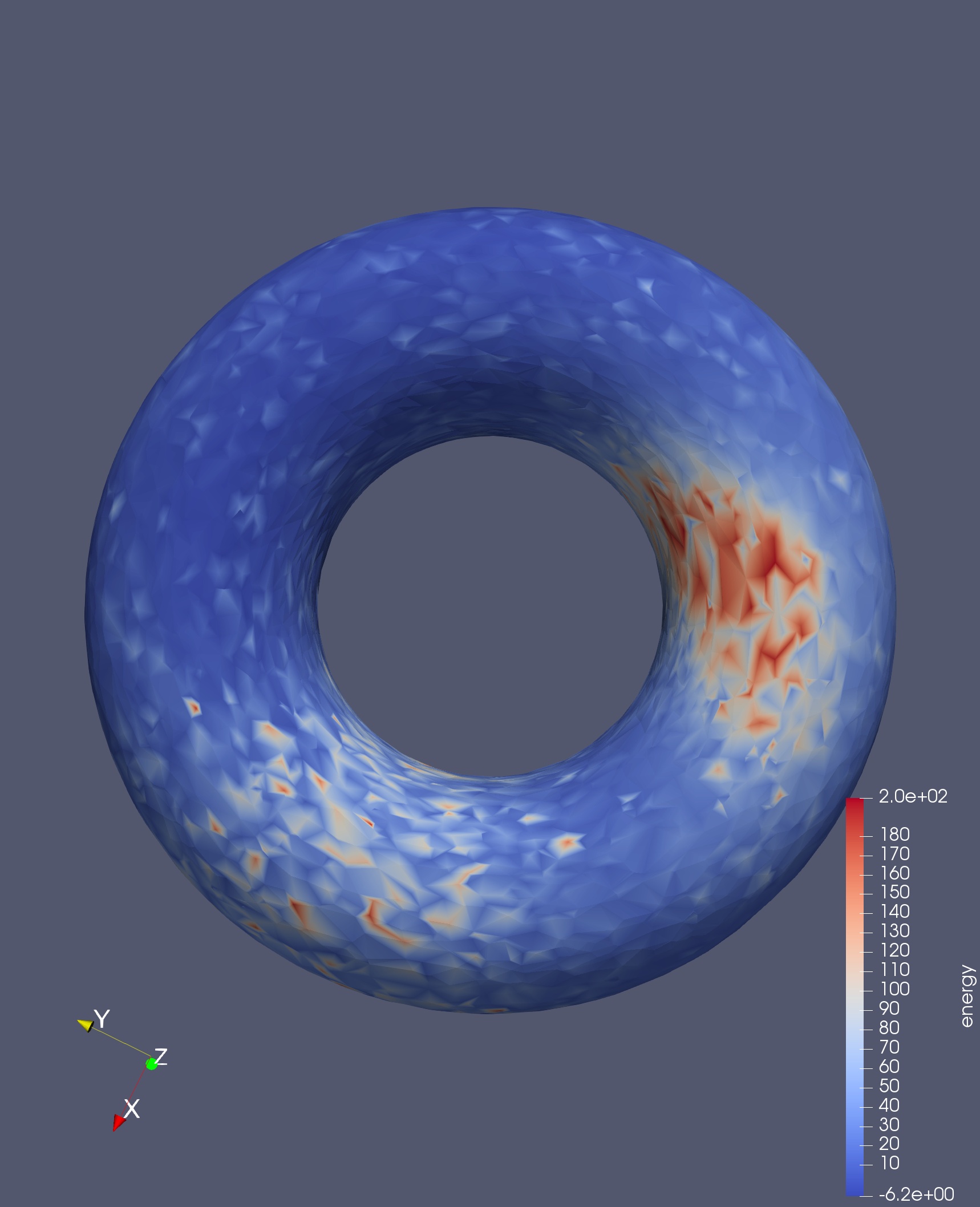}
\caption{}
\end{subfigure}
\begin{subfigure}{.24\textwidth}
\includegraphics[width=\linewidth]{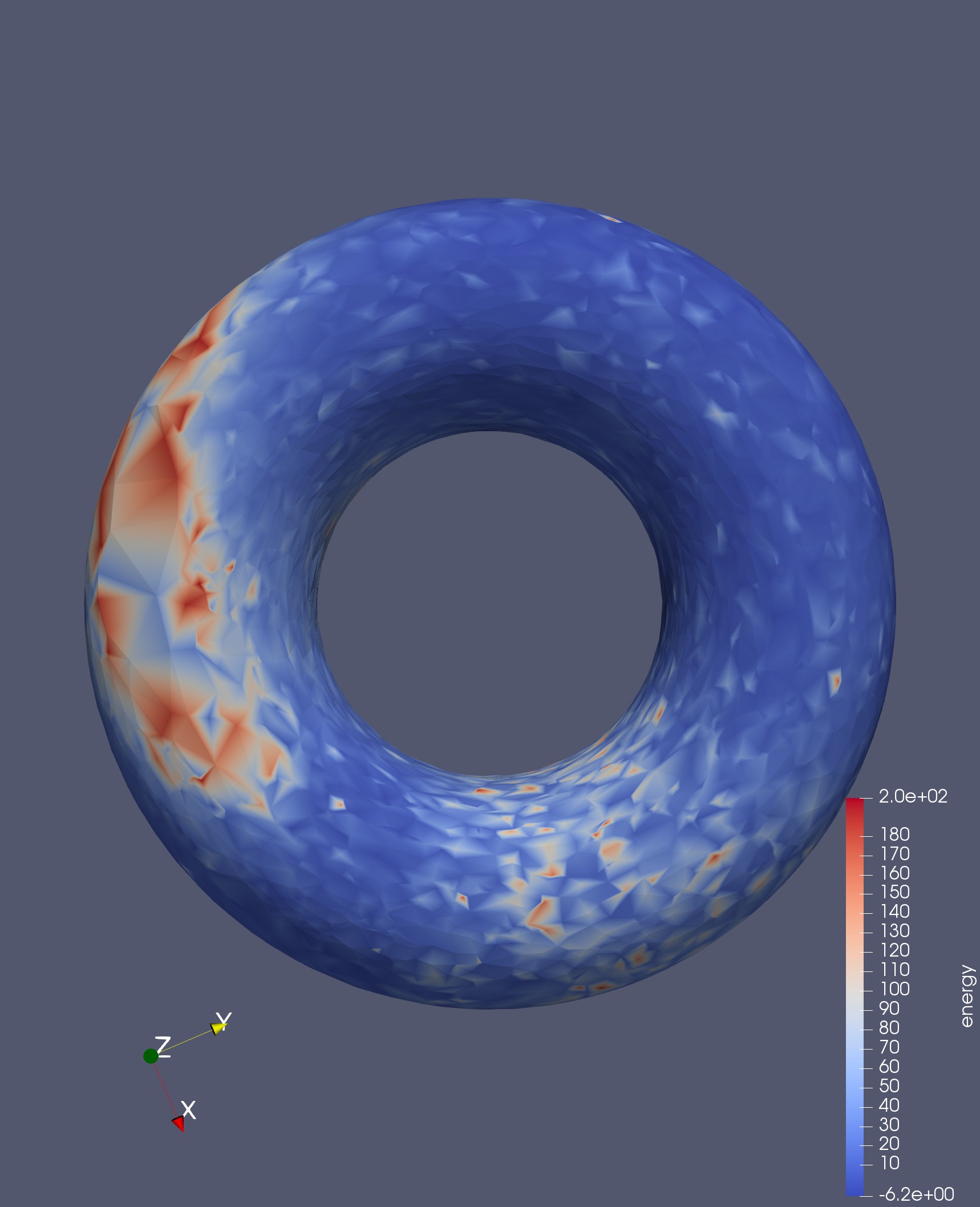}
\caption{}
\end{subfigure}
\begin{subfigure}{.24\textwidth}
\includegraphics[width=\textwidth]{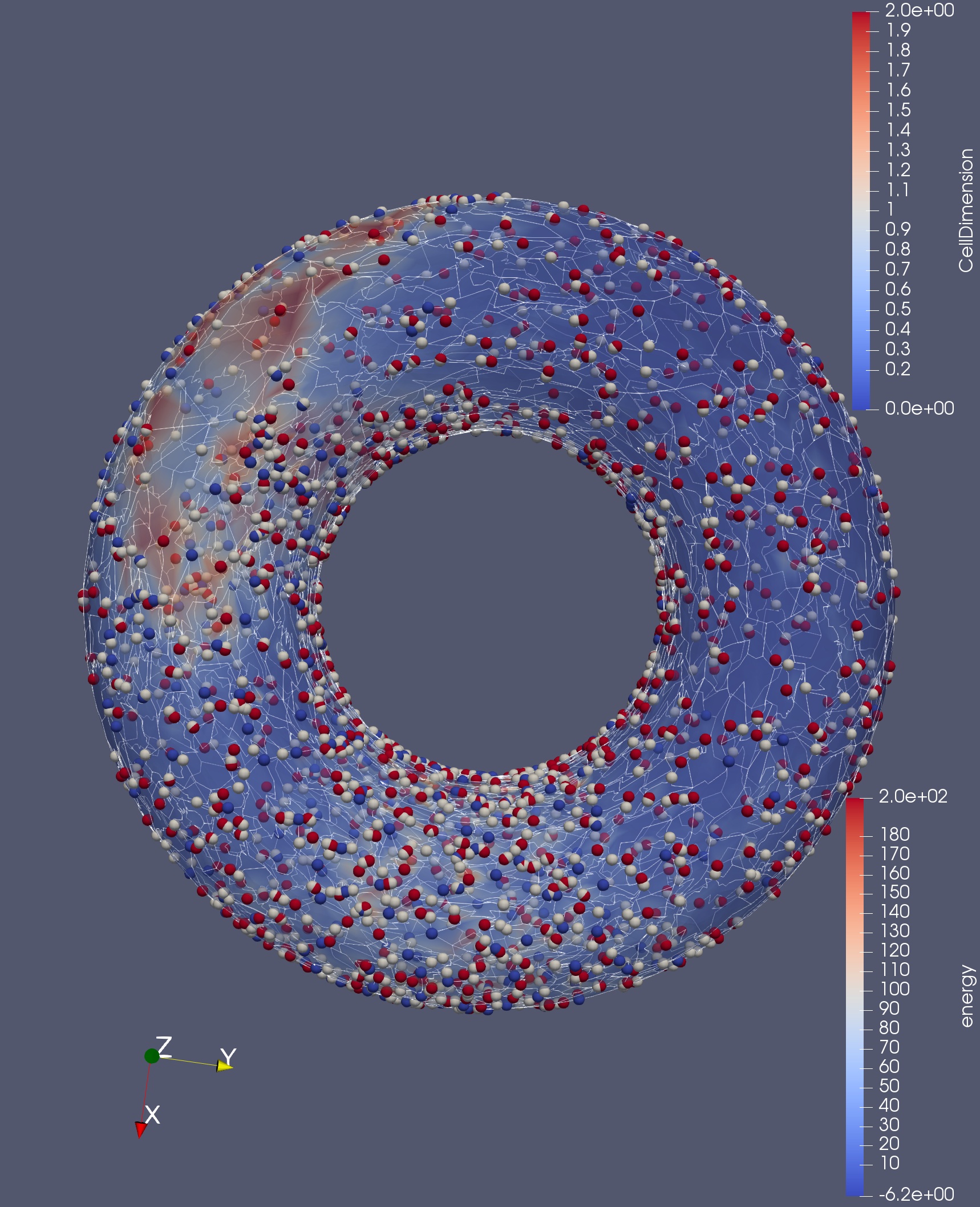}
\caption{}
\end{subfigure}
\begin{subfigure}{.24\textwidth}
\includegraphics[width=\textwidth]{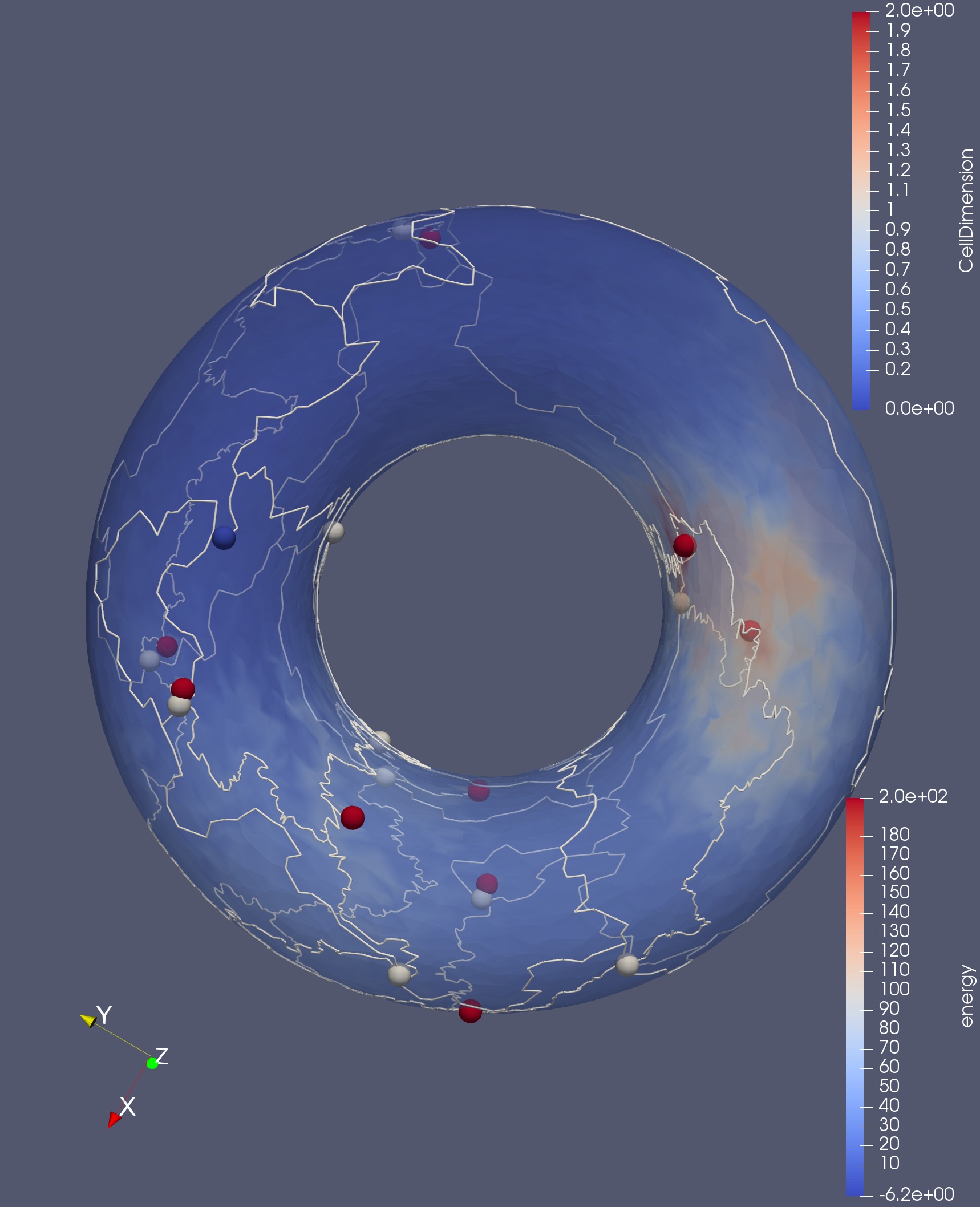}
\caption{}
\end{subfigure}
\caption{The Morse-Smale complex for Alanine Dipeptide, computed using the \emph{ttk} software. The figures (a) and (b) show the potential energy function from both sides of the torus. The image (c) shows the Morse-Smale complex, together with the Morse function after very minor topological simplification. The red points are the maxima, the white points are saddle points and the blue points are the energy minima. Finally, image (d) shows the Morse-Smale complex after severe topological simplification, leaving only one minimum.}
\label{alanine}
\end{figure}

We then did the same calculation for the free energy landscape. The results are displayed in Figure \ref{freeEalanine}. The obvious conclusion is that the free energy landscape is significantly less noisy. Another difference is the change in the number of minima. There are far more local maxima and saddle points in the potential energy landscape than in the free energy landscape.
%After minimal smoothing, one is only left with one minimum.

%\mariam{Lee, maybe a sentence of two about potential chemical conclusions we can make would be nice.}
%\lee{Can you give me the conformer IDs of minima, maxima and saddle points? The chemical conclusions would be that they match the ones found with other methods. In fact, just the $(\phi,\psi)$ angles would do (and may in fact be better!). Also - would it be too much to talk about the paths between the critical points found by this decomposition?}
%\mariam{Sounds good. Will try to say something.}

\begin{figure}
\centering
\begin{subfigure}{.245\textwidth}
\includegraphics[width=\linewidth]{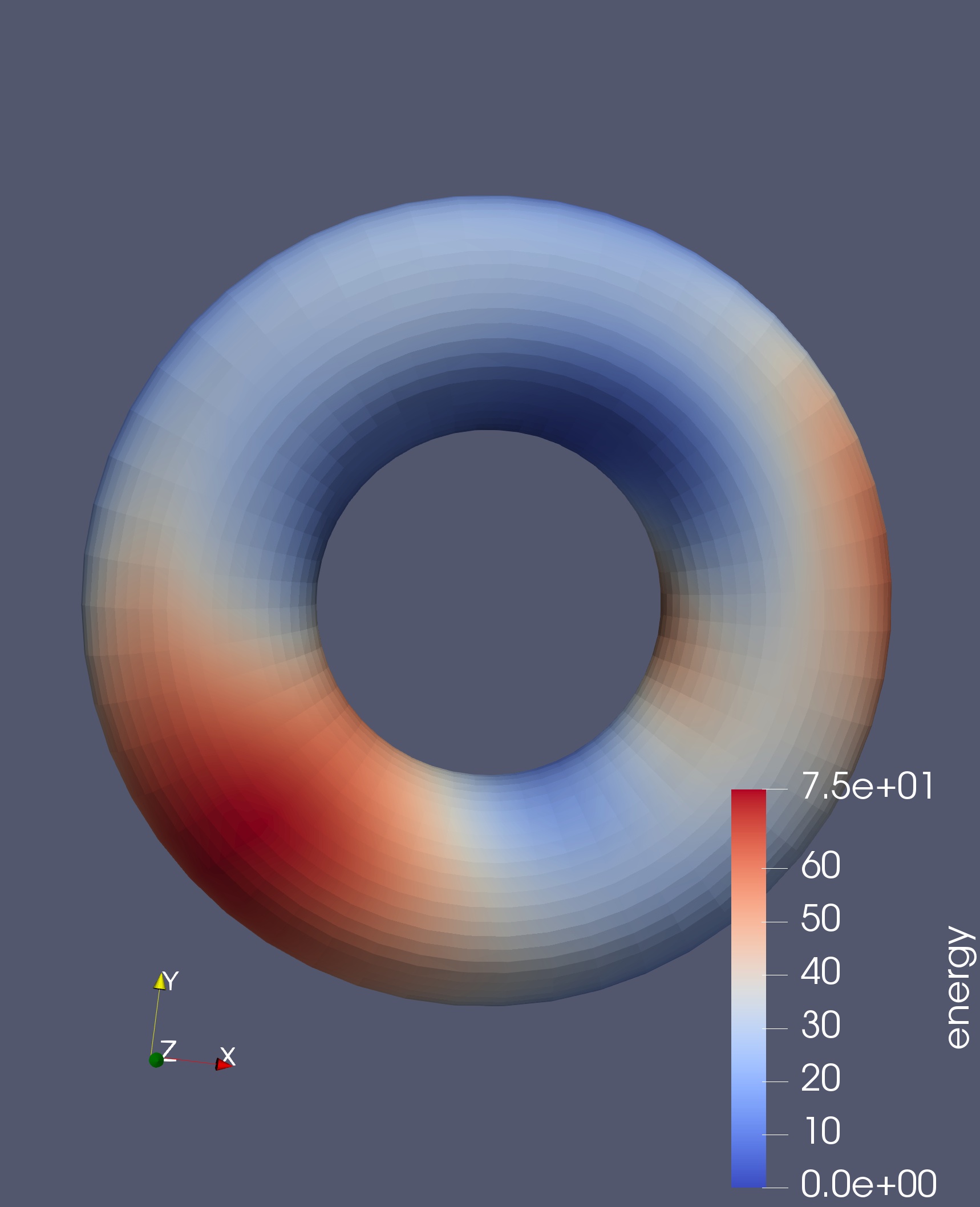}
\caption{}
\end{subfigure}
\begin{subfigure}{.245\textwidth}
\includegraphics[width=\linewidth]{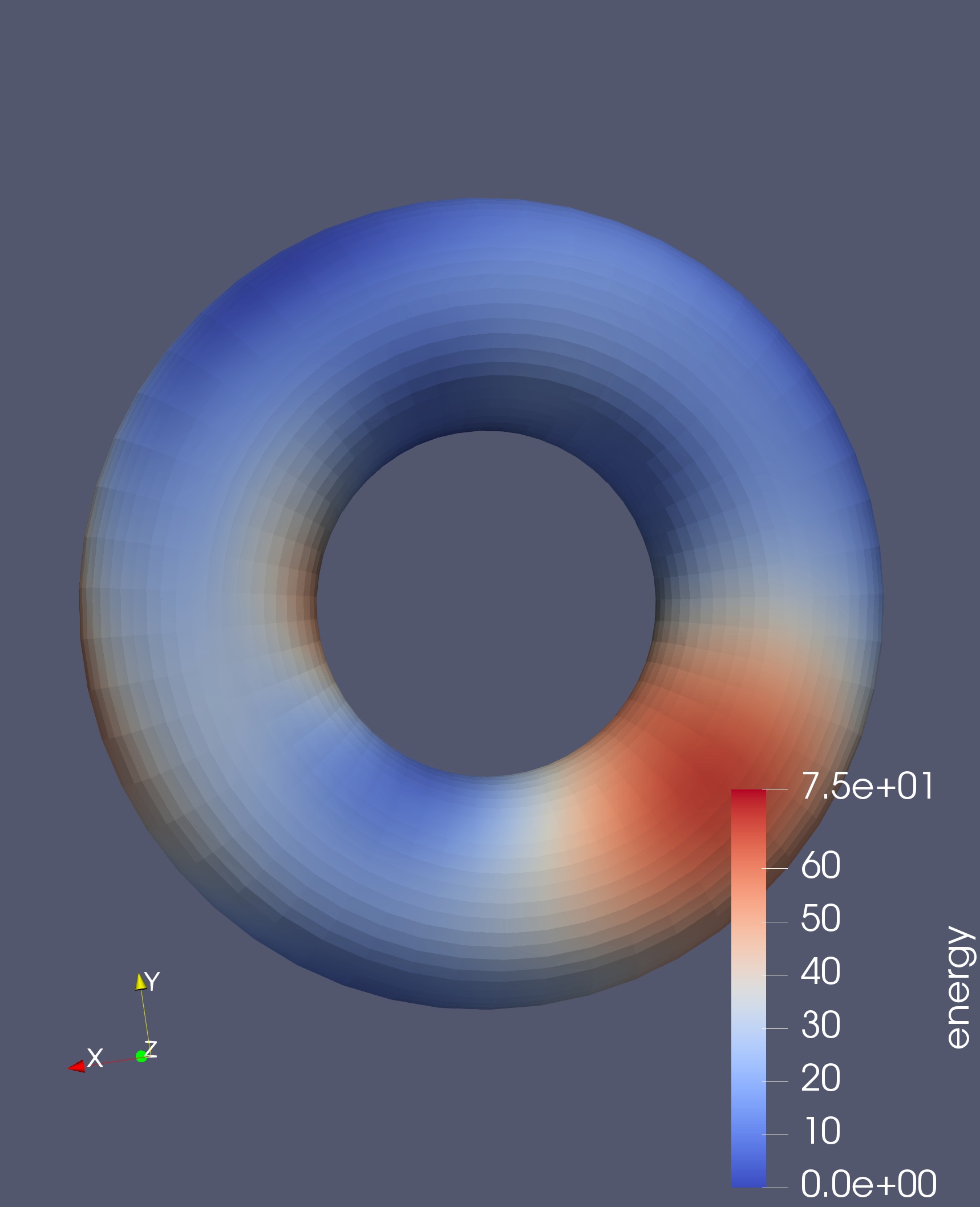}
\caption{}
\end{subfigure}
\begin{subfigure}{.245\textwidth}
\includegraphics[width=\linewidth]{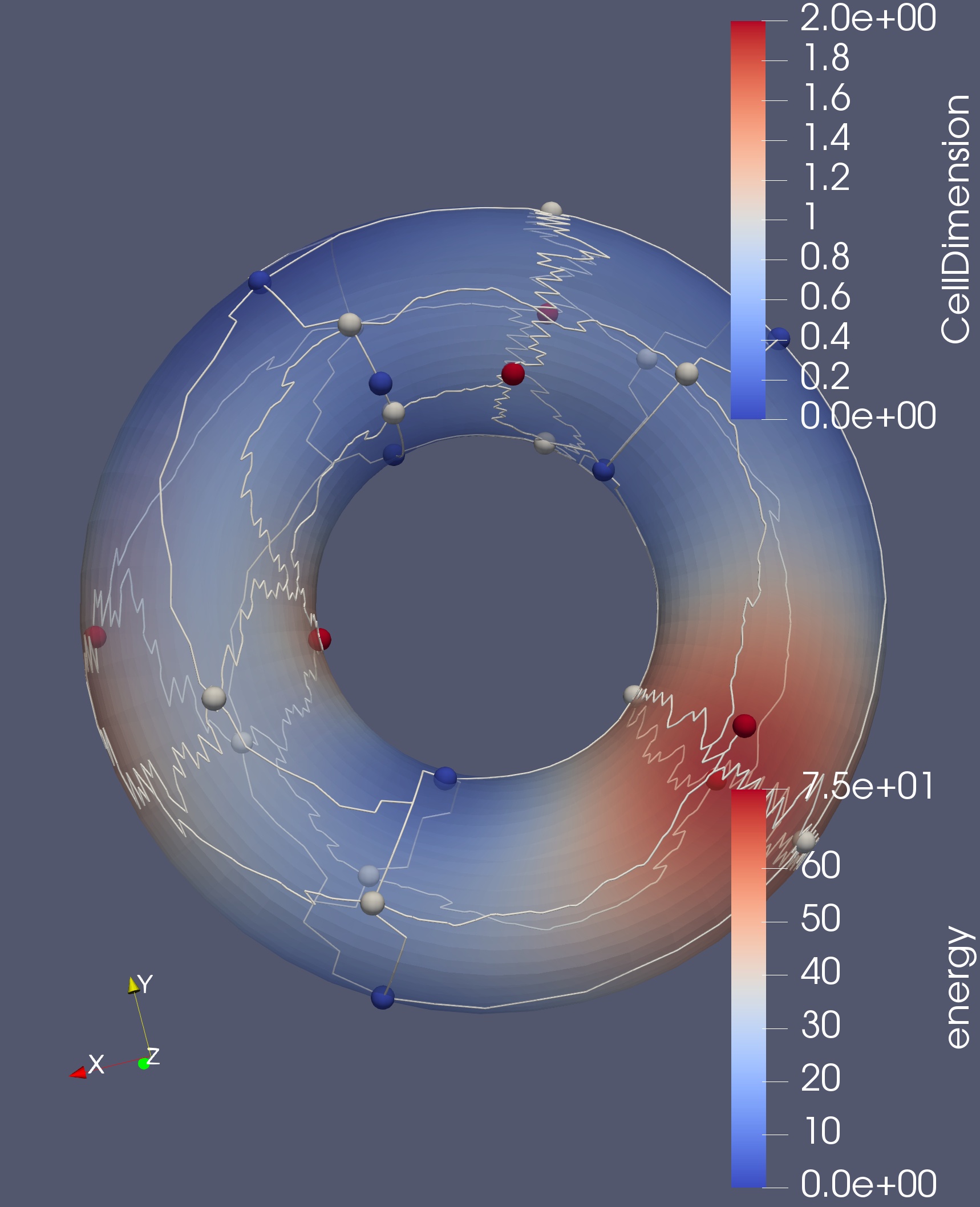}
\caption{}
\end{subfigure}
\begin{subfigure}{.245\textwidth}
\includegraphics[width=\textwidth]{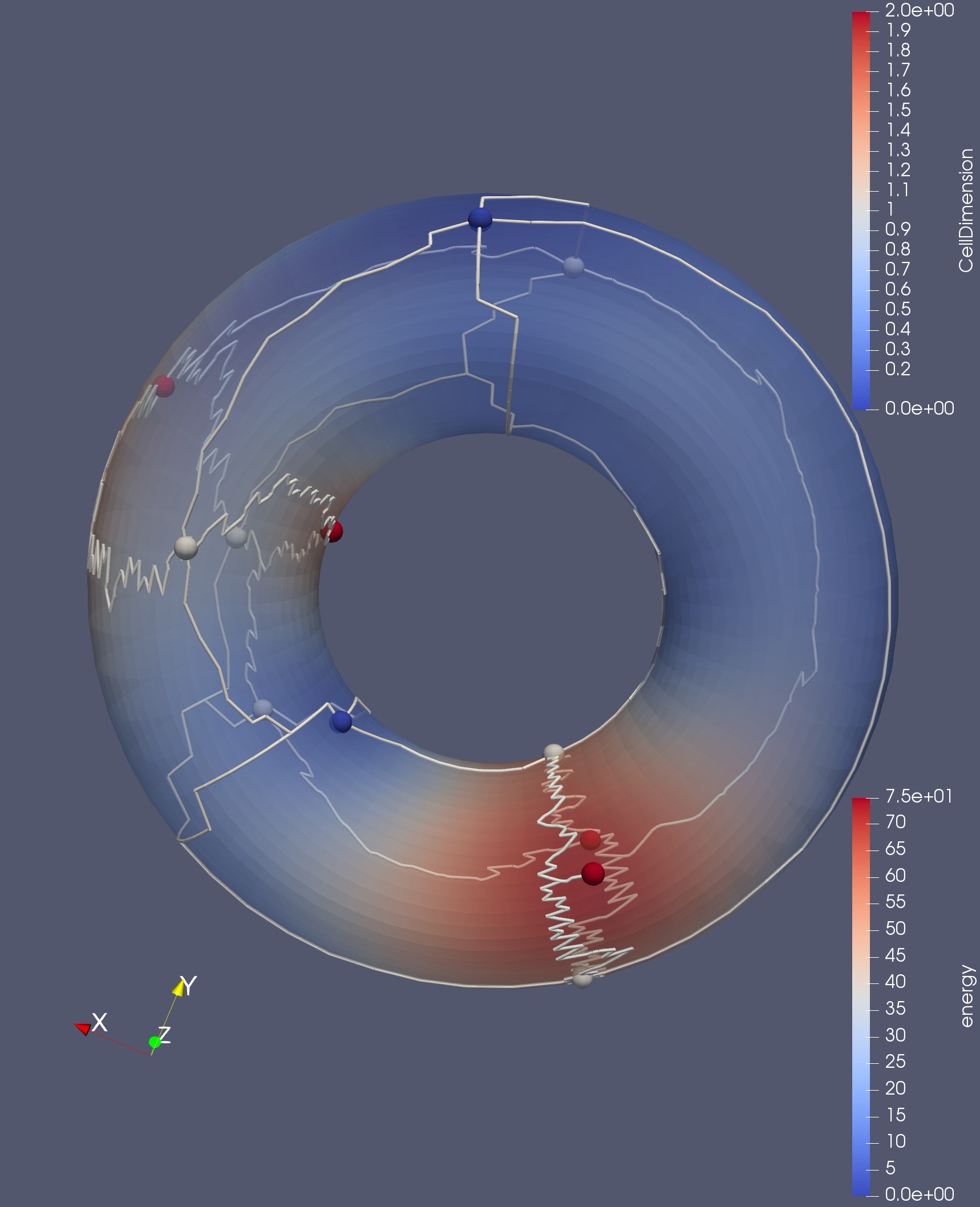}
\caption{}
\end{subfigure}
\caption{The analysis of the free energy surface for alanine dipeptide, computed using the \emph{ttk} software. The first two figures (a) and (b) %\jacek{Are these two views from the top and the bottom of the torus?}
are showing the energy function from both sides of the surface of the torus, while (c) and (d) are showing the Morse-Smale complex, together with the energy function, before and after topological simplification.}
\label{freeEalanine}
\end{figure}

\

%\subsubsection{Pentane}
For pentane, the conformational space is a torus as well. The potential energy surface was analysed in the same fashion as for alanine dipeptide. In Figure \ref{pentane}, the results of this analysis are depicted. The symmetry of the molecule can be observed in the energy landscape as well, in contrast to the energy landscape of alanine dipeptide. After thresholding, only one minimum remains.
\begin{figure}
\centering
\begin{subfigure}{.32\textwidth}
\includegraphics[width=\linewidth]{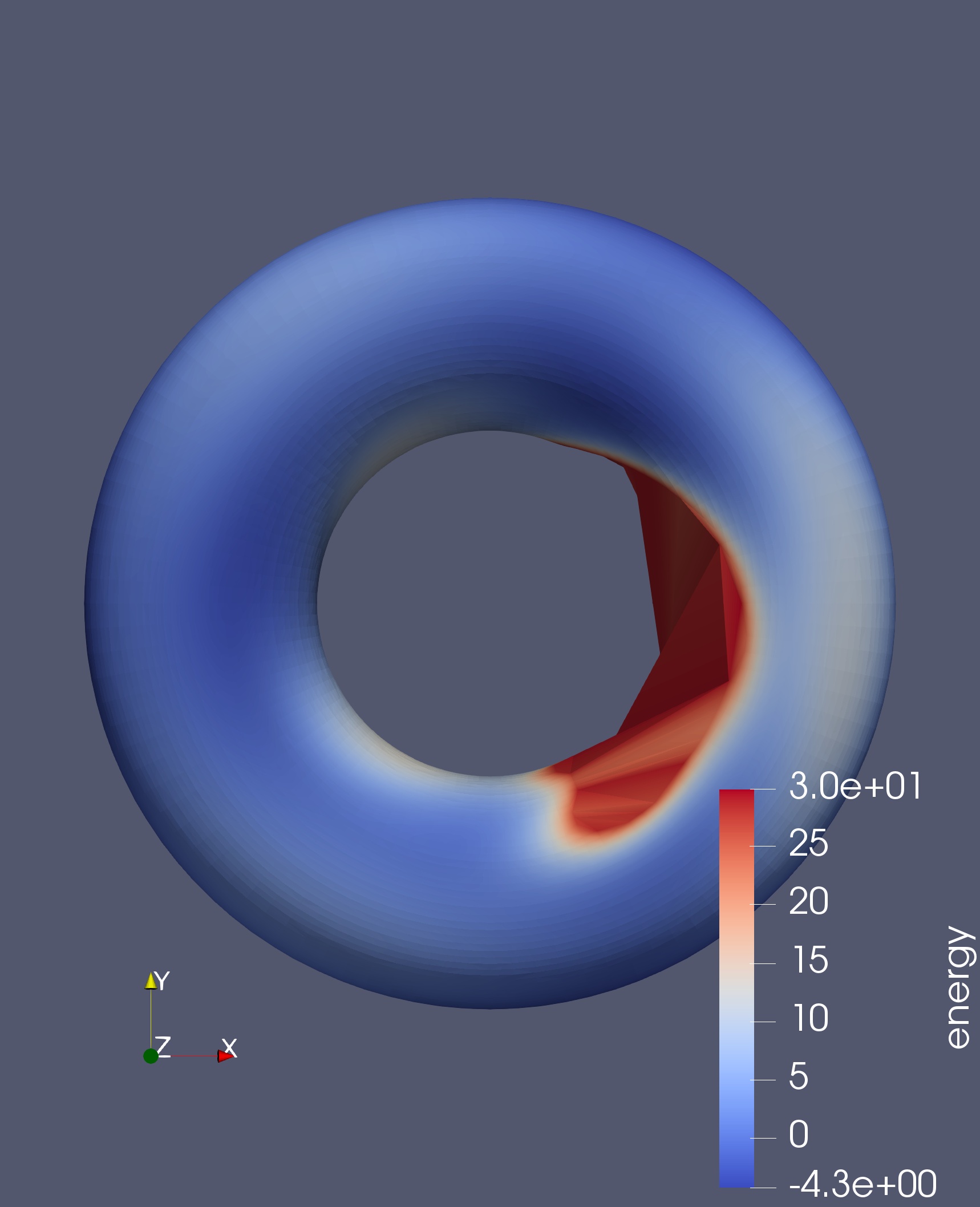}
\caption{}
\end{subfigure}
\begin{subfigure}{.32\textwidth}
\includegraphics[width=\linewidth]{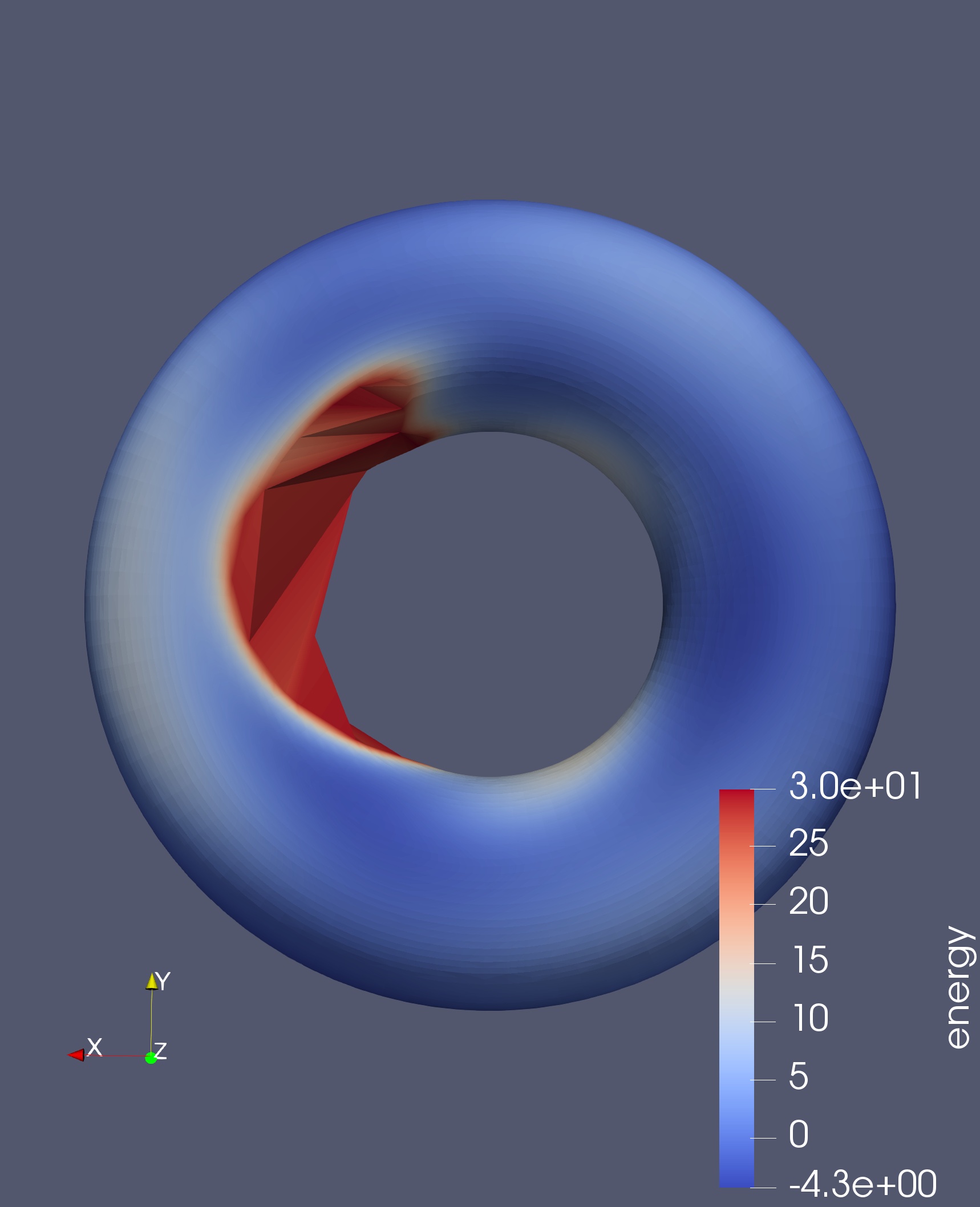}
\caption{}
\end{subfigure}
\begin{subfigure}{.32\textwidth}
\includegraphics[width=\textwidth]{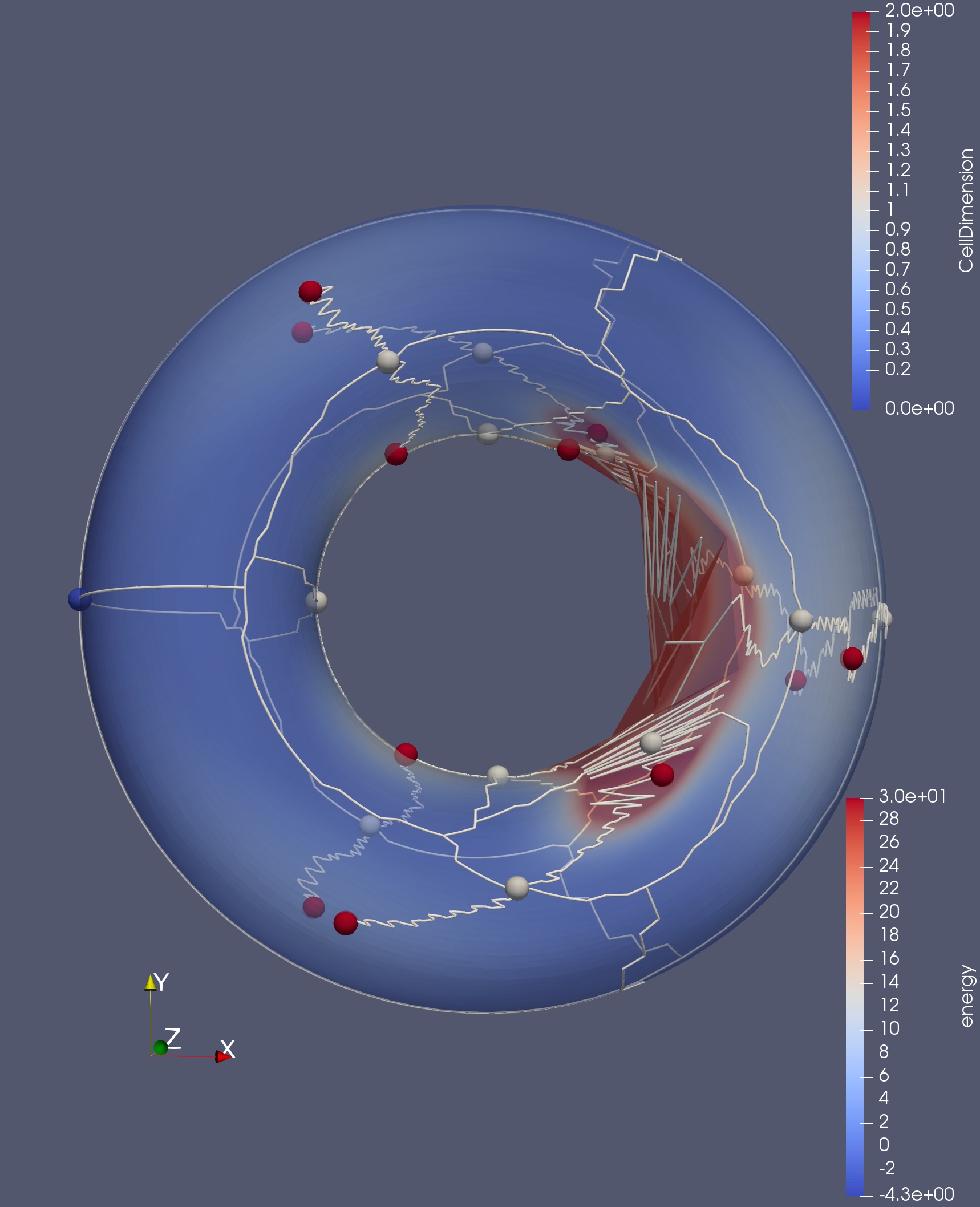}
\caption{}
\end{subfigure}
\caption{The analysis of the free energy surface for pentane, computed using the \emph{ttk} software. The first two figures (a) and (b) are showing the energy function from both sides of the surface of the torus, while (c) is showing the simplified Morse-Smale complex, together with the energy function.}
\label{pentane}
\end{figure}

Fluoropentane is a molecule that is very similar to pentane. It has the exact same underlying graph with a different vertex colouring. %\mariam{Something about only having one different atom, yes?}\lee{The thing we call `fluoropentane' is actually 1,1,1,2,2-pentafluoropentane. It's not a real molecule, in the sense I can't find evidence of anyone that has created it. But really we created it because we knew it would have an identical conformation space (like they do to alanine dipeptide), but slightly different energy landscapes. Hope that makes things clearer.}
Naturally, this implies that the conformation spaces of the two molecules are identical. This means we can have a direct comparison of the differences in the energy landscapes.
\begin{figure}
\centering
\begin{subfigure}{0.32\textwidth}
\includegraphics[width=\textwidth]{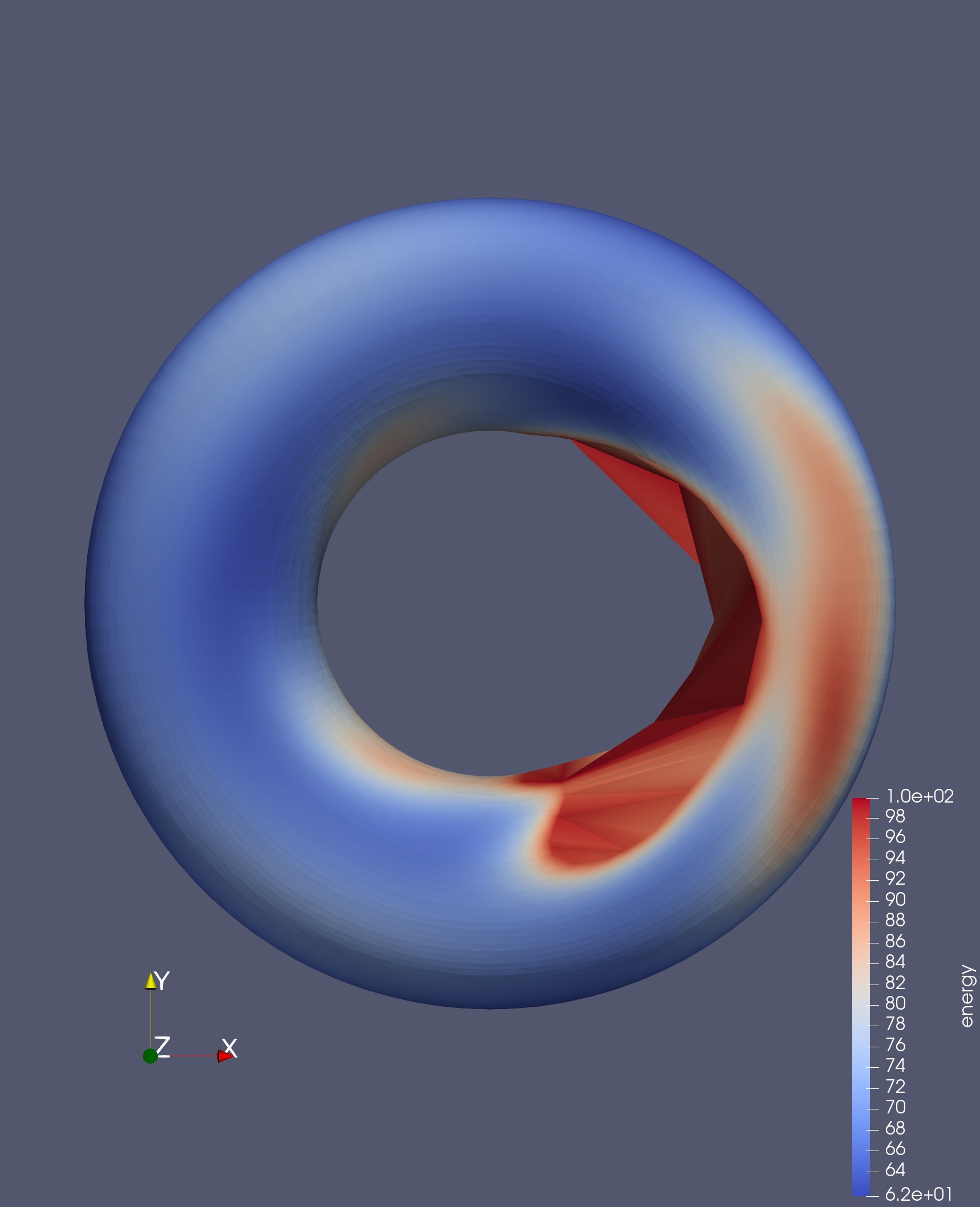}
\caption{}
\end{subfigure}
\begin{subfigure}{.32\textwidth}
\includegraphics[width=\textwidth]{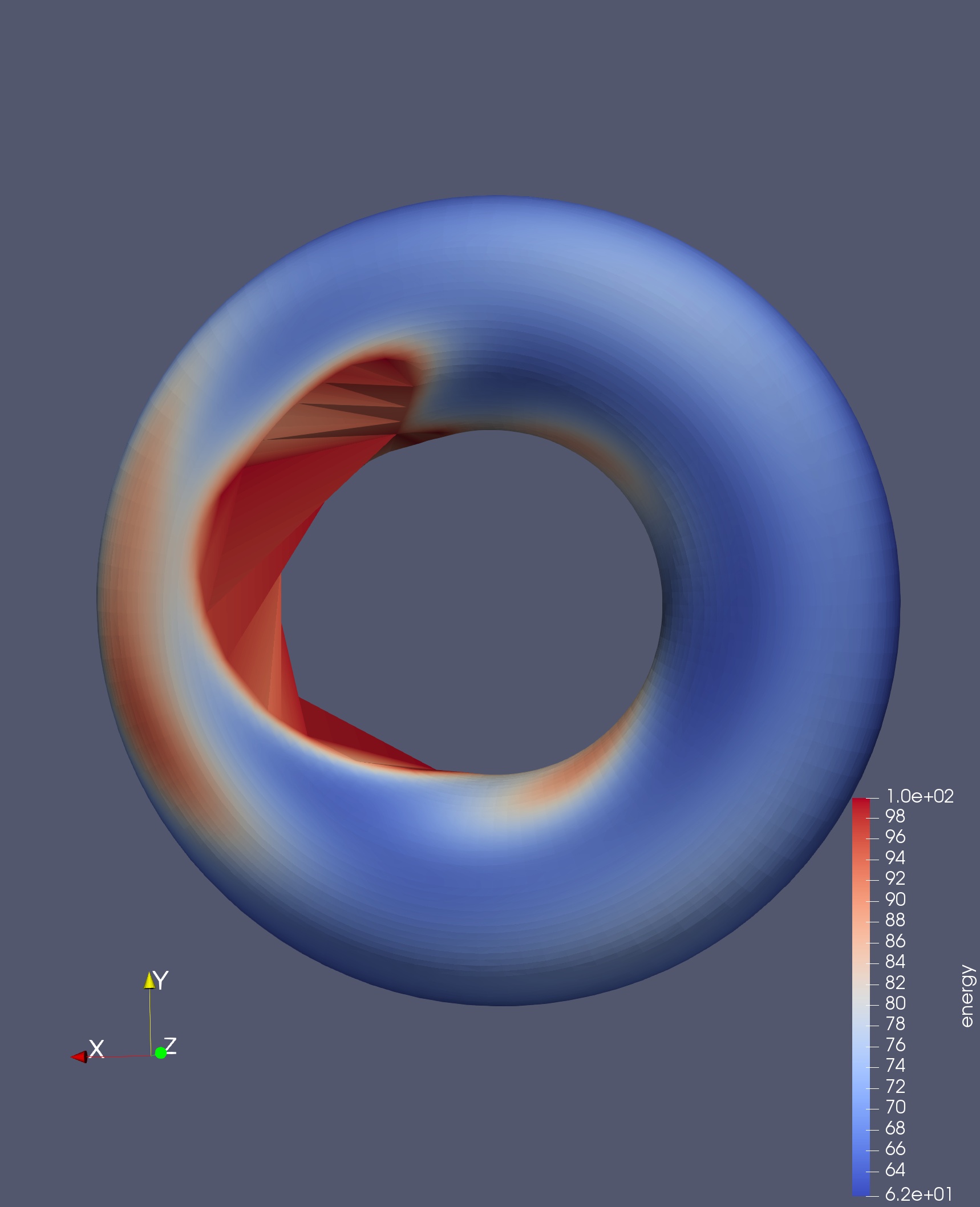}
\caption{}
\end{subfigure}
\begin{subfigure}{.32\textwidth}
\includegraphics[width=\textwidth]{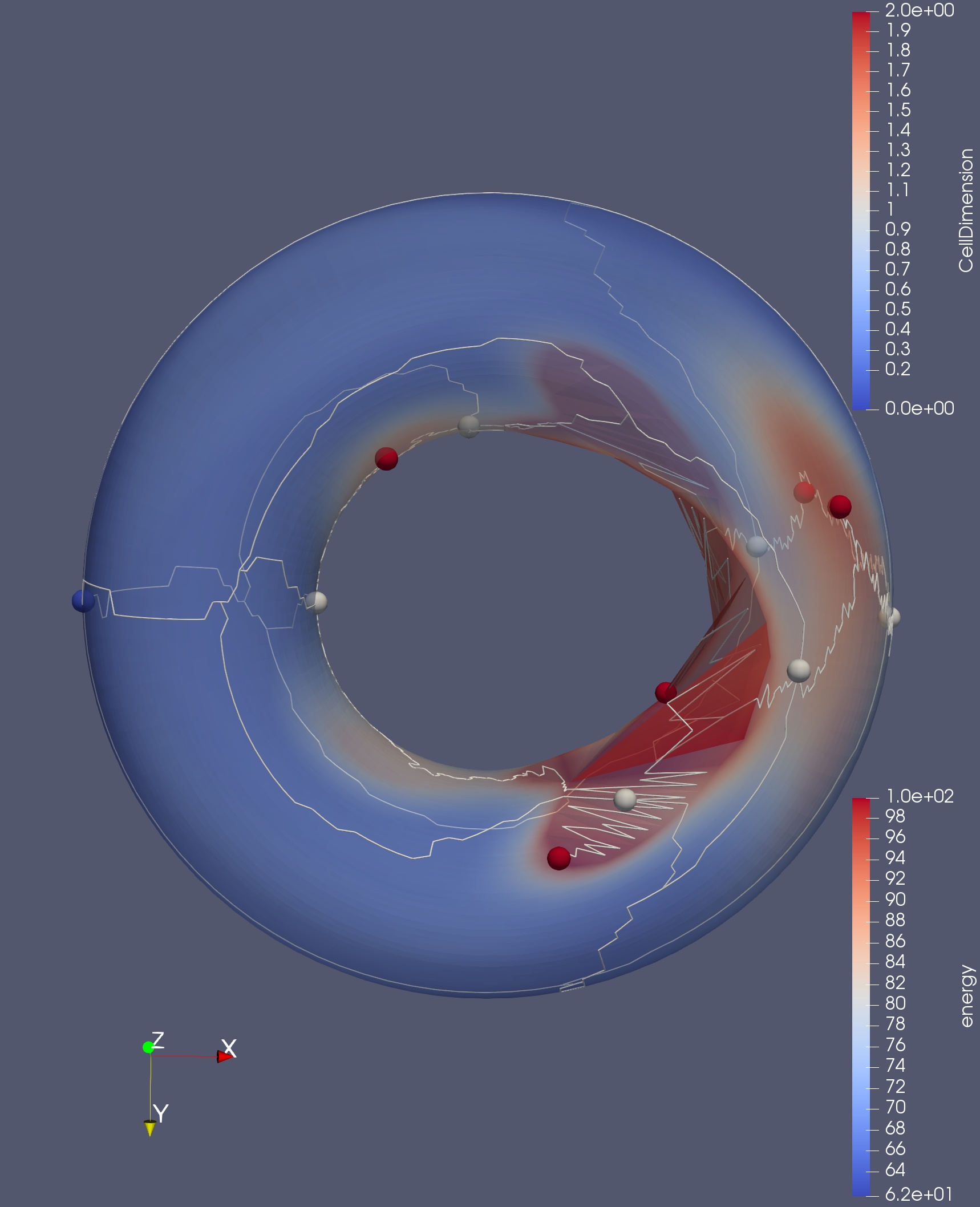}
\caption{}
\end{subfigure}
\caption{The analysis of the energy landscape for fluoropentane, computed using the \emph{ttk} software. The first two figures (a) and (b) are showing the energy function from both sides of the surface of the torus, while (c) is showing the simplified Morse-Smale complex, together with the energy function.}
\label{pentane}
\end{figure}

Finally, we are able to analyse the energy landscape for the Klein bottle component of cyclooctane. As the Klein bottle is not an orientable manifold, and also cannot be embedded into $\mathbb{R}^3$, we are unable to use \emph{ttk}. However, due to the link between Morse theory and persistent homology, we can use persistent homology to find the values of the extrema.

Firstly, we must create a simplicial complex. It can be seen from the persistence diagram of the Rips filtration of the Klein bottle component that a filtration value of $0.6$ leads to the correct topology. The MMFF94 energy of each simplex is then found, and the persistent homology of the energy function is calculated. This can be seen in Figure \ref{fig:CO_KB_energy}.

The infinite bars correspond to the topology of the space itself (i.e. the Klein bottle). Of interest, however, are the features that have a death value. Those correspond to the local critical values of the energy function. The zero-dimensional points are born at local minima and die at saddle points, while the one-dimensional points are born at saddle points and die at local maxima. We speculate that this methodology, of using persistence to find simplicial complexes and then critical values of complex energy landscapes could prove very fruitful. As an example, we propose a similarity measure in the chemical space making use of the topological properties of the energy landscapes. 

\begin{figure}[h]
\centering
\includegraphics[width=0.49\textwidth]{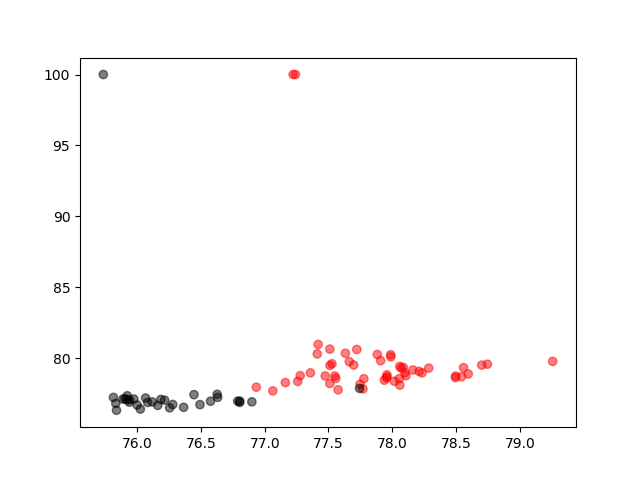}
\caption{Persistence of the sublevel sets of the MMFF94 energy function defined on the Klein bottle component. Calculated with coefficients in $\mathbb{Z}_2$.}
\label{fig:CO_KB_energy}
\end{figure}
\begin{figure}[h]
\centering
\includegraphics[width=0.49\textwidth]{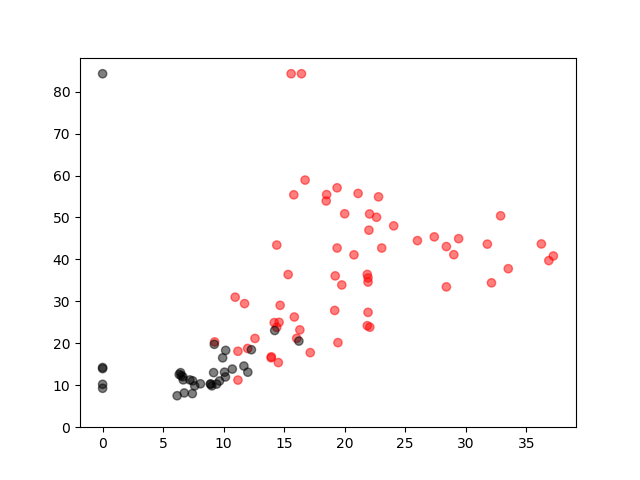}
\caption{Persistence of the superlevel sets of the MMFF94 energy function defined on the Klein bottle component. Calculated with coefficients in $\mathbb{Z}_2$.}
\label{fig:CO_KB_energy}
\end{figure}

The polynomial associated to a molecular graph $\mathcal M$ is defined as  
\begin{equation}
\mathcal P_{\mathcal M}(t):=\sum_{p\in CP} t^{\lambda(p)}
\end{equation}
where $CP$ is the set of critical points of the PES and $\lambda(p)$ is the index of the critical point $p$ (and $c(p)$ is the critical value at $p$.  Let $\mathcal P_1$ and $\mathcal P_2$ be two Morse polynomials associated to the molecules $\mathcal M_1$ and $\mathcal M_2$, respectively. We define the potential similarity measure $S:Chem\times Chem  \to \mathbb R$ in the space of molecular graphs as follows  
\begin{equation}
    S(\mathcal M_1,\mathcal M_2):=\int_{0}^1|\mathcal P_{\mathcal M_1}(t)-\mathcal P_{\mathcal M_2}(t)|^2dt
\end{equation}

\section{Conclusions}
We have developed a data-driven approach to understanding molecular conformational spaces and energy landscapes. We have used this method to demonstrate that conformational spaces of linear molecules match chemical intuition, and are namely products of circles caused by torsional flexibility. Further, bond stretching and bending do not change the homology groups of the conformational space as they lead to spaces related via a retraction. By using this method on different commonly used representations of conformational spaces (namely vector and metric spaces), we demonstrated that it is only the metric space representation that can consistently recreate the expected conformational space. Furthermore, we have demonstrated that the conformational space analysis still holds when molecular symmetry is taken into account.

Spaces of conformers of molecules play a fundamental role in Chemistry. Understanding these spaces might lead us to understand the relationship between the structure and the activity of biomolecules, drugs and other important classes of molecules. For many years chemists have modelled molecules using combinatorial objects, such as graphs. Associating both algebraic and geometric objects to molecules seems to provide new ways to study and to understand their chemical properties. In this paper we have shown that there exists a rich variety of algebraic, geometric and topological tools that can be used to model molecules and their conformational spaces. Symmetry groups of molecules are closely related to the topology of the conformational spaces. Methods developed in geometrical and topological data analysis seems to provide a good source of tools to analyse conformational spaces and functions defined on them such as (free) energy landscapes.

\appendix\section{Local PCA}\label{a:lpca}
\begin{comment}
\lee{I was reading through the appendices. I feel that if we want this to be understood by more chemists, it might be best to include some pictures?}\ingrid{good point I will add some pictures}
\end{comment}

Principal component analysis (PCA) is a mathematical method to transform a set of possible correlated variables into a set of linearly uncorrelated variables by means of orthogonal transformations. This method is used for dimensionality reduction of large data sets. The local version of PCA can be used to determine local dimensions of spaces using samples of points. In local PCA, a neighborhood of a point in a data set is transformed to a new coordinate system using orthogonal transformations. 

Let $S=\{A_i\}_{i=1}^N$ be a data set of points sampled from a topological space $X\subset R^d$, with $1\leq i\leq n$, $d\geq 1$. Given a data point $A_i$, we can reorder the data set $A_{i_1},A_{i_2},\dots,A_{i_{N}}$ with the order given by $\|A_i-A_{i_1}\|\leq \|A_i-A_{i_2}\|\leq\cdots\leq \|A_i-A_{i_N}\|$.  The $k$ nearest neighbours of $A_i$ are the first $k$-th data points of ordered set. Here we assume that $k\ll N$. The result of local $PCA$ on the set set $\mathcal K(A_i)$ of $k$ nearest neighbours of $A_i$. We try to estimate the local dimension at the point $A_i$ for the singular values of the mean shifted local data matrix $M(A_i)=[A_{i_1}-\mu A_{i_1}-\mu\cdots A_{i_N}-\mu]$, where $\mu=\frac{1}{k}\sum_{j=1}^kA_{i_j}$. Let $s_i^1\geq s_i^2\geq\cdots\geq s_i^k$ be the singular values of $M(A_i)$. If $X$ is a manifold of dimension $p$, the local dimension at any point $A_i\in X$ is well defined and indeed it is equal to $p$. In this case the local dimension of $X$ at the point $A_i$ is estimated as the number of non-vanishing singular values. When working with large data sets embbeded in high dimensional euclidean spaces it is common to estimate the local dimension as the number of singular values with the highest percentage of the variability of the data. For that it is needed to choose a threshold, that is, a number $\alpha$ between 0 and 1, and the dimension is estimated as the smallest integer $ld$ such that
$$\frac{\sum_{j=1}^{ld}s_j^2}{\sum_{j=1}^ms_j^2}>\alpha$$
where $m$ is the minimum $m=min\{k,d\}$.
\section{Principal bundles and orbifolds}\label{pgborbi}
We give a brief introduction to fibre bundles and orbifolds. Standard references for the theory of fibre bundles and orbifolds are \cite{Hsm} and \cite{ALR}. A (left) action of a topological group $G$ on a topological space $X$ is a map $G\times X\to X$, $(g,x)\mapsto gx$ such that $1x=x$ and $g(hx)=(gh)x$ for all $x\in X$ and for all $g,h\in G$.  Given $x\in X$ the \textit{stabilizer of x} is the group of elements $g\in G$ such that $gx=x$. The action of $G$ on $X$ is called free if for all $x\in X$, $G_x=\{1\}$; the action is almost free if the groups $G_x$ are finite. The \textit{orbit of x} is the set $\mathcal O_x=\{gx\in X\mid g\in G\}$. An action is effective if $gx=x$ for all $x$ then $g=1$. The action of $G$ induces a partition of $X$ into disjoint subsets called orbits. The \textit{orbit space}, denoted $X/G$, is the set of all orbits of $X$. The orbit space is endowed with the quotient topology induced by the quotient map $X\to X/G$.  

Let $B$, $E$ be topological spaces. A \emph{bundle} is a triple $(B,E,\pi)$, where $\pi:E\rightarrow B$ is a map. The space $B$ is the \emph{base space}, the space $E$ is the \emph{total space}, and the map $\pi$ is the \emph{projection} of the bundle. For each $b\in B$, the space $\pi^{-1}(b)$ is called the \emph{fibre} of the bundle over $b$.
A \emph{cross section} of a bundle $(E,B,\pi)$ is a map $s:B\rightarrow E$ such that $p\circ s=1_B$.

A \emph{principal $G$-bundle}, denoted is a bundle $(B,E,\pi)$ and an action $E\times G\to E$, $(x,g)\mapsto xg$, such that the following hold:
\begin{enumerate}[(1)]
\item the map $E\times G\to E\times E$ is given by $$(x,g)\mapsto (x,xg)\qquad x\in E,\;\; g\in G$$ is a homeomorphism onto its image;
\item $B=E/G$ and the projection $p$ is the quotient map;
\item for all $b\in B$ there exists an open neighborhood $V$ together with a homeomorphism $\phi:V\times G\to \pi^{-1}(V)$ such that the diagram
\begin{equation*}
\xymatrix{
V_{}\times G\ar[rr]^{\phi_{}}\ar[rd]_-{p_1}&&\pi^{-1}(V_{})\ar[ld]^-{\pi}\\
&V_{}
}
\end{equation*}
commutes, and for all $x\in \pi^{-1}(V)$ and $g\in G$, $\phi^{-1}(xg)=\phi^{-1}(x)g$, where the action on $V\times G$ is given by $(x,g)g'=(x,gg')$. 
\end{enumerate}

We say that the sequence
$$G\to E\xrightarrow{\pi} B$$ is a \emph{principal $G$-bundle} over $B$.

Let  $X$ be a topological space. An $n$-dimensional orbifold chart on  $X$ is a  tuple $(\tilde U,G,\varphi)$, where $\tilde U\subset \mathbb R^n$ is open, $G$ is a finite group of smooth automorphisms of $\tilde U$ and $\varphi:\tilde U\to X$ is a  map such that the diagram
\begin{equation}
\xymatrix{
\tilde U\ar[r]^-{\varphi}\ar[d]_-{q}& X\\
\tilde U/G\ar[ru]_-{\tilde\varphi},
}
\end{equation}
commutes, where $q$ is the quotient map and $\tilde\varphi $ is a homeomorphism onto its image. Two orbifold charts $(\tilde U_1,G_1,\varphi_1)$ and $(\tilde U_2,G_2,\varphi_2)$ on $X$ are compatible if for every $x\in U_1\cap U_2\subset X$ there exists a neighbourhood $W$ of $x$ and an orbifold chart  $(\tilde W,H, \psi)$, with $\psi(\tilde W)=W$, such that there are smooth embeddings $\lambda_i:\tilde W\to U_i$ and $\psi=\varphi_i\circ\lambda_i$, for $i=1,2$. An  \textit{orbifold atlas} on $X$ is a collection $\mathcal U=\{U_\alpha,G_\alpha,p_\alpha\}_{\alpha\in I}$ of compatible $n$-dimensional charts that cover $X$.  An $n$-dimensional \textit{orbifold} $\mathcal X$ is a paracompact Hausdorff space $X$ with an orbifold atlas of $n$-dimensional charts on $X$. We additionally will require for each chart $(\tilde U,G,\varphi)$, the action of $G$ on $\tilde U$ to be effective. Since every smooth action is locally smooth, any orbifold has an atlas of the form $(\mathbb R^n,G,\varphi)$ where $G\subset O(n)$ acts on $\mathbb R^n$ via orthogonal representation. CITE ADEM
Given a point $x\in X$ and a local chart $(\tilde U,G,\varphi)$ around $x$, the \textit{local group} at $x=\varphi(y)$, denoted $G_x$, is the stabiliser of $y$. The set of singularities $\Sigma(\mathcal X)$ or singular locus of an orbifold $\mathcal X=(X,\mathcal U)$ is defined by $$\Sigma(\mathcal X)=\{x\in X\mid G_x\neq\{1\}\}.$$

The set $\Sigma(\mathcal X)$ may have several connected components and if $\Sigma \mathcal X=\emptyset$ then  $\mathcal X$ is a manifold. Orbifolds arise in a natural way via actions of groups of transformations on manifolds. A quotient orbifold $\mathcal X$ is the one obtained as the quotient compact Lie group acting smoothly, effectively and almost freely on a smooth manifold $M$.

\section{Persistent homology}
The persistence modules most commonly used in topological data analysis arise from filtered simplicial complexes, whose combinatorial nature is very suitable for computations.

A simplicial complex $K$ with vertex set $S$ is a family of nonempty, finite subsets of $S$. Subsets of $S$ of $p+1$ elements are called $p$-simplices. A $p$-simplex is represented as a list of its vertices $[v_0,\cdots ,v_p]$. In a simplicial complex $K$, one requires that all elements $v$ of $S$ form $0$-simplices $[v]$ in $K$, and if $\sigma\in K$ and $\emptyset\neq\tau\subset \sigma$, then $\tau\in K$. We usually consider the case when $S$ is finite. A simplicial complex $K$ has the associated space $| K | $, called the geometric realisation, which can be regarded as a triangulated polyhedron in an appropriate Euclidean space.
\

We need a metric space to make use of persistent homology. Let $S$ be the set of points sampled from the conformational space. For any two points $z,w\in S$, the distance $d_S(z,w)$ is given by the Euclidean distance between $z$ and $w$.

In general, to a metric space $(X,d)$ equipped with a function $f:X\rightarrow \mathbb{R}$ one can assign a persistence module as follows: First, we define the \emph{sub-level set of $X$ with respect to $\alpha\in \mathbb{R}$} by
$$X^\alpha:=\{x\in X \ | \ f(x)\leq \alpha \}.$$
If  $\alpha_1 < \alpha_2$ we have the inclusion $X^{\alpha_1}\subseteq X^{\alpha_2}$.

For $p=1,2,\cdots$, the $p^{th}$ homology gives information about the $p$-dimensional \emph{holes}: For $p=0$, this refers to connected components, for $p=1$, it refers to loops, for $p=2$, it is cavities or voids, etc.
The number $\alpha$ is called a $p$-critical value of $f$ if the number of $p$-dimensional holes of $X^{\alpha-\epsilon}$ and $X^{\alpha+\epsilon}$ changes for all small $\epsilon >0$. For $p=0$, by $0$-dimensional holes we mean connected components, for $p=1$, we mean standard holes, for $p=2$ we are talking about voids or cavities, etc.

\

For the set $S$ of points sampled from a conformational space this works as follows. We define $\Delta^S$ to be the simplex with the elements of $S$ as its vertices. Then we define a function $f:\Delta^S\rightarrow \mathbb{R}$, which assigns to each point $s\in S$ the value $0$, and to each higher simplex in $\Delta^S$ the maximal pairwise distance between the vertices it contains.

We take all $p$-critical values $\alpha_1<\alpha_2<\cdots <\alpha_n$. Then the sub-level sets connected by natural inclusion maps give rise to a filtration:
$$X^{\alpha_1}\subseteq X^{\alpha_2}\subseteq \cdots\subseteq X^{\alpha_n}=X.$$

\

The zeroth persistence diagram $Dg_0(f)$ captures the connected components that were born or died passing through a critical point. It consists of a set of points in the plane $\{(a,b)\in \mathbb{R}^2 \ | \ a<b\}$. Each point can occur more than once. The coordinates $a$ and $b$ of a point indicate the birth and death times of the connected components. The multiplicity of the point indicates the number of connected components that were born at time $a$ and died at time $b$. The first persistence diagram $Dg_1(f)$ does the same for loops instead of connected components. Similarly for higher dimensional persistence diagrams.

\section{Morse theory}

Morse theory addresses the relationship between the properties of a manifold and the properties of its real-valued functions. In particular, it aims to relate the topological properties of the $n$-dimensional manifold $M$ with analytical properties of smooth functions $f:M\rightarrow \mathbb{R}$.

For an in-depth discussion of Morse theory, we refer the reader to, e.g. \cite{Milnor}. Here we only recall the concepts that we are interested in when generalising to the discrete setting.

Morse theory captures the relationship between a function $f$ and the manifold $M$ by relying on the gradient $\nabla_x f(x)=df /dx (x)$ and its flow. The gradient defines a preferential direction at every point (the direction of steepest ascent) except where it vanishes (i.e. where $\nabla_xf=0$). Those particular points are called critical points and can be classified according to the sign of the Hessian matrix, the $n\times n$ matrix of the second derivatives $H_f (x) = d^2f/dx_idx_j (x)$.

Let $M$ be a (smooth) manifold, and let $f:M\rightarrow \mathbb{R}$ be a smooth function. A point $m\in M$ is called a critical point of $f$ if $\nabla_f(m)= 0$. The index of $m$ is defined to be the number of negative eigenvalues of the Hessian matrix $H_f(m)$.

Note that according to this definition, the index of a critical point is defined by the sign of the eigenvalues of the Hessian, which must therefore be non-null. This condition is essential to Morse theory: a function $f$ which obeys Morse theory must necessarily satisfy this constraint. Conversely, such functions are called Morse functions:

Let $m$ be a critical point of $f:M\rightarrow \mathbb{R}$. We say that $m$ is a non-degenerate critical point if and only if the nullity of $H_f(m)$, i.e. the dimension of the $0$-eigenspace of $H_f(m)$, is zero. We say that $f$ is a Morse function if and only if every critical point of $f$ is non-degenerate.

On an $n$-dimensional manifold $M$, the index of any local maximum point is $n$, and the index of any local minimum point is $0$.

At the location of any non-critical point, one can define specific lines, the integral lines, by following the preferred direction of the gradient flow.

Integral lines represent the flow along the gradient between critical points. They have the following properties:
\begin{itemize}
\item Two integral lines are either disjoint or the same.
\item Integral lines cover all of $M$.
\item The origin and destination of an integral line are critical points of $f$ (except at boundary).
\item In a gradient vector field, integral lines are monotonic, i.e. the origin is distinct from the destination.
\end{itemize}

Let $m$ be a critical point of $f:M\rightarrow \mathbb{R}$. The ascending manifold of $m$ is the set of points belonging to integral lines whose origin is $m$. The descending manifold of $m$ is the set of points belonging to integral lines whose destination is $m$. Note that ascending and descending manifolds are also referred to as unstable and stable manifolds.

For a Morse function $f:M\rightarrow \mathbb{R}$, the complex of the descending manifold of $f$ is called the Morse complex.

A Morse function $f$ is \emph{Morse-Smale} if the ascending and descending manifolds intersect only transversally. Intuitively, an intersection of two manifolds is transversal when they are not 'parallel' at their intersection. A pair of critical points that are the origin and destination of an integral line in the Morse-Smale function cannot have the same index. Furthermore, the index of the critical point at the origin is less than the index of the critical point at the destination.

Given a Morse-Smale function $f$, the Morse-Smale complex of $f$ is the complex formed by the intersection of the Morse complex of $f$ and the Morse complex of $-f$.

For a $2$-dimensional manifold $M$, any Morse-Smale function will give rise to a Morse-Smale complex with the following combinatorial properties:
\begin{itemize}
\item The nodes of the complex are exactly the critical points of the Morse function.
\item The saddle points have exactly four arcs incident to them.
\item All regions are quadrangles.
\item The boundary of a region alternates between saddle points and extrema.
%\item The number of minima and maxima is even.
\end{itemize}

\subsection{Discrete Morse Theory}
Discrete Morse Theory was defined by Robin Forman in 1995 \cite{Forman1,Forman2}. It is, as the name suggests, a combinatorial analogue of classical Morse Theory, providing combinatorial equivalents of several core concepts of classical Morse theory, such as Morse functions, gradient vector fields, critical points, and a cancellation theorem for the elimination of pairs of critical points from a vector field. The discrete theory maintains the intuition of its classical counterpart while enabling simple and explicit constructions that are considerably more involved in the smooth setting.

Let $M$ be any finite simplicial complex, $K$ the set of simplices of $M$, and $K_p$ the simplices of dimension $p$. A discrete Morse function on $M$ will actually be a function on $K$. That is, we assign a single real number to each simplex in $M$. Write $\sigma^{(p)}$ if $\sigma$ has dimension $p$, and $\tau>\sigma$ if $\sigma$ lies in the boundary of $\tau$. We say a function $f:K\rightarrow \mathbb{R}$ is a discrete Morse function if it satisfies the following two properties:
\begin{enumerate}
\item For any $\tau\in K$, the number of $\sigma<\tau$ for which $f(\tau)\leq f(\sigma)$ is at most one.
\item For any $\sigma\in K$, the number of $\sigma<\tau$ for which $f(\tau)\leq f(\sigma)$ is at most one.
\end{enumerate}

We say $\sigma^{(p)}$ is critical (with index $p$) if 
\begin{enumerate}
\item there is no $\tau < \sigma^{(p)}$ with $f(\tau)\geq f(\sigma)$,
\item there is no $\tau > \sigma^{(p)}$ with $f(\tau)\leq f(\sigma)$.
\end{enumerate}

So in the condition of a Morse function above, the conditions can be broken for one simplex each, but if the properties are satisfied for all face simplices or all coface simplices, we have a critical simplex. This definition provides a discrete analogue of the smooth notion of a critical point of index $p$. Note that in the discrete case, we have a simplex of dimension $p$ rather than a point.

A discrete Morse function $f$ over $K$ defines a discrete gradient vector field by coupling simplices in gradient arrows (also called gradient pairs):
\begin{enumerate}
\item if a simplex $\sigma_p$ has exactly one lower valued coface $\sigma_{p+1}$, then [$\sigma_p, \sigma_{p+1}$] form a gradient arrow,
\item if a simplex $\sigma_p$ has exactly one higher valued facet $\sigma_{p - 1}$, then [$\sigma_{p-1}, \sigma_p$] form a gradient arrow,
\item if a simplex $\sigma_p$ is critical, it does not belong to a gradient arrow.
\end{enumerate}

So a discrete vector field $V$ on a simplicial complex $K$ is a set of pairs of simplices $(\sigma,\tau )$, with $\sigma$ a facet of $\tau$, such that each simplex of $K$ is contained in at most one pair of $V$. A simplex $\sigma\in K$ is critical with respect to $V$ if $\sigma$ is not contained in any pair of $V$. The dimension of a critical simplex is also called its index.

A pair $(\sigma,\tau)$ in a discrete vector field $V$ can be visualised by an arrow from $\sigma$ to $\tau$. We consider an important subclass of vector fields in which the arrows do not form closed paths. This can be made precise using the concept of $V$-paths.

Given a discrete vector field $V$, a $V$-path is a sequence $\tau_0,\sigma_1,\tau_1,\cdots,\sigma_l,\tau_l,\sigma_{l+1}$, where $(\sigma_i,\tau_i) \in V$ for every $i = 1,\cdots,l$, and each $\sigma_{i+1}$ is a facet of $\tau_i$ for each $i = 0,\cdots,l$. If $l = 0$, the $V$-path is trivial. This $V$-path is cyclic if $l > 0$ and $(\sigma_{l+1},\tau_0)\in V$; otherwise, it is acyclic, in which case we call this $V$-path a gradient path. We say that a gradient path is a vertex-edge gradient path if dimension $(\sigma_i) = 0$, implying that dimension $(\tau_i) = 1$. Similarly, it is an edge-triangle gradient path if dimension $(\sigma_i) = 1$.

A discrete vector field $V$ becomes a discrete gradient vector field if there are no cyclic $V$-paths induced by $V$.

Intuitively, a gradient path $\tau_0,\sigma_1,\tau_1,\cdots,\sigma_l,\tau_l,\sigma_{l+1}$ is the analogue of an integral line in the smooth setting. Different from the smooth setting, a maximal gradient path may not start or end at critical simplices. However, those that do (i.e, when $\tau_0$ and $\sigma_{k+1}$ are critical simplices) are analogous to maximal integral lines in the smooth setting which start and end at critical points, and for convenience one can think of critical $k$-simplices in the discrete Morse setting as index-$k$ critical points in the smooth setting.

For example, for a function on a two-dimensional simplicial complex, critical $0$-, $1$- and $2$-simplices in the discrete Morse setting correspond to minima, saddles and maxima in the smooth setting, respectively.

For a critical edge $e$, we define its stable (or descending) manifold to be the union of edge-triangle gradient paths that ends at $e$. Its unstable (or ascending) manifold is defined to be the union of vertex-edge gradient paths that begins with $e$.

\subsection{Morse theory and Persistent homology for topological simplification}
From the perspective of persistence for a Morse real-valued function, the infinite persistence barcodes determine the Betti numbers of the underlying manifold and the finite barcodes give information about the multitude of visible trajectories between the critical points.

The ends of any barcode for $f$ are among the critical values of $f$, the values $t$ for which the homology of the level of $f$ at $t$ differs from the homology of levels at values in an arbitrary neighborhood of $t$.

Given a scalar function $f$, such as the energy function, defined on a topological space, we can talk about topology-preserving simplification of the scalar function \cite{Bauer2012}. By this we mean that topological features with persistence greater than a specified level are guaranteed to remain.

In the discrete case, our topological space gets replaced by a simplicial complex, while the scalar function is required to be a piecewise linear function defined on the simplices of the simplicial complex, a discrete Morse function as defined above. The topological features are then the pairings of critical cells, which correspond to the persistent bars in the persistence barcode.

There exist special configurations where two critical points are linked by two or more different paths. Those particular configurations cannot be cancelled, as applying a discrete gradient reversal would result in the formation of a $V$-path that loops onto itself. However, for cancellable pairs, we specify a threshold, and cancel all critical pairs that correspond to a persistence bar below that threshold.

This is, in essence, a way of denoising the scalar function, while having a tight control of the size of noise.

%\begin{backmatter}
\section*{Acknowledgements}
This research was supported by the EPSRC grants EP/N014189/1 Joining the dots: from data to insight, and EP/L015722/1 Centre for Doctoral Training in Theory and Modelling in Chemical Sciences. LS thanks Khaled Abdel-Maksoud for assistance in running metadynamics simulations.

\bibliographystyle{amsplain}
\bibliography{Conf_spaces.bib}
%\end{backmatter}
\end{document}